\documentclass{llncs}
\usepackage[dvips]{graphicx}
\usepackage{pstricks}
\usepackage{amssymb}
\usepackage{amsmath}
\usepackage{verbatim}
\usepackage{psfrag}

\spnewtheorem{mylemma}[theorem]{Lemma}{\bf}{\it}
\spnewtheorem{myprop}[theorem]{Proposition}{\bf}{\it}
\newcommand{\outputstyle}[1]{\textbf{#1}}
\newcommand{\leftind}[1]{{#1}_{\mathrm{left}}}
\newcommand{\rightind}[1]{{#1}_{\mathrm{right}}}
\begin{document}

\mainmatter
\title{Cleaning Interval Graphs}
\author{D\'aniel Marx\inst{1} \and Ildik\'o Schlotter\inst{2}}

\institute{
Tel Aviv University, Israel
\and 
Budapest University of Technology and Economics, Hungary \\
\email{\{dmarx,ildi\}@cs.bme.hu}\\
}

\maketitle

\begin{abstract} We investigate a special case of the \textsc{Induced Subgraph Isomorphism} problem,
where both input graphs are interval graphs. We show the NP-hardness of this problem,
and we prove fixed-parameter tractability of the problem with non-standard parameterization,
where the parameter is the difference $|V(G)|-|V(H)|$, with $G$ and $H$ being the
larger and the smaller input graph, respectively.
Intuitively, we can interpret this problem as ``cleaning'' the graph $G$, regarded as a pattern containing
extra vertices indicating errors, in order to obtain the graph $H$ representing the original pattern.
We also prove W[1]-hardness for the standard parameterization where the parameter is $|V(H)|$.
\end{abstract}

\section{Introduction}
Problems related to graph isomorphisms play a significant role in algorithmic graph theory.
The \textsc{Induced Subgraph Isomorphism} problem is one of the basic problems of this area:
given two graphs $H$ and $G$, find an induced subgraph of $G$ isomorphic to $H$, if this is possible.
In this general form, \textsc{Induced Subgraph Isomorphism} is NP-hard, since it contains several
well-known NP-hard problems, such as \textsc{Independent Set} or \textsc{Induced Path}.
%%% Modification.
As shown in Sect.\ref{sect_hard}, the special case of  \textsc{Induced Subgraph Isomorphism}
when both input graphs are interval graphs is NP-hard as well. 

As \textsc{Induced Subgraph Isomorphism} has a wide range of important applications, polynomial time
algorithms have been given for numerous special cases, such as the case when both input graphs
are trees \cite{matula-subtree} or 2-connected outerplanar graphs \cite{lingas-outerplanar}.
However, \textsc{Induced Subgraph Isomorphism} remains NP-hard even if $H$ is a forest and $G$ is a tree,
or if $H$ is a path and $G$ is a cubic planar graph \cite{garey-johnson-book}.
In many fields where researchers face hard problems, parameterized complexity theory
(see e.g. \cite{downey-fellows-book}, \cite{niedermeier-book} or \cite{grohe-flum-book})
has proved to be successful in the analysis and design of algorithms that have a tractable running time
 in many applications. In parameterized complexity, a parameter $k$ is introduced
besides the input $I$ of the problem, and the aim is to find an algorithm with running time $O(f(k)|I|^c)$
where $f$ is an arbitrary function and $c$ is a constant, independent of $k$. A parameterized
problem is \emph{fixed-parameter tractable} (FPT), if it admits such an algorithm.

Note that \textsc{Induced Subgraph Isomorphism} is trivially solvable in 
time $O(|V(G)|^{|V(H)|}|E(H)|)$ on input graphs $H$ and $G$.
As $H$ is typically much smaller than $G$ in many applications related to pattern matching,
the usual parameterization of \textsc{Induced Subgraph Isomorphism} is to define the parameter to be $|V(H)|$.
FPT algorithms are known if $G$ is planar \cite{eppstein-subgraphs}, 
has bounded degree \cite{cai-random-separation},
or if $H$ is a log-bounded fragmentation graph and $G$ has bounded treewidth 
\cite{hajiaghayi-nishimura-logbounded-fragmentation}.
%%% Modification.
In Sect. \ref{sect_hard}, we show that the case when both input graphs are interval graphs
is W[1]-hard with this parameterization. 

%%% Modification.
Our main objective is to consider another parameterization of \textsc{Induced Subgraph Isomorphism},
where the parameter is the difference $|V(G)|-|V(H)|$.
Considering the presence of extra vertices as some kind of error or noise, the problem of finding
the original graph $H$ in the ``dirty'' graph $G$ containing errors is clearly meaningful.
In other words, the task is to ``clean'' the graph $G$ containing errors
in order to obtain $H$. For two graph classes $\mathcal{H}$ and $\mathcal{G}$ we define the
\textsc{Cleaning}($\mathcal{H}, \mathcal{G}$) problem: given a pair of graphs $(H,G)$ with
$H \in \mathcal{H}$ and $G \in \mathcal{G}$,
find a set of vertices $S$ in $G$ such that $G-S$ is isomorphic to $H$.
The parameter associated with the input $(H,G)$ is $|V(G)|-|V(H)|$.
For the case when $\mathcal{G}$ or $\mathcal{H}$ is the class of all graphs, we will use the notation
\textsc{Cleaning}($\mathcal{H},-$) or \textsc{Cleaning}($-,\mathcal{G}$), respectively.

In the special case when the parameter is 0, the problem is equivalent to the \textsc{Graph Isomorphism} problem,
so we cannot hope to give an FPT algorithm for the general problem \textsc{Cleaning}($-,-$).
Several special cases have already been studied.
FPT algorithms were given for the problems \textsc{Cleaning}(\textsl{Tree},$-$) 
\cite{marx-schlotter-dam-cleaning},
\textsc{Cleaning}(\textsl{3-Connected-Planar, Planar}) \cite{marx-schlotter-dam-cleaning} and
\textsc{Cleaning}(\textsl{Grid},$-$) \cite{diaz-thilikos-grid},
where \textsl{Tree}, \textsl{Planar}, \textsl{3-Connected-Planar} and \textsl{Grid} denote
the class of trees, planar graphs, 3-connected planar graphs, and rectangular grids, respectively.
Without parameterization, all of these problems are NP-hard.

Here we consider the special case where the input graphs are from \textsl{Interval}, denoting
the class of interval graphs.
In Sect. \ref{sect_algo}, we present an FPT algorithm for
\textsc{Cleaning}(\textsl{Interval, Interval}).

%%%%%%%%%%%%%%%%%%%%%%%%%%%%%%%%% Here begins the preliminaries %%%%%%%%%%%%%%%%%%%%%%%%%%%%%%%%%%%%%

\section{Notation and preliminaries}

We denote $\{1, \dots, n\}$ by $[n]$.
We denote the neighbors of a vertex $x \in V(G)$ in $G$ by $N_G(x)$. For some $X \subseteq V(G)$, let
$N_G(X)$ denote those vertices in $V(G) \setminus X$ that are adjacent to a vertex in $X$ in $G$, and let
$N_G[X]=N_G(X) \cup X$. If $X \subseteq V(G)$ then $G-X$ is obtained from $G$ by deleting $X$,
and $G[X]=G-(V(G) \setminus X)$. For some vertex $x$, sometimes we will use only $x$ instead of $\{x\}$, 
but this will not cause any confusion.
We say that two subsets of $V(G)$ are \emph{independent} in $G$, if no edge of $G$ runs between them.
Otherwise, they are \emph{neighboring}.

Let $G$ be an interval graph, meaning that $G$ can be regarded as the intersection graph
of a set of intervals. Formally, an interval representation of $G$ is
a set $\{ I_i  \mid i \in [n] \}$ of intervals, where $I_i$ and $I_j$ intersect each other if
and only if $v_i$ and $v_j$ are adjacent.
We say that two intervals \emph{properly intersect}, if they intersect, but none of them contains the other. 

Let $\mathcal{C}(G)$ be the set of all maximal cliques in $G$, and let
$\mathcal{C}(v) = \{ C \mid v \in C, C \in \mathcal{C}(G)\}$ for some $v \in V(G)$.
It is known that a graph is an interval graph if and only if its maximal cliques can
be \emph{ordered consecutively}, i.e. there is an ordering of $\mathcal{C}(G)$ such that
the cliques in $\mathcal{C}(v)$ form a consecutive subsequence \cite{gilmore-hoffman}.
Note that any interval representation gives rise to a natural ordering of $\mathcal{C}(G)$, 
which is always a consecutive ordering.
The set of all consecutive orderings of $\mathcal{C}(G)$ are usually represented by PQ-trees, 
a data structure introduced by Booth and Lueker \cite{booth-lueker-consecutive}.

A \emph{PQ-tree} of $G$ is a rooted tree $T$ with ordered edges with the following properties:
every non-leaf node is either a Q-node or a P-node, each P-node has at least 2 children,
each Q-node has at least 3 children, and the leaves of $T$ are bijectively associated with
the elements of $\mathcal{C}(G)$.
The \emph{frontier} $F(T)$ of the PQ-tree $T$ is the permutation of $\mathcal{C}(G)$
that is obtained by ordering the cliques associated with the leaves of $T$
simply from left to right.
Two PQ-trees $T_1$ and $T_2$ are \emph{equivalent}, if one can be obtained from the other by
applying a sequence of the following transformations: permuting the children
of a P-node arbitrarily, or reversing the children of a Q-node.
The consecutive orderings of the maximal cliques of a graph can be represented by a PQ-tree
in the following sense: for each interval graph $G$ there exists a PQ-tree $T$, such that
$\{ F(T') \mid  T'$ is a PQ-tree equivalent to $T \}$ yields the set
of all consecutive orderings of $\mathcal{C}(G)$. Such a PQ-tree \emph{represents} $G$.
For any interval graph $G$ a PQ-tree representing it can be constructed in linear time 
\cite{booth-lueker-consecutive}.

This property of PQ-trees can be used in the recognition of interval graphs.
However, to examine isomorphism of interval graphs, the information stored in a PQ-tree is
not sufficient. For this purpose, a new data structure, the labeled PQ-tree has been defined
\cite{booth-lueker-interval-isomorphism,booth-colbourn}.
For a PQ-tree $T$ and some node $s \in V(T)$, let $T_s$ denote the subtree of $T$ rooted at $s$. 
For each vertex $v$ in $G$, let the \emph{characteristic node} $R(v)$ of $v$ in a PQ-tree $T$ 
representing $G$ be the deepest node $s$ in $T$ such that the frontier of $T_s$ contains $\mathcal{C}(v)$.
For a node $s \in V(T)$, we will also write $R^{-1}(s) = \{ x \in V(G) \mid R(x)=s \}$,
and if $T'$ is a subtree of $T$, then $R^{-1}(T')= \{ x \in V(G) \mid R(x) \in V(T')\}$.
Observe that if $R(v)$ is a P-node, then every clique in the frontier of $T_{R(v)}$ contains $v$.
It is also true that if $R(v)$ is a Q-node with children $x_1, x_2, \dots, x_m$, then 
those children of $R(v)$ whose frontier contains $v$ form a consecutive subseries of $x_1, \dots x_m$.
Formally, there must exist two indices $i < j$ such that 
$\mathcal{C}(v) = \{C \mid C \in F(T_{x_h})$ for some $i \leq h \leq j\}$.

A \emph{labeled PQ-tree} of $G$ is a labeled version of a PQ-tree $T$ of $G$
where the labels store the following information. If $x$ is a P-node or a leaf, then its label
is simply $|R^{-1}(x)|$. 
If $q$ is a Q-node with children $x_1, x_2, \dots, x_m$ (from left to right),
then for each $v \in R^{-1}(T_q)$ we define $Q_q(v)$ to be the pair $[a,b]$ 
such that $x_a$ and $x_b$ are the leftmost and rightmost children of $q$ 
whose frontier in $T$ contains $\mathcal{C}(v)$.
Also, if $Q_q(v)=[a,b]$ for some vertex $v$, then we let $Q_q^{\mathrm{left}}(v)=a$ and 
$Q_q^{\mathrm{right}}(v)=b$. 
For some $1 \leq a \leq b \leq m$, the pair $[a,b]$ is a \emph{block} of $q$.
Considering blocks of a Q-node, we will use a terminology that treats them like intervals, 
so two blocks can be disjoint, intersecting, they contain indices, etc.
The label $L(q)$ of $q$ encodes the values $|L_q(a,b)|$ for each $a<b$ in $[m]$, where
$L_q(a,b)$ is the set $\{ v \in R^{-1}(q) \mid Q_q(v)=[a,b] \}$.

Note that a PQ-tree can be labeled in linear time.
Two labeled PQ-trees are \emph{identical}, if they are isomorphic as rooted trees
and the corresponding vertices have the same labels. Two labeled PQ-trees are \emph{equivalent},
if they can be made identical by applying a sequence of transformations as above,
with the modification that when reversing the children of a Q-node, its label must also be adjusted correctly.
The key theorem that yields a way to handle isomorphism questions on interval graphs is the following:
\begin{theorem}[\cite{booth-lueker-interval-isomorphism}]
\label{LPQ_tree}
Let $G_1$ and $G_2$ be two interval graphs, and let $T^L(G_1)$ and $T^L(G_2)$ be the labeled
version of a PQ-tree representing $G_1$ and $G_2$, respectively. Then $G_1$ is isomorphic to $G_2$
if and only if $T^L(G_1)$ is equivalent to $T^L(G_2)$.
\end{theorem}

\begin{figure}[t]
\begin{center}
    \psfrag{a1}[][]{$a_1$}
    \psfrag{a2}[][]{$a_2$}
    \psfrag{a3}[][]{$a_3$}
    \psfrag{b1}[][]{$b_1$}
    \psfrag{b2}[][]{$b_2$}
    \psfrag{b3}[][]{$b_3$}
    \psfrag{c1}[][]{$c_1$}
    \psfrag{c2}[][]{$c_2$}
    \psfrag{c3}[][]{$c_3$}
    \psfrag{c4}[][]{$c_4$}
    \psfrag{d1}[][]{$d_1$}
    \psfrag{d2}[][]{$d_2$}
    \psfrag{d3}[][]{$d_3$}
    \psfrag{d4}[][]{$d_4$}
    \psfrag{e1}[][]{$e_1$}
    \psfrag{e2}[][]{$e_2$}
    \psfrag{f}[][]{$f$}
    \psfrag{g}[][]{$g$}
    \psfrag{T}[][]{$T$}
    \psfrag{G}[][]{$G$}
    \psfrag{p1}[][]{$p_1$}
    \psfrag{p2}[][]{$p_2$}
    \psfrag{q1}[][]{$q_1$}
    \psfrag{q2}[][]{$q_2$}
    \psfrag{Lf}[][]{$f:[1,2]$}
    \psfrag{Lg}[][]{$g:[2,3]$}
    \psfrag{Lc2}[][]{$c_2:[1,2]$}
    \psfrag{Lc3}[][]{$c_3:[1,3]$}
    \psfrag{Lc4}[][]{$c_4:[1,4]$}
    \psfrag{Ld1}[][]{$d_1:[2,5]$}
    \psfrag{Ld2}[][]{$d_2:[3,5]$}
    \psfrag{Ld3}[][]{$d_3:[4,5]$}
    \psfrag{Le1}[][]{$e_1:[2,4]$}
    \psfrag{Le2}[][]{$e_2:[2,4]$}
    \psfrag{--}[][]{$-$}
\includegraphics[scale=0.4]{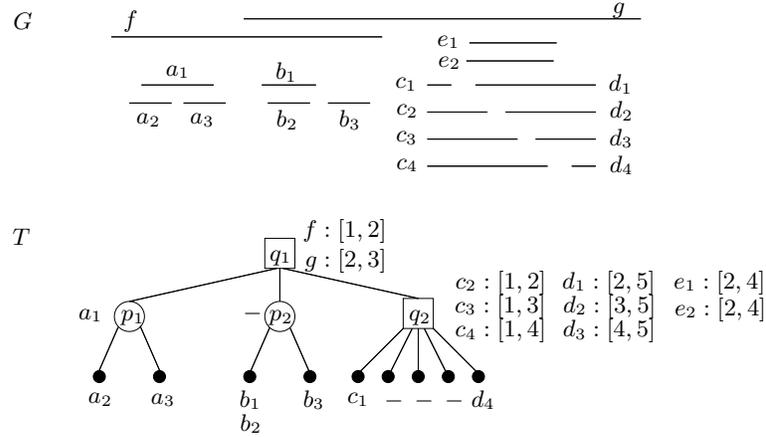}
\end{center}
\caption{
An interval representation of an interval graph $G$ and a labeled PQ-tree $T$ representing $G$. 
Squares, white and black circles represent Q-nodes, P-nodes and leaves, resp.
For each node $x$ in $T$, we listed the vertices in $R^{-1}(x)$, 
together with the values $Q_x(v)$ for each $v \in R^{-1}(x)$ where $x$ is a Q-node. 
As an example, the frontier of $p_2$ is $(\{b_1,b_2,f,g\},\{b_3,f,g\})$. 
}
\label{fig_pq_tree}
\end{figure}

Given a Q-node $q$ in a PQ-tree $T$, let $x_1, \dots, x_m$ denote its children from left to right.
For a given child $x_i$ of $q$, we define $M_q(i)$ to be the set of vertices $v \in R^{-1}(q)$ 
for which $Q_q(v)$ contains $i$, 
i.e. $M_q(i)$ is the union of those sets $L_q(a,b)$ for which $[a,b]$ contains $i$.
Clearly, $M_q(i) \neq M_q(j)$ if $i \neq j$, since this would imply the interchangeability of the
nodes $x_i$ and $x_j$. 
We say that some $w \in R^{-1}(q)$ \emph{starts} or \emph{ends} at $i$ 
if $Q_q^{\mathrm{left}}(v)=i$ or $Q_q^{\mathrm{right}}(v)=i$, respectively. 
We also denote by $M_q^+(i)$ and $M_q^-(i)$ the set of vertices that start or end at $i$, respectively. 
The maximality of the cliques in $F(T_{x_i})$ implies the following observation. 
\begin{myprop}
\label{notempty}
If $q$ is a Q-node in a PQ-tree $T$ and $x_i$ is the $i$-th child of $q$, then 
neither $R^{-1}(T_{x_i}) \cup M_q^+(i)$ nor $R^{-1}(T_{x_i}) \cup M_q^-(i)$ can be empty.
\end{myprop}

Given some interval representation $\rho$ for an interval graph $G$, we denote by 
$v^{\mathrm{left}}_\rho$ and $v^{\mathrm{right}}_\rho$ the left and right endpoints 
of the interval representing $v \in V(G)$. 
If no confusion arises, then we may drop the subscript $\rho$.

%%%%%%%%%%%%%%%%%%%%%%%%%%%%%%%%% Proving NP-hardness %%%%%%%%%%%%%%%%%%%%%%%%%%%%%%%%%%%%%%%%%%%%%%%%%%

%%% Modification
\section{Hardness results}
\label{sect_hard}

In this section, we prove the NP-hardness of \textsc{Induced Subgraph Isomorphism} 
for the case of interval graphs, and we also show the parameterized hardness of this problem, 
where the parameter is the size of the smaller graph.

\begin{theorem}
(1) The \textsc{Induced Subgraph Isomorphism} problem is W[1]-hard if both input graphs are interval graphs, 
and the parameter is the number of vertices in the smaller input graph. \\
(2) The \textsc{Induced Subgraph Isomorphism} problem is NP-complete, if both input graphs are interval graphs.
\end{theorem}

\begin{proof}
To prove (1), we give an FPT reduction from the parameterized \textsc{Clique} problem.
Let $F=(V,E)$ and $k$ be the input graph and the parameter given for \textsc{Clique}.
We assume w.l.o.g. that $F$ is simple and $V=\{v_i \mid i \in [n]\}$.
We construct two interval graphs $G$ and $H$ with $|V(H)|=O(k^2)$ 
such that $H$ is an induced subgraph of $G$ if and only if $F$ has a $k$-clique. 

The vertex set of $G$ consist of the vertices $a_i^{s}, b_i^{s}, c_i^{s}, d_i^{s}, f_{i,i}$ for each $i \in [n]$ and $s \in \{-,+\}$, 
vertices $f_{i,j}, f_{j,i}$ for each $v_i v_j \in E$, and two vertices $g^-$ and $g^+$. 
Note that $|V(G)|=9n+2|E|+2$, which is polynomial in $n$.  
We define the edge set of $G$ by giving an interval representation for $G$.
The intervals $I(x)$ representing a vertex $x \in V(F)$ are defined below.
See also the illustration of Fig.~\ref{fig_gadget}.

\begin{figure}[t]
\begin{center}
    \psfrag{ai+}[][]{$a_i^+$}
    \psfrag{bi+}[][]{$b_i^+$}
    \psfrag{ci+}[][]{$c_i^+$}
    \psfrag{di+}[][]{$d_i^+$}
    \psfrag{ai-}[][]{$a_i^-$}
    \psfrag{bi-}[][]{$b_i^-$}
    \psfrag{ci-}[][]{$c_i^-$}
    \psfrag{di-}[][]{$d_i^-$}
    \psfrag{aj+}[][]{$a_j^+$}
    \psfrag{bj+}[][]{$b_j^+$}
    \psfrag{cj+}[][]{$c_j^+$}
    \psfrag{dj+}[][]{$d_j^+$}
    \psfrag{aj-}[][]{$a_j^-$}
    \psfrag{bj-}[][]{$b_j^-$}
    \psfrag{cj-}[][]{$c_j^-$}
    \psfrag{dj-}[][]{$d_j^-$}
    \psfrag{fii}[][]{$f_{i,i}$}
    \psfrag{fij}[][]{$f_{i,j}$}
    \psfrag{fji}[][]{$f_{j,i}$}
    \psfrag{fjj}[][]{$f_{j,j}$}
    \psfrag{g+}[][]{$g^+$}
    \psfrag{g-}[][]{$g^-$}
    \psfrag{0}[][]{$0$}
\includegraphics[scale=0.4]{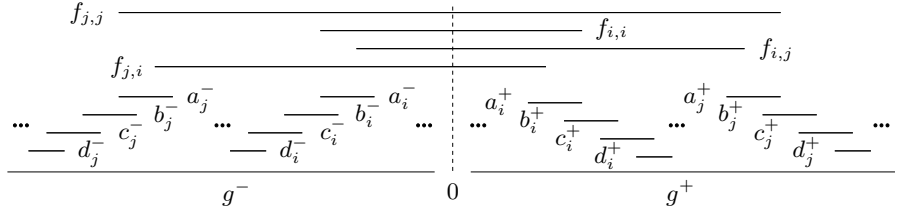}
\end{center}
\caption{
Illustration of the construction of the graph $G$. (The picture assumes $v_i v_j \in E$.)
}
\label{fig_gadget}
\end{figure}

\begin{tabular}{lll}
$I(a_i^+)=[10i-8, 10i-5]$   &   $I(a_i^-)=[-10i+5, -10i+8] \;\;$ & if $i \in [n]$ \\
$I(b_i^+)=[10i-6, 10i-3]$   &   $I(b_i^-)=[-10i+3, -10i+6] \;\;$ & if $i \in [n]$ \\
$I(c_i^+)=[10i-4, 10i-1]$   &   $I(c_i^-)=[-10i+1, -10i+4] \;\;$ & if $i \in [n]$ \\
$I(d_i^+)=[10i-2, 10i]$     &   $I(d_i^-)=[-10i, -10i+2] \;\;$ & if $i \in [n]$ \\
$I(f_{i,i})=[-10i+5, 10i-5]$    &   & if $i \in [n]$ \\
$I(f_{i,j})=[-10i+7, 10j-7] \;$ &  $I(f_{j,i})=[-10j+7, 10i-7] \;\;$ & if $v_i v_j \in E$ \\
$I(g^-)=[-10n,-1]$          &   $I(g^-)=[1,10n]$ & 
\end{tabular}

Note that this construction is symmetric in the sense that for any interval $[x_1,x_2]$ in this interval representation, 
the interval $[-x_2,-x_1]$ is also present.

Also, we define the graph $H$, having $k^2+8k+2$ vertices, as follows. Let the vertex set of $H$ consist of the vertices 
$\widetilde{a}_i^{s}, \widetilde{b}_i^{s}, \widetilde{c}_i^{s}, \widetilde{d}_i^{s}$ for each $i \in [k]$ and $s \in \{-,+\}$, the vertices 
$\widetilde{f}_{i,j}$ for each $(i,j) \in [k]^2$, and two vertices $\widetilde{g}^-$ and $\widetilde{g}^+$. 
Again, we define the edge set of $H$ by giving an interval representation for $H$ as follows.

\begin{tabular}{lll}
$I(\widetilde{a}_i^+)=[10i-8, 10i-5]$   &   $I(\widetilde{a}_i^-)=[-10i+5, -10i+8] \;\;$ & if $i \in [k]$ \\
$I(\widetilde{b}_i^+)=[10i-6, 10i-3]$   &   $I(\widetilde{b}_i^-)=[-10i+3, -10i+6] \;\;$ & if $i \in [k]$ \\
$I(\widetilde{c}_i^+)=[10i-4, 10i-1]$   &   $I(\widetilde{c}_i^-)=[-10i+1, -10i+4] \;\;$ & if $i \in [k]$ \\
$I(\widetilde{d}_i^+)=[10i-2, 10i]$     &   $I(\widetilde{d}_i^-)=[-10i, -10i+2] \;\;$ & if $i \in [k]$ \\
$I(\widetilde{f}_{i,i})=[-10i+5, 10-5]$ &  & if $i \in [k]$ \\
$I(\widetilde{f}_{i,j})=[-10i+7, 10j-7] \;$ &  $I(\widetilde{f}_{j,i})=[-10j+7, 10i-7] \;\;$ & if $i,j \in [k]$, $i \neq j$\\
$I(\widetilde{g}^-)=[-10k,-1]$          &   $I(\widetilde{g}^-)=[1,10k]$ \
\end{tabular}

First, if $C$ is a set of $k$ vertices in $F$ that form a clique, 
then $H$ is isomorphic to the subgraph of $G$ induced by the vertices $a_i^{s}, b_i^{s}, c_i^{s}, d_i^{s}, f_{i,i}$ 
for each $v_i \in C$ and $s \in \{-,+\}$, 
the vertices $f_{i,j}, f_{j,i}$ for each $\{v_i, v_j\} \subseteq C$, and the two vertices $g^-$ and $g^+$. 
This can be proven by presenting an isomorphism $\varphi$ from $H$ to the subgraph of $G$ induced by these vertices. 
It is easy to verify that the function $\varphi$ defined below indeed yields an isomorphism. 
Here, $c(i)$ denotes the index of the $i$-th vertex in the clique $C$, i.e. $C=\{v_{c(i)} \mid i \in [k]\}$. 

\begin{tabular}{ll}
$\varphi(\widetilde{x}_i^s)=x_{c(i)}^s \quad$  & for each $x \in \{a,b,c,d\}, s \in \{-,+\}, i \in [k]$ \\
$\varphi(\widetilde{f}_{i,j})=f_{c(i),c(j)}^s \quad$ & for each $i,j \in [k]^2$ \\
$\varphi(\widetilde{g}^s)=g^s \quad$ & for each $s \in \{-,+\}$ \\
\end{tabular}

For the other direction, suppose that $\varphi$ is an isomorphism from $H$ to an induced subgraph of $G$.
We set $F=\{ f_{i,j} \mid i=j$ or $v_i v_j \in E \}$, and we define $Z$ to contain those vertices of $G$ whose interval contains $0$.  

{\bf Claim.} 
If there is a disjoint union of two $k$-stars in $K$ with centers $u_1$ and $u_2$, induced by vertices $\{u_1,u_2\} \cup J$,
then $\{\varphi(u_1), \varphi(u_2)\} = \{g^-, g^+ \}$ and $\varphi(J) \cap F = \emptyset$. 
To prove this claim, note that the vertices of $J$ are independent, so 
there can be at most one vertex in $\varphi(J)$ whose interval contains $0$. Thus, either
$\varphi(u_1)$ or $\varphi(u_2)$ must not be in $Z$, and must be adjacent to at least $k$ vertices not in $Z$.
This implies that $\varphi(u_1)$ or $\varphi(u_2)$ must indeed be $g^-$ or $g^+$. Assuming, say, $\varphi(u_1) =  g^-$ 
(the remaining cases are analogous), we obtain that the only common neighbor of the $k$ vertex of $\varphi(J)$ 
not adjacent to $\varphi(u_1)$ can be $g^+$. This immediately implies $\{\varphi(u_1), \varphi(u_2)\} = \{g^-, g^+ \}$. 
From this, $\varphi(J) \cap F = \emptyset$ is clear, since no vertex of $J$ is adjacent to both $u_1$ and $u_2$. 
Hence, the claim is true.

Now, note that for some $x \in \{a,b,c,d\}$, the vertex set
$\{ \widetilde{x}_i^s \mid i \in [k], s \in \{-,+\} \} \cup \{ \widetilde{g}^-, \widetilde{g}^+\}$ 
induces the disjoint union of two $k$-stars having centers $\widetilde{g}^-$ and $\widetilde{g}^+$ in $H$. 
Therefore, applying the above claim to each these vertex sets with $x \in \{a,b,c,d\}$, 
we obtain that $\{ \varphi(\widetilde{g}^-), \varphi(\widetilde{g}^+)\}=\{ g^-,g^+\}$, and also that
$\varphi(\widetilde{X}) \cap F = \emptyset$ for the set $\widetilde{X}$ 
containing the vertices of the form $\widetilde{x}_i^s$ where 
$x \in \{a,b,c,d\}, s \in \{-,+\}$ and $i \in [k]$. 
By the symmetry of $H$ and $G$, we can assume w.l.o.g. that $\varphi(\widetilde{g}^-)=g^-$ and $\varphi(\widetilde{g}^+)=g^+$.

From this, we have that exactly $4k$ vertices of $\varphi(\widetilde{X}) $ 
are represented by an interval whose left endpoint is positive, 
and the remaining $4k$ vertices of  $\varphi(\widetilde{X}) $ are represented by an interval whose right endpoint is negative.
Now, observe that the vertices of $\widetilde{X}$ induce exactly $2k$ paths of length 4 in $H$, which leads us to 
the fact that their images by $\varphi$ must also induce 4-paths. 
Using this, it follows that for each $i \in [k]$ we can define $c(i,+), c(i,-) \in [n]$ such that 
$$
\varphi(\{\widetilde{a}_i^s,\widetilde{b}_i^s,\widetilde{c}_i^s,\widetilde{d}_i^s \}) 
= \{a_{c(i,s)}^s, b_{c(i,s)}^s, c_{c(i,s)}^s, d_{c(i,s)}^s \}
$$
for each $i \in [k]$ and $s \in \{-,+\}$.

Note also that for both $s \in \{-,+\}$, the vertex $\widetilde{f}_{i,i}$ is adjacent to exactly two vertices from 
$\{\widetilde{a}_i^s,\widetilde{b}_i^s,\widetilde{c}_i^s,\widetilde{d}_i^s\}$, 
but the only vertex adjacent to exactly two vertices from $\{a_{c(i,s)}^s, b_{c(i,s)}^s, c_{c(i,s)}^s, d_{c(i,s)}^s\}$
is the vertex $f_{c(i,s),c(i,s)}$. From this, we get that $\varphi(\widetilde{f}_{i,i})=f_{c(i,-),c(i,-)} =f_{c(i,+),c(i,+)}$, 
implying also $c(i,-) = c(i,+)$. 

Finally, note that if $i \neq j$, then $\widetilde{f}_{i,j}$ is adjacent to exactly one vertex both from 
$\{\widetilde{a}_i^-,\widetilde{b}_i^-,\widetilde{c}_i^-,\widetilde{d}_i^-\}$ and from 
$\{\widetilde{a}_j^+,\widetilde{b}_j^+,\widetilde{c}_j^+,\widetilde{d}_j^+\}$. This implies that 
$\varphi(\widetilde{f_{i,j}} = f_{c(i,-),c(j,+)}$ must hold, but $f_{c(i,-),c(j,+)}$ only exists if 
$v_{c(i,-)}$ and $v_{c(j,+)}$ are adjacent in $F$. Clearly, this implies 
that the vertices $\{ v_{c(i,-)} = v_{c(i,+)}  \mid i \in [k] \}$ form a clique in $F$, 
hence the second direction of the reduction is correct as well. 

Observe that by the size of $G$ and $H$, this yields an FPT-reduction from the parameterized \textsc{Clique} problem to 
the \textsc{Cleaning}(\textsl{Interval, Interval}) problem 
(i.e. the \textsc{Induced Subgraph Isomorphism} problem for interval graphs)
parameterized by the number of vertices in the smaller input graph, proving (1). 
Also, note that the construction of $G$ and $H$ takes time polynomial in $|V(F)|$ and $k$, so by the NP-hardness of the (unparameterized) 
\textsc{Maximum Clique} problem, this proves that the (unparameterized) \textsc{Cleaning}(\textsl{Interval, Interval}) problem 
is NP-hard as well. Its containment in NP is trivial, finishing the proof of (2). 
\qed
\end{proof}

%%%%%%%%%%%%%%%%%%%%%%%%%%%%%%%%% Here begins the interval cleaning algo %%%%%%%%%%%%%%%%%%%%%%%%%%%%%%%%%%

\section{Cleaning an interval graph}
\label{sect_algo}
In this section, we present an algorithm that solves the \textsc{Cleaning}(\textsl{Interval, Interval}) 
problem. Given an input $(G',G)$ of this problem, we call a set $S \subseteq V(G)$ a 
\emph{solution} for $(G',G)$, if $G'$ is isomorphic to $G-S$. In this case, 
let $\phi_S$ denote an isomorphism from $G'$ to $G-S$.
Remember that $k=|V(G)|-|V(G')|$ is the parameter of the instance $(G',G)$.
We denote by $T$ and $T'$ the labeled PQ-tree representing $G$ and $G'$, respectively.
%Let us fix an interval representation $\rho$ of $G$ such that the root of $T$ agrees $\rho$.
Let us fix an interval representation of $G$.
For a subset $X$ of $V(G)$, let $X^{\mathrm{left}}=\min \{x^{\mathrm{left}} \mid x \in X\}$ 
and $X^{\mathrm{right}}=\max \{x^{\mathrm{right}} \mid x \in X\}$.

Our algorithm for \textsc{Cleaning}(\textsl{Interval, Interval}) is based on an algorithm denoted by
$\mathcal{A}$ whose output on an input $(G',G)$ can be one of the following three concepts:
\begin{itemize}
\item a \outputstyle{necessary set}. 
We call a set $N \subseteq V(G)$ a \emph{necessary set} for $(G',G)$, if $(G',G)$
has a solution if and only if there is a vertex $x \in N$ such that $(G',G-x)$ has a solution.
Given a necessary set for $(G',G)$, we can branch on including one of its vertices in the solution.
\item a \outputstyle{reduced input}.
For subgraphs $H$ and $H'$ of $G$ and $G'$, respectively, we say that 
$(H',H)$ is a \emph{reduced input} for $(G',G)$, if $(G',G)$ is solvable if and only
if $(H',H)$ is solvable, every solution for $(H',H)$ is a solution for $(G',G)$,
and $|V(H')|+|V(H)|<|V(G')|+|V(G)|$.
Given a reduced input for $(G',G)$, we can clearly solve it instead of solving $(G',G)$.
\item an \outputstyle{independent subproblem}.
For subgraphs $H$ and $H'$ of $G$ and $G'$, respectively, we say that 
$(H',H)$ is an \emph{independent subproblem} of $(G',G)$ having parameter $k$, 
if its parameter is at least 1 but at most $k-1$, 
for any solution $S$ of $(G',G)$ the set $S \cap V(H)$ is a solution for $(H',H)$,
and if $(G',G)$ admits a solution then any solution $S$ of $(H',H)$ 
can be extended to be a solution for $(G',G)$.
Note that given an independent subproblem of $(G',G)$, we can find a vertex of the solution
by solving the independent subproblem having parameter smaller than $k$.
\end{itemize}
Observe that if $N$ is a necessary set for either an independent subproblem or a 
reduced input for $(G',G)$, then $N$ must be a necessary set for $(G',G)$ as well.

In Section~\ref{structure} we make some useful observations 
about the structure of an interval graph.
In Sections \ref{reductions} and \ref{qq_case}, we describe algorithm $\mathcal{A}$, 
that, given an input instance of \textsc{Cleaning}(\textsl{Interval, Interval}) with parameter $k$, 
does one of the followings in linear time:
\begin{itemize}
\item
either determines a reduced input for $(G',G)$, 
\item 
or branches into at most $f_1(k)=k^{O(k)}$ possibilities, in each of the branches
producing a necessary set of size at most $2k+1$ or an independent subproblem of $(G',G)$.
\end{itemize}
Note that in the first case no branching is involved. 
If the second case applies and $\mathcal{A}$ branches, 
then the collection of outputs returned in the obtained branches must contain a correct output. 
In other words, at least one of the branches must produce an output that is 
indeed a necessary set of the required size or an independent subproblem of $(G',G)$.

\begin{figure}[t]
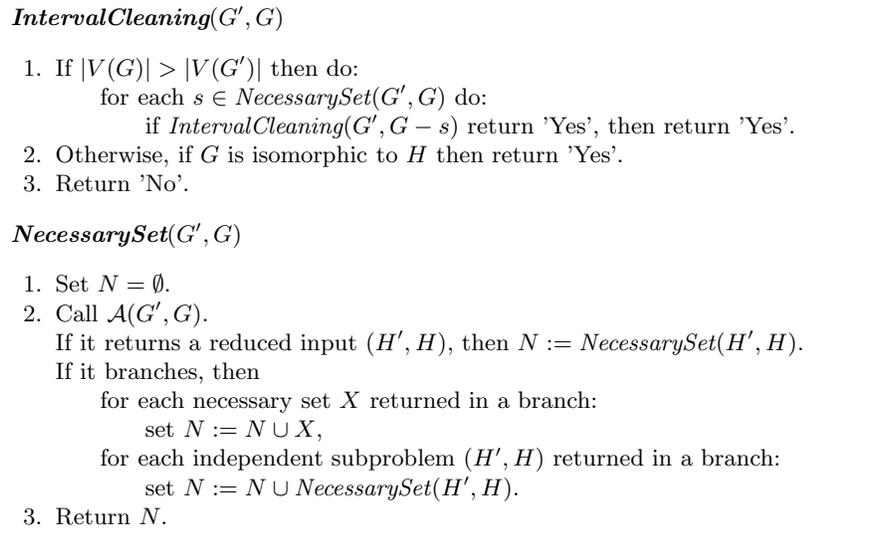

\begin{center}
\fbox{ 
\parbox{11.7cm}{ \textbf{$\textit{IntervalCleaning}(G',G)$}
\begin{enumerate}
\item If $|V(G)|>|V(G')|$ then do: 
\begin{enumerate}
\item[ ] for each $s \in \textit{NecessarySet}(G',G)$ do:
\begin{enumerate}
\item[ ] if  $\textit{IntervalCleaning}(G',G-s)$ return 'Yes', then return 'Yes'.
\end{enumerate}
\end{enumerate}
\item Otherwise, if $G$ is isomorphic to $H$ then return 'Yes'.
\item Return 'No'.
\end{enumerate}
\textbf{$\textit{NecessarySet}(G',G)$}
\begin{enumerate}
\item Set $N=\emptyset$.
\item Call  $\mathcal{A}(G',G)$. \\
If it returns a reduced input  $(H',H)$, then $N:= \textit{NecessarySet}(H',H)$. \\
If it branches, then
\begin{enumerate}
\item[ ] for each necessary set $X$ returned in a branch:
\begin{enumerate}
\item[ ]  set $N:=N \cup X$,
\end{enumerate}
\item[ ] for each independent subproblem $(H',H)$ returned  in a branch: 
\begin{enumerate}
\item[ ] set $N:=N \cup \textit{NecessarySet}(H',H)$. 
\end{enumerate}
\end{enumerate}
\item Return $N$.
\end{enumerate}
}}
\end{center}
\caption{Outline of algorithms $\textit{IntervalCleaning}$ and $\textit{NecessarySet}$.}
\label{algo}
\end{figure}

Let us show how such an algorithm can be used as a sub-procedure in 
order to solve the \textsc{Cleaning}(\textsl{Interval, Interval}) problem.
(See Fig.~\ref{algo} for an outline of the algorithm.) 
First, we construct an algorithm called $\textit{NecessarySet}$ that
given an instance $(G',G)$ of \textsc{Cleaning}(\textsl{Interval, Interval}) finds 
a necessary set for $(G',G)$ in quadratic time.
$\textit{NecessarySet}$ works by running $\mathcal{A}$ repeatedly, starting with the given input.
In the case when $\mathcal{A}$ returns a reduced input, $\textit{NecessarySet}$ 
runs $\mathcal{A}$ with this reduced input again. 
In the case when $\mathcal{A}$ branches, returning a necessary set or an independent subproblem
in each branch, $\textit{NecessarySet}$ runs $\mathcal{A}$
on each independent subproblem produced in any of the branches. 
Applying this method iteratively (and thus possibly branching again), 
we will get a necessary set at the end of each branch.
Note that the parameter of the input decreases whenever a branching happens, 
and thus the corresponding search tree has at most $f_1(k)^k$ leaves. 
Since at least one of the branches is correct, by taking the union of the necessary sets 
produced in the leaves of the search tree, 
we get a necessary set of size $f_2(k)=(2k+1)f_1(k)^k$ for $(G',G)$.
As each run of $\mathcal{A}$ takes linear time, and the number of calls of $\mathcal{A}$
is also linear in a single chain of branches, the whole algorithm takes quadratic time. 

Now, we can solve \textsc{Cleaning}(\textsl{Interval, Interval}) by  using
$\textit{NecessarySet}$. 
First, given an input $(G',G)$, we run $\textit{NecessarySet}$ on it. 
We branch on choosing a vertex $s$ from the produced output to put into the solution,  
and repeat the whole procedure with input $(G',G-s)$.
This means a total of $f_2(k)=(2k+1)f_1(k)^k$ new inputs to proceed with.
We have to repeat this at most $k$ times, so 
the whole algorithm has running time $O(f_2(k)^k |I|^2)$, 
where $|I|$ is the size of the original input of the problem.
We can state this in the following theorem:

\begin{theorem}
\label{FPT}
\textsc{Cleaning}(\textsl{Interval, Interval}) on input $(G',G)$ can be  solved in
time $O(f(k)n^2)$ for some function $f$, where $|V(G')|=n$ and $|V(G)|=n+k$.
\end{theorem}

\subsection{Some structural observations}
\label{structure}

A nonempty set $M \subseteq V(G)$ is a \emph{module} of $G$, if 
for every $x \in V(G) \setminus M$, $N_G(x)$ either includes $M$ or is disjoint from $M$.
A module $M$ in $G$ is \emph{complete}, if $G[M]$ is connected and there is no vertex in
$x \in N_G(M)$ such that $N_G(x) \subseteq N_G[M]$.
Lemma~\ref{module_char} gives a characterization of the complete modules of an interval graph.
For an illustration, see Fig.~\ref{fig_pq_tree}. Note that $\{a_1\}$ and $\{a_2,a_3\}$ are modules of $G$
that are not complete. The sets $\{a_1,a_2,a_3\}$, $\{b_1,b_2\}$ and
$\{c_1, c_2,c_3,c_4,d_1,d_2,d_3,d_4,e_1,e_2\}$ are examples of complete module characterized by 
(a) of Lemma~\ref{module_char}, and the set $\{e_1,e_2\}$ illustrates the complete modules characterized by 
(b) of Lemma~\ref{module_char}.

\begin{mylemma}
\label{module_char} Given an interval graph $G$ and a labeled PQ-tree $T$ representing $G$, 
some set $M \subseteq V(G)$ is a complete module of $G$,
if and only if one of the following statements holds:  \\
(a) $M=R^{-1}(T_z)$ for some $z \in V(T)$, and if $z$ is a P-node then $R^{-1}(z) \neq \emptyset$ \\
(b) $M=L_q(a,b) $ for some Q-node $q \in V(T)$ having children $x_1,\dots, x_m$ 
and some pair $(a,b)$ with $a<b$, such that $R^{-1}(T_{x_i}) = \emptyset$ for each  $i$ contained in $[a,b]$,
and $L_q(a',b')=\emptyset$ for each $[a',b']$ properly contained in $[a,b]$.
\end{mylemma}

\begin{proof}
First, let $M$ be a complete module in $G$. 
Let us choose a vertex $v \in M$ such that $R(v)$ is the closest possible 
to the root of $T$. Since $G[M]$ is connected, $v$ is unique, 
and we also get $R^{-1}(T_{R(v)}) \supseteq M$.
First, suppose that $R(v)$ is a P-node or a leaf. 
Then $v$ is contained in each clique of $F(T_{R(v)})$.
Thus, if $R(x)$ is in $T_{R(v)}$ for some vertex $x$, then $N_G(x) \subseteq N_G(v) \subseteq N_G[M]$.
By the completeness of $M$, we get $x \in M$. Hence, $R^{-1}(T_{R(v)}) \subseteq M$
implying $R^{-1}(T_{R(v)}) = M$. Therefore, (a) holds in this case.

Now, suppose that $R(v)$ is a Q-node $q$ with children $x_1, \dots, x_m$, and let 
$M_q = M \cap R^{-1}(q)$. Let $a = \min \{ Q_q^{\mathrm{left}}(w) \mid w \in M_q\}$ and 
$b = \max \{ Q_q^{\mathrm{right}}(w) \mid w \in M_q\}$
Using the completeness of $M$, we can argue again that 
$R^{-1}(T_{x_h}) \subseteq M$ for each $h$ contained in $[a,b]$ and that $w \in M$ 
holds for each $w \in R^{-1}(q)$ such that $Q_q(w)$ is contained in $[a,b]$.
Thus, if $[a,b]=[1,m]$ then $M=R^{-1}(T_q)$, implying that (a) holds. 
Otherwise, as $q$ is a Q-node, there must exist a  vertex $u \in R^{-1}(q) \setminus M$ 
such that $Q_q(u)$ properly intersects $[a,b]$.
As $u$ must be adjacent to each vertex of $M$ (as $M$ is a module), we get that 
$R^{-1}(T_{x_h}) = \emptyset$ for every $h$ in $[a,b]$ that is not contained in $Q_q(u)$. 
In particular, we get that either $R^{-1}(T_{x_a})=\emptyset$ or $R^{-1}(T_{x_b})=\emptyset$.
We can assume w.l.o.g. that $R^{-1}(T_{x_a})=\emptyset$ holds.
Thus, $M_q^-(a) \neq \emptyset$, and since $M_q^-(a) \cap M \neq \emptyset$, 
using again that $M$ is a module, we obtain that each $w \in M_q$ must start in $a$ 
and also that $R^{-1}(T_{x_h}) = \emptyset$ for every $h$ in $[a,b]$.
Note that this implies $M_q=M$.
Now, from $R^{-1}(T_{x_b})=\emptyset$ we get in a similar way that each $w \in M$ must end in $b$, 
proving $Q_q(w)=[a,b]$ for every $v,w \in M$. Now, using the completeness of $M$ 
and putting together these facts, we get that the conditions of (b) must hold.

For the other direction, it is easy to see that if (a) holds for some $M$, 
then $M$ indeed must be a complete module of $G$.
Second, if $M=L_q(a,b)$ for some $q$ and $[a,b]$,  then $M$ is clearly a module, 
and the remaining conditions of (b) ensure that $M$ is complete.
\qed
\end{proof}

We will say that a complete module $M$ is \emph{simple}, if the conditions in (b) hold for $M$.
Clearly, $N_G(M)$ is a clique if and only if $M$ is not simple, 
and if $M$ is simple then $G[M]$ is a clique. 
In Fig.~\ref{fig_pq_tree}, $\{e_1,e_2\}$ is a simple complete module.

For a graph $H$, some set $M \subseteq V(G)$ is \emph{an occurrence of $H$ in $G$ as a complete module}, 
if $M$ is a complete module for which $G[M]$ is isomorphic to $H$.
Let $\mathcal{M}(H,G)$ be the set of the occurrences of $H$ in $G$ as a complete module.
Using that each element of $\mathcal{M}(H,G)$ is a subset of $V(G)$ having size $|V(H)|$, 
we obtain the following consequence of Lemma \ref{module_char}. 

\begin{myprop}
\label{occurrences}
For a graph $H$, the elements $\mathcal{M}(H,G)$ are pairwise disjoint. 
\end{myprop}

Moreover, if the graph $H$ is not a clique, then none of the occurrences of $H$ in 
$G$ as a complete module can be simple, so each set in $\mathcal{M}(H,G)$ must be 
of the form $R^{-1}(T_z)$ for some non-leaf node $z$ of $T$. 
This yields that the sets in $\mathcal{M}(H,G)$ are independent 
(where two vertex sets in a graph are independent if there is no edge between them). 
Lemma~\ref{module_indep} below states some observations about what happens to a set of 
disjoint and independent complete modules in a graph after adding or deleting a vertex. 

\begin{mylemma}
\label{module_indep} 
Suppose that $s \in V(G)$. \\
(1) If $M_1, \dots, M_\ell$ are disjoint independent complete modules in $G-s$,
then $M_i$ is a complete module in $G$ for at least $\ell-4$ indices $i \in [\ell]$. \\
(2) If $M_1, \dots, M_\ell$  are disjoint independent complete modules in $G$,
then $M_i$ is a complete module in $G-s$ for at least $\ell-4$ indices $i \in [\ell]$.
\end{mylemma}

\begin{figure}[t]
\begin{center}
    \psfrag{s1}[][]{$s_1$}
    \psfrag{s2}[][]{$s_2$}
    \psfrag{s3}[][]{$s_3$}
    \psfrag{s4}[][]{$s_4$}
    \psfrag{s5}[][]{$s_5$}
    \psfrag{s6}[][]{$s_6$}
    \psfrag{s7}[][]{$s_7$}
    \psfrag{s8}[][]{$s_8$}
    \psfrag{Mi}[][]{$M_i$}
    \psfrag{Mi-1}[][]{$M_{i-1}$}
    \psfrag{Mi+1}[][]{$M_{i+1}$}
\includegraphics[scale=0.4]{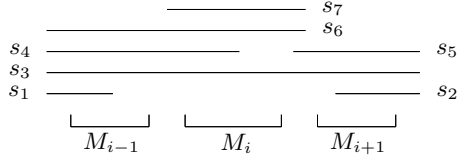}
\end{center}
\caption{
$M_{i-1}$, $M_i$, and $M_{i+1}$ illustrate complete modules of $G-S$. The set
$M_i$ is untouched by $s_1$, $s_2$, and $s_3$, but this is not true for any vertex $s_j$, $j \geq 4$.
}
\label{fig_modules}
\end{figure}

\begin{proof}
As $M_i$ and $M_j$ are independent if $i \neq j$, we can
assume that $M_1^{\mathrm{left}} \leq M_1^{\mathrm{right}} < \dots 
< M_{\ell}^{\mathrm{left}} \leq M_{\ell}^{\mathrm{right}}$.
Recall that each $M_i$ is connected by the definition of a complete module. 
We say that $M_i$ is \emph{untouched} (by $s$), if either 
$s^{\mathrm{left}} \leq M_{i-1}^{\mathrm{right}}$ and  $s^{\mathrm{right}} \geq M_{i+1}^{\mathrm{left}}$, 
or $s^{\mathrm{right}} < M_{i-1}^{\mathrm{right}}$, or  $s^{\mathrm{left}} > M_{i+1}^{\mathrm{left}}$.
(See also Fig.~\ref{fig_modules}.)
If $M_a$ and $M_b$ are the first and the last one, respectively, among the sets 
$M_1, \dots, M_\ell$ that have a vertex adjacent to $s$, then each $M_i$ except for 
$M_{a-1}$, $M_{a}$, $M_{b}$, and $M_{b+1}$  must be untouched by $s$.

To see (1), we show that if a complete module of $G-s$ is untouched, 
then it is a complete module of $G$.
So assume that $M_i$ is untouched. Clearly, $s \notin M_i$.
Since either $N_G(s) \supseteq M_i$ or $N_G(s) \cap M_i=\emptyset$,
$M_i$ remains to be a module in $G$.
Also, if $s \in N_G(M_i)$, then $s$ must have a neighbor in $M_{i-1}$ and $M_{i+1}$. 
Thus, $N_G(s) \not\subseteq N_G[M_i]$, so the completeness of $M_i$ in $G-s$ implies 
its completeness in $G$ as well.

To prove (2), suppose that $M_i$ is an untouched complete module in $G$.
Clearly, $M_i$ is a module in $G-s$ as well, and since $s \notin M_i$, 
$M_i$ remains connected in $G-s$. Let $x$ be a vertex in $N_G(M_i)$. 
By the completeness of $M_i$, $x$ is adjacent to some vertex $y \notin N_G[M_i]$.
Suppose that $x$ doesn't have a neighbor outside $N_{G-s}[M_i]$ in $G-s$.
This can only happen if $y = s$. Now, since $y \notin N_G[M_i]$ and $M_i$ is untouched by $s$, 
$x$ must also be adjacent to a vertex of $M_{i-1}$ or $M_{i+1}$.  
Thus, $x$ has a neighbor in $V(G-s) \setminus N_{G-s}[M_i]$, proving the completeness of $M_i$.
As $M_i$ is untouched for at least $\ell-4$ indices $i \in [\ell]$, the statement follows.
\qed
\end{proof}

In the case when $H$ is a clique and $K$ is an occurrence of $H$ in $G$
as a complete module, we get that 
either $K = R^{-1}(\ell)$ for some leaf $ \ell \in V(T)$, or $K$ is simple, i.e.
$K =L_q(a,b)$ for some Q-node $q \in V(T)$ and some block $[a,b]$.
In the latter case, Lemma~\ref{h-short_module} states a useful observation 
about the block $[a,b]$. This lemma uses the following definition:
we say that a complete module $K$ of $G$ is \emph{$h$-short}, if either 
$K = R^{-1}(\ell)$ for some leaf $\ell \in V(T)$, or $K = L_q(a,b)$ 
for some Q-node $q \in V(T)$ and some block $[a,b]$ with $b-a \leq h$. 
The sets $\{e_1,e_2\}$ and $\{b_1,b_2\}$ are $2$-short complete modules 
of $G$ in Fig.~\ref{fig_pq_tree}.

\begin{mylemma}
\label{h-short_module}
If $K$ is a complete module in $G$ such that $G[K]$ is a clique but $K$ is not $h$-short, 
then $|N_G(K)| \geq 2(h+1)$.
\end{mylemma}

\begin{proof}
By the conditions of the lemma, we know that $K =L_q(a,b)$ for some Q-node $q \in V(T)$ 
with children $x_1, \dots, x_m$ and some block $[a,b]$ such that $b-a \geq h+1$.
By the completeness of $K$, we get that $R^{-1}(T_{x_h})=\emptyset$ for any $h$ contained in $[a,b]$,
so $M^+(h)$ and $M^-(h)$ cannot be empty.
Taking these sets for all $h$ in $[a,b]$, with the exception of the sets $M^+(a)$ and $M^-(b)$, 
we get $2(b-a) \geq 2(h+1)$ nonempty sets that are pairwise disjoint, 
each containing some vertex of $N_G(K)$. This implies the bound $N_G(K) \geq 2(h+1)$.
\qed
\end{proof}

Observe that if two different $h$-short complete modules $K_1$ and $K_2$ in $G$ are not independent, 
then $K_1 = L_q(a,b)$ and $K_2 = L_q(c,d)$ must hold for some Q-node $q$ in $T$ 
and some blocks $[a,b]$ and $[c,d]$ that properly intersect each other.
Now, if $b-a \leq h$, then there can be at most $2h$ such blocks $[c,d]$ for which these conditions hold. 
This implies that given a $h$-short complete module $K$, there can be at most $2h$ different
$h$-short complete modules of $G$ neighboring $K$ (but not equal to $K$).
It is also easy to see that the maximum number of pairwise neighboring $h$-short complete modules 
in a graph is at most $h+1$. Making use of these facts, Lemma~\ref{module_clique} 
states some results about $h$-short complete modules of a graph 
in a similar fashion as Lemma~\ref{module_indep}.
As opposed to Lemma~\ref{module_indep}, here we do not require the 
complete modules to be independent.

\begin{mylemma}
\label{module_clique}
Suppose that $s \in V(G)$. \\
(1) If $M_1, \dots, M_\ell$ are disjoint $h$-short complete modules in $G-s$,
then $M_i$ is a complete module in $G$ for at least $\ell-(3h+5)$ indices $i \in [\ell]$. \\
(2) If $M_1, \dots, M_\ell$  are disjoint $h$-short complete modules in $G$,
then $M_i$ is a complete module in $G-s$ for at least $\ell-(4h+3)$ indices $i \in [\ell]$.
\end{mylemma}

\begin{proof} 
The proof relies on the observation that there are only a few indices $i$ such that 
$[M_i^{\mathrm{left}},M_i^{\mathrm{right}}]$ contains $s^{\mathrm{left}}$ or $s^{\mathrm{right}}$.

To see (1), suppose that $M_i$ is not a $h$-short complete module in $G$ for some $i$. 
Clearly, $G[M_i]$ is connected. First, assume that $M_i$ is not a module 
because there are some $x,y \in M_i$ such that $s$ is adjacent to $x$ but not to $y$.
In this case, either $x^{\mathrm{left}}< s^{\mathrm{right}} <y^{\mathrm{left}}$ 
or $y^{\mathrm{right}}<s^{\mathrm{left}} <x^{\mathrm{right}}$. It is not hard to see that 
this implies that there can be at most two such modules $M_i$. 
Now, assume that $M_i$ is a module, but is not complete. 
This implies that $M_i \subseteq N_G(s) \subseteq N_G[M_i]$ is true. 
Note that if $j \neq i$ then $M_j \subseteq N_G(s) \subseteq N_G[M_j]$ is only possible
if $M_i$ and $M_j$ are neighboring. Thus, there can be at most $h+1$ such indices $i$.

Finally, if $M_i$ is complete module in $G$ but it is not $h$-short, then 
the number of maximal cliques containing the vertices of $M_i$ must be more in $G$ than in $G-s$, 
implying that either $M_i^{\mathrm{left}} < s^{\mathrm{left}} \leq M_i^{\mathrm{right}}$ 
or $M_i^{\mathrm{left}} \leq  s^{\mathrm{right}} \leq M_i^{\mathrm{right}}$.
As $M_i^{\mathrm{left}} < s^{\mathrm{left}} \leq M_i^{\mathrm{right}}$ 
and $M_j^{\mathrm{left}} < s^{\mathrm{left}} \leq M_j^{\mathrm{right}}$ can only 
hold simultaneously if $M_i$ and $M_j$ are neighboring, 
and such a statement is also true for the latter condition, 
we get that there can be at most $2(h+1)$ indices $i$ for which $M_i$ is $h$-short in $G-s$ but not in $G$.
Summing up these facts, we obtain that there can be at most $2+(h+1)+2(h+1) = 3h+5$ indices $i$
for which $M_i$ is not a $h$-short complete module in $G$.

To prove (2), notice that each $M_i$ remains a module in $G-s$ as well.
Observe also that if $s \notin M_i$, then $M_i$ remains connected in $G-s$.
By the disjointness of the sets $M_1, \dots, M_\ell$, each of them is connected in $G-s$ except for at most one.
Suppose that $M_{i_1}$, $M_{i_2}$, and $M_{i_3}$ are independent, and for each $j \in \{1,2,3\}$, 
$M_{i_j}$ is a connected module in $G-s$ but it is not complete. 
This means that there are vertices $x_1$, $x_2$, and $x_3$
such that $x_j \in N_G(M_{i_j})$, but $N_G(x_j) \subseteq  N_G[M_{i_j}]$ for each $j$.
By the completeness of these modules in $G$, this implies that each of 
$x_1$, $x_2$, and $x_3$ are adjacent to $s$, and $s \notin N_G[M_{i_j}]$ for any $j$.
But this can only hold if some $x_j$ 
is adjacent to each vertex of $M_{i_{j'}}$ for some $j \neq j'$, 
and since $M_{i_{j}}$ and $M_{i_{j'}}$ are independent, this
contradicts the assumption that $N_G(x_j) \subseteq  N_G[M_{i_j}]$.
Thus, there cannot exist such indices $i_1$, $i_2$ and $i_3$, implying that 
we can fix two indices $j$ and $j'$ such that for any $M_i$ that is a 
connected module in $G-s$ but not complete,  $M_i$ is neighboring either $M_j$ or $M_{j'}$, 
implying that there can be at most $2(2h)+2$ such indices $i$.
To finish, observe that if $M_i$ is a complete module in $G-s$, then it must be $h$-short, 
as the deletion of $s$ cannot increase the number of maximal cliques that contain $M_i$.
\qed
\end{proof}

\subsection{Reduction rules}
\label{reductions}
In this section, we introduce some reduction rules, each of which can be applied in linear time,
and provides a necessary set, an independent subproblem, or a reduced input, as described earlier.
Our aim is to handle all cases except for the case when both $G$ and $G'$ have 
a PQ-tree with a Q-node root. We always apply the first possible reduction.
From now on, we assume that $S$ is a solution for $(G',G)$ and
$\phi_S$ is an isomorphism from $G'$ to $G-S$.

%%%%%%%%%%%%%%%%%%%%%%%%%%%%%%%%% Isomorphic components %%%%%%%%%%%%%%%%%%%%%%%%%%%%%%%%%%%

{\bf Rule 1. Isomorphic components.}
Lemma \ref{iso_comp} yields a simple reduction: if $G$ and $G'$ have isomorphic components, then 
algorithm $\mathcal{A}$ can output a \outputstyle{reduced input} of $(G',G)$.
%isomorphism of interval graphs can be tested in linear time, 
Note that partitioning a set of interval graphs into isomorphism equivalence classes can be done in linear time
\cite{booth-lueker-interval-isomorphism} (see also \cite{dinitz-itai-rodeh-subtree-iso,zemlyachenko-canonical,zemlyachenko-tree-iso}).
Hence, this reduction can also be performed in linear time.

\begin{mylemma}
\label{iso_comp}
If $K$ and $K'$ are connected components of $G$ and $G'$, respectively, 
and $K$ is isomorphic to $K'$,
then $(G'-K',G-K)$ is a reduced input of $(G',G)$.
\end{mylemma}

\begin{proof}
Trivially, $G'-K'$ has fewer vertices than $G'$, and any solution for $(G'-K',G-K)$
is a solution for $(G',G)$ as well, by the isomorphism of $K'$ and $K$.
Therefore, we only have to prove that if $(G',G)$ is solvable then $(G'-K',G-K)$ is also solvable.
Clearly, if $S \cap V(K) = \emptyset$, then we can assume w.l.o.g. 
that $\phi_S(K')=K$. In this case, $S$ is a solution for $(G'-K',G-K)$.

On the other hand, if $S \cap V(K) \neq \emptyset$ then 
$K$ and $\phi_S(K')$ are disjoint. Moreover, $K$ and $\phi_S(K')$ are disjoint 
isomorphic connected components of $G-S_0$ where $S_0=S \setminus V(K)$. 
Let $\kappa$ be an isomorphism from $K$ to $\phi_S(K')$.
Notice that the role of $K$ and $\phi_S(K')$ can be interchanged, 
and we can replace $S \cap V(K)$ with $\kappa(S \cap V(K))$ in the solution.
Thus, $S_0 \cup \kappa(S \cap V(K))$ is a solution for $(G',G)$ that is disjoint from $K$.
Since this yields a solution for $(G'-K',G-K)$ as well, this finishes the proof.
\qed 
\end{proof}

%%%%%%%%%%%%%%%%%%%%%%%%%%%%%%%%% Many components in G' %%%%%%%%%%%%%%%%%%%%%%%%%%%%%%%%%%%

{\bf Rule 2. Many components in $G'$.} 
This reduction is possible in the case when $G'$ has at least $4k+1$ components.
Since Rule 1 cannot be applied, none of the components of $G$ is isomorphic to a component of $G'$.
Our aim is to locate $\phi_S(K')$ in $G$ for one of the components $K'$ of $G'$. 
If we find $\phi_S(K')$ then we know that $N_G(\phi_S(K'))$ must be contained in $S$, 
so we can produce a \outputstyle{necessary set} of size $1$ by outputting any of the vertices of $N_G(\phi_S(K'))$.

Given a graph $H$, recall that $\mathcal{M}(H,G)$ denotes the occurrences of $H$ in $G$ as a complete module. 
By Prop. \ref{occurrences}, the elements of $\mathcal{M}(H,G)$ are disjoint subsets of $V(G)$.
We can find $\mathcal{M}(H,G)$ in linear time, using the labeled PQ-tree of $G$
and the characterization of Lemma \ref{module_char}.

Relying on Lemmas \ref{module_indep} and \ref{module_clique}, the algorithm performs the following reduction.
Suppose that $K'_1, K'_2, \dots, K'_{k'}$ are the $k'=4k+1$ largest connected components of $G'$, 
ordered decreasingly by their size, and let $S$ be a solution for $(G',G)$.
As the vertex sets of the connected components of $G'$ are complete modules of $G'$,
the sets $K_i=\phi_S(V(K'_i))$ for $i \in [k']$ are complete modules of $G-S$. 
By definition, these sets are also disjoint and independent.
As a consequence of (1) in Lemma \ref{module_indep}, we get that for at least $k'-4k = 1$ indices $i \in [k']$ 
the set $K_i$ will be a complete module of $G$.
We branch on the choice of $i$ to find such a set $K_i$, resulting in at most $k'$ possibilities. 
Observe that w.l.o.g. we can assume that the subgraph  $G[K_i]$ is the first one 
(according to the given representation of $G$) among the components of $G-S$ isomorphic to $K'_i$.

It remains to describe how we can find $K_i$ in $G$. 
To begin, we suppose now that $K'_i$ is not a clique.
Let us discuss a simplified case first, where we assume that $K'_i$ is not 
contained as an induced subgraph in any of the components $K'_j$ if $j \neq i$.  
Let $\mathcal{M}(K'_i,G)=\{A_1, A_2, \dots \}$, where the sets in $\mathcal{M}(K'_i,G)$ 
are ordered according to their order in the interval representation of $G$.
Let $i^*$ denote the index for which $A_{i^*}$ is the first
element in $\mathcal{M}(K'_i,G)$ that is a complete module in $G-S$ as well.
Since $K'_i$ is not contained in a component of $G'$ having more vertices than $|V(K'_i)|$, 
$G[A_{i^*}]$ must be a connected component of $G-S$. Also, $G[A_{i^*}]$ is isomorphic to $K'_i$, 
and by the definition of $A_{i^*}$, it must be the first such component of $G-S$. 
Thus, we can conclude that $A_{i^*}$ equals $K_i$.
By (2) of Lemma~\ref{module_indep}, there can be at most $4k$ sets in  $\mathcal{M}(K'_i,G)$
that are not complete modules in $G-S$, so we get that $i^* \leq 4k+1=k'$. 
Hence, we can find $K_i$ by guessing $i^*$ and branching into $k'$ directions.

\begin{figure}[t]
\begin{center}
    \psfrag{A}[][]{$A$}
    \psfrag{B}[][]{$B$}
    \psfrag{C}[][]{$C$}
    \psfrag{D}[][]{$D$}
    \psfrag{E}[][]{$K'_i$}
    \psfrag{F}[][]{$F$}
    \psfrag{fA}[][]{$\phi_S(A)$}
    \psfrag{fB}[][]{$\phi_S(B)$}
    \psfrag{fC}[][]{$\phi_S(C)$}
    \psfrag{fD}[][]{$\phi_S(D)$}
    \psfrag{fE}[][]{$K_i$}
    \psfrag{fF}[][]{$\phi_S(F)$}
    \psfrag{G}[][]{$G$}
    \psfrag{G'}[][]{$G'$}
\includegraphics[scale=0.4]{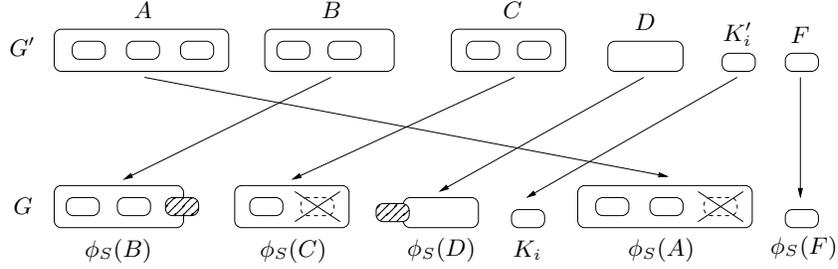}
\end{center}
\caption{
An illustration of Rule 2. 
In this example, the small rectangles denote elements of $\mathcal{M}(K'_i,G')$ and $\mathcal{M}(K'_i,G)$.
Rectangles with a skew pattern are elements of $\mathcal{M}(K'_i,G)$ that are not complete modules of $G-S$.
Crossed rectangles with a dashed border indicate if some set $\phi_S(X)$ is not a complete module of $G$ 
for some $X \in \mathcal{M}(K'_i,G')$. 
In this example, $i=5$, $\delta(1)=0, \delta(2)=\delta(3)=\delta(4)=1$, $|\mathcal{M}'|=4$ and $i^*=6$.
}
\label{fig_components}
\end{figure}

Let us consider now the general case, where some of the components $K'_j$ can 
contain $K'_i$ in $G'$. (We still suppose that $K_i$ is not a clique.) 
For each $j<i$, we define an indicator variable $\delta(j)$
which has value 1 if and only if $K_j$ precedes $K_i$ in $G-S$. 
We guess $\delta(j)$ for each $j \in [i-1]$, which means at most $2^{k'-1}$ possibilities.

Again, let $\mathcal{M}(K'_i,G)=\{A_1, A_2, \dots \}$, where the sets $A_h$
are ordered according to their order in the interval representation of $G$.Let $i^*$ denote the index for which $K_i=A_{i^*}$, and 
let $\mathcal{M}'$ stand for $\bigcup_{j<i, \delta(j)=1} \mathcal{M}(K'_i,K'_j)$,
which is a collection of subsets of $V(G')$, each inducing a subgraph of $G'$ isomorphic to $K'_i$.
As $K_i$ is not a clique, the elements of $\mathcal{M}'$ in $G'$ are disjoint and independent,
so by (1) of Lemma~\ref{module_indep} we get that for at least $|\mathcal{M}'|-4k$ 
sets $A \in \mathcal{M}'$, the set $\phi_S(A)$ will be a complete module of $G$ as well.
As all these sets precede $K_i$ in $G$, we get that $\phi_S(A) \in \{A_1, \dots, A_{i^*-1}\}$ 
holds for at least $|\mathcal{M}'|-4k$ sets $A \in \mathcal{M}'$.
From this,  $i^*-1 \geq |\mathcal{M}'|-4k$ follows.

On the other hand, for all those sets $A \in \{ A_1, \dots, A_{i^*-1}\}$
which are complete modules in $G-S$ as well, 
$\phi_S^{-1}(A)$ must be contained in a component of $G'$ which is larger than $K'_i$.
Here we used again the assumption that $G[K_i]$ is the first one among 
the components of $G-S$ isomorphic to $K_i$. 
Since such an $A$ precedes $K_i$, we obtain $\phi_S^{-1}(A) \in \mathcal{M}'$.
By (2) of Lemma~\ref{module_indep}, there can be at most $4k$ sets among $A_1, \dots, A_{i^*-1}$ 
that are not complete modules in $G-S$, so we get that $\phi_S^{-1}(A) \in \mathcal{M}'$
for at least $i^*-1-4k$ sets $A$ in $\{A_1, \dots, A_{i^*-1}\}$.
This implies $i^*-1 \leq |\mathcal{M}'|+4k$. Altogether, we get the bounds 
$|\mathcal{M}'|-4k+1 \leq i^* \leq |\mathcal{M}'|+4k+1$.
Since $|\mathcal{M}'|$  can be determined in linear time, by branching on the at most $8k+1$
possibilities to choose $i^*$, we can find the vertex set $K_i$.

Now, we suppose that $K'_i$ is a clique. As $K_i$ is a component of $G-S$, 
$|N_G(K_i)| \leq k$, which by Lemma~\ref{h-short_module} implies that $K_i$ must be $k/2$-short.
Using Lemma~\ref{module_clique}, we can find $K_i$ in a similar manner to the previous case.
We denote by $\mathcal{N}(H,G)$ the occurrences of a graph $H$ in $G$ as a $k/2$-short complete module.
Analogously to the previous case, let $\mathcal{N}(K'_i,G)=\{B_1, B_2, \dots \}$, 
where the sets in $\mathcal{N}(K'_i,G)$ are ordered according to their 
order in the fixed representation of $G$. 
We also let $K_i=B_{i^*}$ and $\mathcal{N}'=\bigcup_{j<i, \delta(j)=1} \mathcal{N}(K'_i,K'_j)$.
Now, using Lemma~\ref{module_clique} just as in the reasoning above, we get
the bounds $|\mathcal{N}'|-k(3k/2+5) \leq i^*-1 \leq |\mathcal{N}'|+k(2k+3)$.
%As $K_i$ is a clique, the elements of $\mathcal{N}^*$ in $G'$ are disjoint and independent,
%By Lemma~\ref{module_clique}, we get that for at least $|\mathcal{N}'|-k(3k/2+5)$ sets $B \in \mathcal{N}'$, 
%the set $\phi(B)$ will be a $k/2$-short complete modules of $G$ as well.
%As these sets all precede $K_i$ in the given representation of $G$, we get $i^* \geq |\mathcal{N}'|-k(3k/2+5)$.
%
%On the other hand, all those occurrences $B$ of $K'_i$ in $G$ which are $k/2$-short 
%complete modules in $G-S$ and precede $K_i$ in $G-S$, 
%$\phi^{-1}(B)$ must be contained in one of the components of $G'$ which are larger than $K'_i$.
%Since $B$ precedes $K_i$, we obtain that $\phi^{-1}(B) \in \mathcal{N}'$.
%Since there can be at most $k(2k+3)$ elements in $\mathcal{N}(K'_i,G)$ that are not complete modules in $G-S$, 
%Using Lemma~\ref{module_indep} again, we get $i^* \leq |\mathcal{N}'|+k(2k+3)$.
Again, $|\mathcal{N}'|$  can be determined in linear time, so by branching on the at most $k(7k/2+8)+1$
possibilities to choose $i^*$, we can find the vertex set $K_i$.

Since Rule~1 cannot be applied, none of the components of $G$ can be isomorphic 
to a component of $G'$, hence $S_i=N_G(K_i)$ is not empty.
Clearly $S_i \subseteq S$, so we get that $\{s\}$ is a necessary set for any $s \in S_i$.
The total number of possible branches in this reduction is at most 
$(4k+1)2^{4k}(k(7k/2+8)+1) = 2^{O(k)}$.

%%%%%%%%%%%%%%%%%%%%%%%%%%%%%%%%% Disconnected G %%%%%%%%%%%%%%%%%%%%%%%%%%%%%%%%%%%

{\bf Rule 3. Disconnected $G$.} Here we give a reduction for the case when $G$ is not connected,
but the previous reductions cannot be performed.
First, observe that each component of $G$ contains at least one vertex from $S$, as none of
them is isomorphic  to a component of $G'$.
Thus, if $G$ has more than $k$ components then there cannot exist a solution of size $k$,
so we can \outputstyle{reject}. Otherwise, let us fix an arbitrary component $K$ of $G$.
We branch on the choice of those components of $G'$ whose vertices are in $\phi_S^{-1}(K-S)$,
for some fixed solution $S$. Let the union of these components be $G'_K$.
Note that guessing $G'_K$ yields at most $2^{4k}$ possibilities, since $G'$ has at most $4k$ components.
By our assumptions, $1 \leq k' <k$ holds for the parameter $k'$ of the instance $(G'_K,K)$,
so $(G'_K,K)$ is clearly an \outputstyle{independent subproblem} of $(G',G)$.

%%%%%%%%%%%%%%%%%%%%%%%%%%%%%%%%% Universal vertices in G %%%%%%%%%%%%%%%%%%%%%%%%%%%%%%%%%%%

{\bf Rule 4. Universal vertex in $G$.} A vertex $x$ is \emph{universal} in $G$, if $N_G(x)=V(G-x)$.
Such vertices imply a simple reduction by Lemma \ref{universal} which allows $\mathcal{A}$
to output either a \outputstyle{necessary set} of size 1 or a \outputstyle{reduced input} of $(G',G)$.

\begin{mylemma}
\label{universal}
Let $x$ be universal in $G$.
If there is no universal vertex in $G'$, then $\{x\}$ is a necessary set for $(G',G)$.
If $x'$ is universal in $G'$, then $(G'-x',G-x)$ is a reduced input of $(G',G)$.
\end{mylemma}

\begin{proof}
Clearly, if $x$ is universal in $G$ and $x \notin S$ for a solution $S$,
then it remains universal in $G-S$. Thus, if no vertex is universal in $G'$, then $x \in S$ must hold.

Suppose $x'$ is universal in $G'$, and $S$ is an arbitrary solution.
Clearly, if $x \notin S$, then $x$ and $y=\phi_S(x')$ are both universal in $G-S$,
so if $x\neq y$ then we can swap the role of $x$ and $y$ such that $\phi_S$ maps $x'$ to $x$.
Now, if $x \in S$ then $S'=S \cup \{y\} \setminus \{x\}$ is a solution in which the
isomorphism from $G'$ to $G-S'$ can map $x'$ to $x$. 
This implies that $(G'-x',G-x)$ is a reduced input of $(G',G)$.
\qed
\end{proof}

%%%%%%%%%%%%%%%%%%%%%%%%%%%%%%%% Q-P case: %%%%%%%%%%%%%%%%%%%%%%%%%%%%%%%%%%%%%%%%%%%%%%%%%%%

{\bf Rule 5. Disconnected $G'$.} Suppose that none of the previous 
reductions can be applied, and $G'$ is disconnected.
This means that $G$ must be connected, and $G'$ has at most $4k$ components.
Let $S$ be a solution. For each component $K'$ in $G'$, let 
$I(K')$ be the union of the intervals representing $\phi_S(K')$ in the fixed representation of $G$, 
i.e. let $I(K')=[\phi_S(V(K'))^{\mathrm{left}},\phi_S(V(K'))^{\mathrm{right}}]$.
Since the components of $G'$ are connected and independent, the intervals 
$I(K'_1)$ and $I(K'_2)$ are disjoint for two different components $K'_1$ and $K'_2$ of $G'$.

Let $Q$ be the component of $G'$ such that $I(Q)$ is the first among the 
intervals $\{ I(K') \mid K'$ is a component of $G'\}$.
Clearly, if $x^{\mathrm{right}} \leq \phi_S(V(Q))^{\mathrm{right}}$ for some vertex $x \in V(G)$,
then either $x \in S$ or $x \in \phi_S(Q)$,
thus the number $s_Q$ of such vertices is at least $|V(Q)|$ but at most $|V(Q)|+k$.
Therefore, we first guess $Q$, and then we guess the value of $s_Q$,
which yields at most $4k(k+1)$ possibilities. 
Now, ordering the vertices of $G$ such that $x$ precedes $y$ if $x^{\mathrm{right}}<y^{\mathrm{right}}$ and
putting the first $s_Q$ vertices in this ordering into a set $B$, 
we get $\phi_S(Q) \subseteq B \subseteq \phi_S(Q) \cup S$. Since $G$ is connected, there must exist an
edge $e=xy$ running between $B$ and $V(G) \setminus B$. 
Clearly, at least one endpoint of $e$ must be in $S$,
thus we can output the \outputstyle{necessary set} $\{x,y\}$.

%%%%%%%%%%%%%%%%%%%%%%%%%%%%%%%%% Universal vertices in G' %%%%%%%%%%%%%%%%%%%%%%%%%%%%%%%%%%%

{\bf Rule 6. Universal vertex in $G'$.} Suppose that some vertex $x'$ is universal in $G'$.
Let $a$ and $b$ be vertices of $G$ defined such that 
$a^{\mathrm{right}}=\min \{ x^{\mathrm{right}} \mid x \in V(G)\}$ and
$b^{\mathrm{left}}=\max \{ x^{\mathrm{left}} \mid x \in V(G)\}$. 
As there is no universal vertex in $G$, we know
that $x^{\mathrm{left}}>a^{\mathrm{right}}$ or $x^{\mathrm{right}}<b^{\mathrm{left}}$ 
for each $x \in V(G)$, i.e. no vertex in $G$ is adjacent to both $a$ and $b$. 
As $\phi_S(x')$ is universal in $G-S$ for any fixed solution $S$, 
we get that $\{a,b\}$ is a \outputstyle{necessary set}.

%%%%%%%%%%%%%%%%%%%%%%%%%%%%%%%%%% End of reductions %%%%%%%%%%%%%%%%%%%%%%%%%%%%%%%%%%%%%%%%%%

\section{The Q-Q case}
\label{qq_case}

From now on, we assume that none of the reductions given in Sect. \ref{reductions} can be applied.
Thus, $G$ and $G'$ are connected, and none of them contains universal vertices,
so in particular, none of them can be a clique. This implies that if $r$ and $r'$ is the root of $T$
and $T'$, respectively, then both $r$ and $r'$ are Q-nodes.
Let $m$ and $m'$ denote the number of the children of $r$ and $r'$, respectively. 
When indexing elements of $[m]$ and $[m']$, we will try to use $i$ and $j$, respectively,
whenever it makes sense.
Let $x_i$ and $x'_j$ denote the $i$-th and $j$-th child of $r$ and $r'$, respectively, and let
$X_i=R^{-1}(T_{x_i})$ and $X'_j=R^{-1}(T'_{x'_j})$, for all $i \in [m]$ and $j \in [m']$.

Let us call a solution $S$ \emph{local}, if there is an $i \in [m]$ such that
$S \supseteq V(G) \setminus N_G[X_i]$, i.e. $S$ contains every vertex of $G$ except for 
the closed neighborhood of some $X_i$.
Suppose that $S$ is a solution that is not local, and $\phi_S$ is an isomorphism 
from $G'$ to $G-S$. The following definitions try to give a bound on those indices $i$ 
in $[m]$ which somehow contribute to $\phi_S(X'_j)$ for some $j \in [m']$.
For an index $j \in [m']$, let:
\begin{eqnarray*}
\alpha_S(j) &=& \min \{ Q_{r}^{\mathrm{left}}(\phi_S(v)) \mid v \in X'_j \cup M^+_{r'}(j) \} \\
\beta_S(j) &=& \max \{ Q_{r}^{\mathrm{right}}(\phi_S(v)) \mid v \in X'_j \cup M^-_{r'}(j) \} .
\end{eqnarray*}
Observe that by
$$ 
\min \{ Q_{r}^{\mathrm{left}}(\phi_S(v)) \mid v \in X'_j \cup M^+_{r'}(j) \} 
\leq \min \{ Q_{r}^{\mathrm{left}}(\phi_S(v)) \mid v \in X'_j \} 
$$
$$
\leq \max \{ Q_{r}^{\mathrm{right}}(\phi_S(v)) \mid v \in X'_j \} \leq 
\max \{ Q_{r}^{\mathrm{right}}(\phi_S(v)) \mid v \in X'_j \cup M^-_{r'}(j) \}
$$
we obtain that $\alpha_S(j) \leq \beta_S(j)$ holds for any $j \in [m']$.
We let $I_S(j)$ to be the block $[\alpha_S(j),\beta_S(j)]$.

The following lemma summarizes some useful observations.
\begin{mylemma}
\label{I_S_prop} 
Suppose that $S$ is a solution for $(G',G)$ that is not local. 
Then either all of the following statements hold, 
or all of them hold after reversing the children of $r'$: \\
(1) $Q_{r'}^{\mathrm{dir_1}}(v) < Q_{r'}^{\mathrm{dir_2}}(w)$ for some $v,w \in V(G')$ and 
$\mathrm{dir_1},\mathrm{dir_2} \in \{\mathrm{left,right}\}$ implies that
$Q_{r}^{\mathrm{dir_1}}(\phi_S(v)) < Q_{r}^{\mathrm{dir_2}}(\phi_S(w))$ holds as well. \\
(2) For any $j_1 <j_2$ in $[m']$, the block $I_S(j_1)$ precedes $I_S(j_2)$. \\
(3) If $m=m'$, then $I_S(j)=[j,j]$ for each $j \in [m']$. \\
(4) If $i \in [m]$ then $X_i \setminus S$ is contained in $\phi_S(X'_j)$ for some $j \in [m']$. \\
(5) $\phi_S(R^{-1}(r')) \subseteq R^{-1}(r)$.
\end{mylemma}

\begin{proof}
Let the sets $A_1, \dots, A_s$ be disjoint subsets of some set $\{a_i | i \in [r]\}$. 
We say that the series  $a_1, \dots, a_r$,  \emph{respects} the series $A_1, \dots, A_s$
if for any $p$, the elements of $A_p$ precede the elements of $A_{p+1}$ in $a_1, \dots, a_r$.

For some $j \in [m']$, let $C'_j$ be the set of maximal cliques of $G'$ contained in $F(T'_{x'_j})$. 
As $T'$ represents $G'$ and its root is a Q-node, any consecutive ordering of the maximal cliques of $G'$ 
respects either $C'_1, C'_2, \dots,  C'_{m'}$ or $C'_{m'}, C'_{m'-1}, \dots,  C'_1$.

Let $C_j$ be the set of maximal cliques of $G$ that contain $\phi_S(K)$ for some $K \in C'_j$. 
Clearly, the sets $C_j$ ($j \in [m']$) are disjoint by the maximality of the cliques in $F(T')$.
The interval representation of $G$ yields an interval representation of $G-S$ and hence of $G'$, 
which implies that any consecutive ordering of the cliques in $F(T)$ must respect 
either $C_1, \dots,  C_{m'}$ or $C_{m'}, \dots,  C_1$.
This implies that if $j_1 \neq j_2$, then the deepest node in $T$ whose frontier contains 
every clique in $C_{j_1} \cup C_{j_2}$ must be some unique Q-node $q$, the same for each pair $(j_1,j_2)$.
Note also that if $q$ is contained in $T_{x_i}$ for some $i \in [m]$, then 
every vertex of $\phi_s(V(G'))$ must be contained in some element of $F(T_{x_i})$. 
This yields that $S \supseteq V(G) \setminus N_G[X_i]$, contradicting to the assumption that $S$ is not local.
Hence $q=r$. The definition of $q$ shows that by possibly reversing the children of $r$ in $T$, 
we can find disjoint blocks $B_1, B_2, \dots, B_{m'}$ 
following each other in this order in [1,m] such that for each $j \in [m']$
every clique of $C_j$ is contained in $\bigcup_{i \in B_j} F(T_{x_i})$.

This observation can be easily seen to imply claim (1), by simply recalling 
what $Q_{r'}^{\mathrm{dir_1}}(v) < Q_{r'}^{\mathrm{dir_2}}(w)$ means by definition for some 
$v,w \in V(G')$ and $\mathrm{dir_1},\mathrm{dir_2} \in \{\mathrm{left,right}\}$, 
considering the maximal cliques of $G'$.

Now, it is easy to prove (2), (3), (4), and (5), using the assumption that (1) holds 
(which can indeed be achieved by possibly reversing the children of $r'$).
To prove (2), observe that for any $j \in [m']$, we have that
$Q_{r'}^{\mathrm{left}}(v)=j$ for each $v \in X'_j \cup M^+_{r'}(j)$ and 
$Q_{r'}^{\mathrm{right}}(v)=j$ for each $v \in X'_j \cup M^-_{r'}(j)$.
Thus, for any $j_1<j_2$ in $[m']$, $v \in X'_{j_1} \cup M^-_{r'}(j_1)$, 
and $w \in X'_{j_2} \cup M^+_{r'}(j_2)$ we obtain
$Q_{r'}^{\mathrm{right}}(v) < Q_{r'}^{\mathrm{left}}(w)$, which immediately
implies $Q_{r}^{\mathrm{right}}(\phi_S(v)) < Q_{r}^{\mathrm{left}}(\phi_S(w))$ by using the
assumption that (1)  holds. 
By definition, this means $\beta_S(j_1)<\alpha_S(j_2)$, from which (2) follows.

Note that (3) is directly implied by (2).

To see (4), consider an $x \in X_i \setminus S$ and let $x'$ be the vertex in $G'$ for which $x=\phi_S(x')$.
Since $Q_r^{\mathrm{left}}(x)=Q_r^{\mathrm{right}}(x)=i$, by (1)
we must have $Q_{r'}^{\mathrm{left}}(x')=Q_{r'}^{\mathrm{right}}(x')$ as well.
Thus, $x' \in X'_j$ for some $j \in [m']$. Now, assuming that $\phi_S(y')$
is also in $X_i$ but $y' \notin X'_j$, we obtain 
that either $Q_{r'}^{\mathrm{right}}(x') < Q_{r'}^{\mathrm{left}}(y')$ or 
$Q_{r'}^{\mathrm{right}}(y')<Q_{r'}^{\mathrm{left}}(x')$.  But using (1), these both 
contradict $\phi_S(y') \in X_i$.

Finally, (5) follows immediately from (4).
\qed
\end{proof}

Using Lemma~\ref{I_S_prop}, we can handle an easy case when no branching is needed, 
and a reduced input can be constructed. Let $L(r)$ and $L'(r')$ denote the 
labels of $r$ and $r'$ in $T$ and $T'$, respectively.

\begin{mylemma}
\label{qq_subtree}
If $m=m'$, $L(r)=L'(r')$ and there is an $i \in [m]$ such that
$G[X_j]$ is isomorphic to $G'[X'_j]$ for all $j \neq i$, then $(G'[X'_i],G[X_i])$ 
is a reduced input of $(G',G)$.
\end{mylemma}

\begin{proof}
First, we show that a solution $S$ for $(G'[X'_i],G[X_i])$ is a solution for $(G',G)$.
Clearly, the conditions of the lemma imply that there is
an isomorphism from $G'-X'_i$ to $G-X_i$ mapping $N_{G'}(X'_i)$ to $N_G(X_i)$.
This isomorphism can be extended to map $X'_i$ to $X_i \setminus S$, 
since $S$ is a solution for $(G'[X'_i],G[X_i])$.
In the other direction, suppose that $S$ is a solution for $(G',G)$. 
Note that $S$ need not be a solution for $(G'[X'_i],G[X_i])$, as $S$ may contain vertices not in $X_i$.
However, to prove the lemma it suffices to show that a solution exists for $(G'[X'_i],G[X_i])$, 
meaning that $G'[X'_i]$ is isomorphic to an induced subgraph of $G[X_i]$.

Let us assume first that $S$ is not local.
Suppose $S \cap X_j \neq \emptyset$ for some $j \neq i$.
Since $G[X_j]$ is isomorphic to $G'[X'_j]$, 
we have $|\phi_S(X'_j)| = |X_j| > |X_j \setminus S|$.
By Lemma~\ref{I_S_prop}, this implies that $I_S(j)=[j,j]$ becomes true only after 
reversing the children of $r'$, meaning that $I_S(h)=[m-h+1,m-h+1]$ for each $h \in [m]$ 
(according to the PQ-tree $T'$). 
In particular, this means that $\phi_S(X'_i) \subseteq X_{m-i+1}$. 
Hence, $G'[X'_i]$ is isomorphic to an induced subgraph of $G[X_{m-i+1}]$. 
Since $G[X_{m-i+1}]$ is isomorphic to $G'[X'_{m-i+1}]$, $G'[X'_i]$ is also isomorphic 
to an induced subgraph of $G'[X'_{m-i+1}]$. 
From this, by $\phi_S(X'_{m-i+1}) \subseteq X_i$
we get that $G'[X'_i]$ is isomorphic to an induced subgraph of $G[X_i]$.

Now, suppose that $S$ is a local solution, and 
$S \supseteq V(G) \setminus N_G[X_h]$ for some $h \in [m]$.
Since each vertex of $R^{-1}(r) \setminus S$ must be adjacent to 
every vertex in $N_G[X_h]$, such vertices would be universal in $G-S$. 
Thus, as $G'$ contains  no universal vertices, $S \supseteq R^{-1}(r)$ follows.
Hence, we get $\phi_S(V(G')) \subseteq X_h$ and thus $|X_h| \geq |V(G')| > |X'_h|$, implying $h=i$.
% Let $S_1=S \cap X_i$ and $S_2 = S \setminus S_1 = V(G) \setminus X_i$. 
But $\phi_S(V(G')) \subseteq X_i$ clearly implies that $G'$ and therefore also $G'[X'_i]$ must be 
isomorphic to an induced subgraph of $G[X_i]$.
This finishes the proof.
\qed
\end{proof}

Observe that it can be tested in linear time whether the conditions of Lemma \ref{qq_subtree} hold
after possibly reversing the children of $r'$. If this is the case, algorithm $\mathcal{A}$ proceeds with
the \outputstyle{reduced input} guaranteed by Lemma \ref{qq_subtree}.
Otherwise it branches into a few directions as follows. In each of these branches, 
$\mathcal{A}$ will output either a necessary set of size at most $2k+1$, or an independent subproblem of $(G',G)$.

In the first branch, it assumes that the solution $S$ is local. 
In this case, given any two vertices $a \in X_1$ and $b \in X_m$, 
a solution must include at least one of $a$ and $b$. (Note that $X_1$ and $X_m$ cannot be empty.)
Thus, $\mathcal{A}$ outputs the \outputstyle{necessary set} $\{a,b\}$. 
In all other branches, we assume that $S$ is not local.
Algorithm $\mathcal{A}$ branches into two more directions, according to whether the children of $r'$ 
have to be reversed to achieve the properties of Lemma~\ref{I_S_prop}.
Thus, in the followings we may assume that these properties hold.

First, observe that Lemma~\ref{I_S_prop} implies that $m \geq m'$ must hold,
so otherwise the algorithm can \outputstyle{reject}. 
First, we examine the case $m=m'$, and then we deal with the case $m>m'$ in Section~\ref{sect_fragments}.
Observe that Lemma~\ref{I_S_prop} also implies that $\phi_S$ must map every
$L_{r'}(a,b)$ to a vertex in $L_r(a,b)$, for any block $[a,b]$ in $[1,m]$.
So, if there is a block $[a,b]$ such that $|L_r(a,b)|<|L_{r'}(a,b)|$, then the
algorithm has to \outputstyle{reject}. If the converse is true, i.e. $|L_r(a,b)|>|L_{r'}(a,b)|$
for some block $[a,b]$, then some vertex $v \in L_r(a,b)$ must be included in $S$.
Since each vertex in $L_{r}(a,b)$ has the same neighborhood, 
the algorithm can choose $v$ arbitrarily from $L_r(a,b)$ and 
output the \outputstyle{necessary set} $\{ v\}$.

If none of these cases happen, then $L(r)=L'(r')$, so as the conditions of 
Lemma~\ref{qq_subtree} do not hold, there must exist two indices $i_1 \neq i_2 \in [m]$, 
such that $G[X_{i_1}]$ and $G[X_{i_2}]$ is not isomorphic to $G'[X'_{i_1}]$ and $G'[X'_{i_2}]$, respectively. 
As $L(r)=L'(r')$, by Lemma~\ref{I_S_prop} we get $\phi_S(R^{-1}(r')) =R^{-1}(r)$ 
and $\phi_S(X'_h)=X_h \setminus S$ for each $h \in [m]$.
Thus, it is easy to see that for each $h \in [m]$, the set $S \cap X_h$ yields a solution 
for the instance $(G'[X'_h],G[X_h])$, 
and conversely, if $(G',G)$ is solvable then
any solution for $(G'[X'_h],G[X_h])$ can be extended to a solution for $(G',G)$.
Now, using $X_{i_1} \cap S \neq \emptyset$ and $X_{i_2} \cap S\neq \emptyset$, 
we know that the parameter of the instance
$(G'[X'_{i_1}],G[X_{i_1}])$ must be at least $1$ but at most $k-1$.
If this indeed holds, then $\mathcal{A}$ outputs $(G'[X'_{i_1}],G[X_{i_1}])$ as an \outputstyle{independent subproblem},
otherwise it \outputstyle{rejects} the instance. 

\subsection{Identifying fragments for the case $m>m'$.}
\label{sect_fragments}
The rest of the paper deals with the case where $S$ is not local, and $m>m'$. 
In this case, we will try to determine $I_S(j)$ for each $j \in [m']$. 
To do this, $\mathcal{A}$ will branch several times on determining $I_S(j)$ for some $j \in [m']$.
Proposition~\ref{I_S_bounds} helps us to bound the number of resulting branches.
To state this proposition, we use the following notation: given some block $[i_1,i_2]$ in $[1,m]$, 
let $W_r(i_1,i_2)$ contain those vertices $v$ for which $Q_r(v)$ is contained in $[i_1,i_2]$.
Let also $w_r(i_1,i_2) = |W_r(i_1,i_2)|$. We define $W_{r'}(j_1,j_2)$ and $w_{r'}(j_1,j_2)$ for some 
block $[j_1,j_2]$ in $[1,m']$ similarly. 
Using Lemma~\ref{I_S_prop} and the definition of $\alpha_S(j)$ and $\beta_S(j)$, 
it is easy to prove the following:

\begin{myprop} 
\label{I_S_bounds} Suppose that $S$ is a solution that is not local, and the properties of Lemma~\ref{I_S_prop} hold. 
This implies the followings.\\
(1)
$\phi_S(W_{r'}(j_1,j_2)) = W_r(\alpha_S(j_1),\beta_S(j_2)) \setminus S$
for every block $[j_1,j_2]$ in $[1,m']$.
\\
(2)
$w_{r'}(1,j) \leq w_r(1,\beta_S(j)) \leq w_{r'}(1,j)+k$ and
$w_{r'}(j,m') \leq w_r(\alpha_S(j),m) \leq w_{r'}(j,m')+k$
hold for every $j \in [m']$.
\end{myprop}

Note that $w_r(1,i)<w_r(1,i+1)$ for every $i \in [m-1]$, even if $X_{i+1} = \emptyset$ 
(as in this case $M^-_r(i+1) \neq \emptyset$), and similarly we get 
$w_r(i,m)>w_r(i+1,m)$ for every $i \in [m-1]$.
Therefore, the bounds of Prop.~\ref{I_S_bounds} yield at most $(k+1)^2$
possibilities for choosing $[\alpha_S(j),\beta_S(j)]$, for some $j \in [m']$.

Since determining $I_S(j)$ for each $j \in [m']$ using Prop.~\ref{I_S_bounds}
would result in too many branches, we need some other tools.
Hence, we introduce a structure called fragmentation that can be used to 
``approximate'' the sets $I_S(j)$ for each $j \in [m']$.
By iteratively refining the fragmentation, we can get closer and closer to actually determine these sets.
Given a set of disjoint blocks  $\{ [a'_h,b'_h] \mid h \in [f]\}$ in $[1,m']$
and a corresponding set of disjoint blocks $\{ [a_h,b_h] \mid h \in [f]\}$ in $[1,m]$
with $F_h$ denoting the pair $([a'_h,b'_h],[a_h,b_h])$,
the set $\{F_h \mid h \in [f]\}$ is a \emph{fragmentation} for $(T,T',S)$, if 
\begin{itemize}
\item $a_h \leq \alpha_S(a'_h)$ and $\beta_S(b'_h) \leq b_h$ for each $h \in [f]$, and
\item $a'_{h+1} = b'_h+1$ and $a_{h+1}=b_h+1$ for each $h \in [f-1]$.
\end{itemize}
We will call the element $F_h$ for some $h \in [f]$ a \emph{fragment}.
We define $\sigma(F_h)=(b_h-a_h)-(b'_h-a'_h)$
and $\delta(F_h)=a_h-a'_h$, which are both clearly non-negative integers.
Note that $\delta(F_{h+1})=\delta(F_h)+ \sigma(F_h)$ holds for each $h \in [f-1]$. 
We say that some $j \in [m']$ is \emph{contained} in the fragment $F_h$, if $a'_h \leq j \leq b'_h$.
In this case, we write $\delta(j)=\delta(F_h)$ and $\sigma(j)=\sigma(F_h)$.
We will say that a fragment $F$ is \emph{trivial} if $\sigma(F_h)=0$, and \emph{non-trivial} otherwise.
We also call an index in $[m']$ trivial (or non-trivial) in a fragmentation, 
if the fragment containing it is trivial (or non-trivial, respectively).
An \emph{annotated fragmentation} for $(T,T',S)$ is a pair $(\mathcal{F},U)$ formed by a 
fragmentation $\mathcal{F}$ for $(T,T',S)$ and a set $U \subseteq [m']$ 
such that each $j \in U$ is trivial in $\mathcal{F}$.
We say that the trivial indices contained in $U$ are \emph{important}.
See Fig.~\ref{fig_fragments} for an illustration.

\begin{figure}[t]
\begin{center}
    \psfrag{0}[][]{$0$}
    \psfrag{1}[][]{$1$}
    \psfrag{2}[][]{$2$}
    \psfrag{3}[][]{$3$}
    \psfrag{4}[][]{$4$}
    \psfrag{5}[][]{$5$}
    \psfrag{6}[][]{$6$}
    \psfrag{7}[][]{$7$}
    \psfrag{8}[][]{$8$}
    \psfrag{9}[][]{$9$}
    \psfrag{10}[][]{$10$}
    \psfrag{r}[][]{$r$}
    \psfrag{r'}[][]{$r'$}
    \psfrag{F1}[][]{$F_1$}
    \psfrag{F2}[][]{$F_2$}
    \psfrag{F3}[][]{$F_3$}
    \psfrag{F4}[][]{$F_4$}
    \psfrag{F5}[][]{$F_5$}
    \psfrag{F}[][]{$F$}
    \psfrag{dF}[b][b]{$\delta(F)$}
    \psfrag{sF}[b][b]{$\sigma(F)$}
%    \psfrag{D1}[][]{$\delta(F_1)=0, \sigma(F_1)=0$}
%    \psfrag{D2}[][]{$\delta(F_2)=0, \sigma(F_2)=1$}
%    \psfrag{D3}[][]{$\delta(F_3)=1, \sigma(F_3)=0$}
%    \psfrag{D4}[][]{$\delta(F_4)=1, \sigma(F_4)=2$}
%    \psfrag{D5}[][]{$\delta(F_5)=3, \sigma(F_5)=2$}
\includegraphics[scale=0.4]{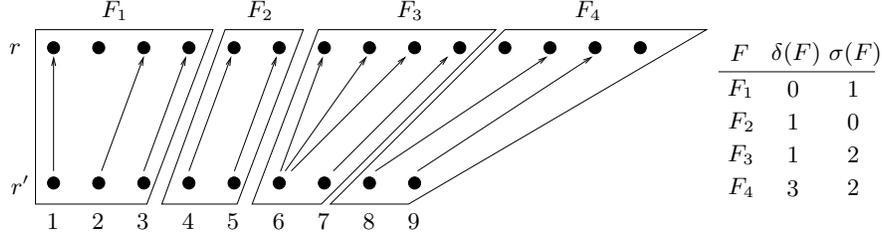}
\end{center}
\caption{
Illustration of a fragmentation containing four fragments. 
An arrow leading from $j$ to $i$ indicates $i \in I_S(j)$.
There are three non-trivial fragments: $F_1, F_3, F_4$. 
The indices $1,4,5$ are left-aligned, $2,3,4,5,7$ are right-aligned, $6$ is wide and $8,9$ are skew. 
}
\label{fig_fragments}
\end{figure}

Let us suppose that we are given a fragmentation $\mathcal{F}$ for $(T,T',S)$, 
and some $j \in [m']$ contained in a fragment $F \in \mathcal{F}$.
We will use the notation $\leftind{j} = j + \delta(F)$ and $\rightind{j} = j + \delta(F)+\sigma(F)$.
Also, we will write $B^+_r(i)=M_{r}^+(i) \cup X_i$ and $B^+_{r'}(j)=M_{r'}^+(j) \cup X'_j$.
For some block $[i_1,i_2]$ in $[1,m]$ let $B^+_{r}(i_1,i_2) = \bigcup_{h \in [i_1,i_2]} B^+_{r}(h)$, and
we define $B^+_{r'}(j_1,j_2)$ for some block $[j_1,j_2]$ in $[1,m']$ analogously. 
Proposition ~\ref{limit_indices} is easy to prove, using Lemma~\ref{I_S_prop}. 

\begin{myprop}
\label{limit_indices} 
For each $j \in [m']$, the followings hold:  \\
(i) 
$\leftind{j} \leq \alpha_S(j) \leq \beta_S(j) \leq \rightind{j}$. \\
(ii) 
$\phi_S(M_{r'}^+(j)) \subseteq \bigcup_{h \in I_S(j)} M_r^+(h) $ and
$\phi_S(M_{r'}^-(j)) \subseteq \bigcup_{h \in I_S(j)} M_r^-(h) $. \\
(iii)
if $I_S(j)=[i,i]$, then 
$M_r^+(i) \setminus S = \phi_S(M_{r'}^+(j))$ and
$M_r^-(i) \setminus S = \phi_S(M_{r'}^-(j))$.\\
(iv) 
if  $j<j'$, $\alpha_S(j)=i$ and $\beta_S(j')=i'$, then 
$\phi_S (B^+_{r'}(j,j')) = B^+_r(i,i') \setminus S$.
\end{myprop}

We will classify the index $j$ as follows:
\begin{itemize}
\item 
If $|I_S(j)|>1$, then $j$ is \emph{wide}.
\item 
If $I_S(j)=[\leftind{j},\leftind{j}]$, then $j$ is \emph{left-aligned}.
\item 
If $I_S(j)=[\rightind{j},\rightind{j}]$, then $j$ is \emph{right-aligned}.
\item 
If $I_S(j)=[i,i]$ such that $\leftind{j} < i < \rightind{j}$, then $j$ is \emph{skew}.
\end{itemize}
If $F$ is trivial, then by Prop.~\ref{limit_indices}, 
only $\alpha_S(j)=\beta_S(j)= \leftind{j}= \rightind{j} $ is possible. 
Thus, each trivial index must be both left- and right-aligned. 

Lemma~\ref{max_fragments} shows that a solvable instance can only contain at most $2k$ non-trivial fragments. 
Thus, if a given fragmentation contains more than $2k$ non-trivial fragments, then
$\mathcal{A}$ can correctly reject, as such a fragmentation does not correspond to any solution.

\begin{mylemma} 
\label{max_fragments}
Any fragmentation $\mathcal{F}$ for $(T,T',S)$ can have at most $2k$ non-trivial fragments.
\end{mylemma}

\begin{proof}
We will show that every non-trivial fragment of $\mathcal{F}$ contains an index $j$ for which 
$X_j \cup M^+_r(j) \cup M^-_r(j)$ contains a vertex of $S$. Since any $s \in S$
can be contained in at most two sets of this form, 
this proves that there can be at most $2k$ non-trivial fragments  in $\mathcal{F}$.

Observe that if $F=([a',b'],[a,b])$ is a non-trivial fragment in $\mathcal{F}$, 
then either $|I_S(j)|>1$ must hold for some $j$ in $[a',b']$, 
or some $i$ in $[a,b]$ is not contained in any of the blocks $\{ I_S(h) \mid a' \leq h \leq b'\}$.
In the latter case we have $X_i \cup M^+_r(i) \cup M^-_r(i) \subseteq S$.
By Prop.~\ref{notempty}, $X_i \cup M^+_r(i) \cup M^-_r(i) \neq \emptyset$, so it indeed contains a vertex of $S$.
Now, suppose that the former case holds, and $j$ is wide.
Reversing the children of $r$ between $\alpha_S(j)$ and $\beta_S(j)$ 
cannot result in a PQ-tree representing $G$, hence there must be a vertex 
$z  \in R^{-1}(r)$ such that $Q_r(z)$ properly intersects $I_S(j)$.

Let $Q_r(z) =[z_1,z_2]$. We assume $z_1 < \alpha_S(j) \leq z_2 < \beta_S(j)$, 
as the case $\alpha_S(j)< z_1 \leq \beta_S(j) < z_2$ can be handled analogously. 
First, if $z \in S$, then $M^-_r(z_2)$ contains a vertex of $S$, so we are done.
Otherwise, $z \notin S$ implies that $z$ must be contained in $\phi_S(M^-_{r'}(j))$, from which 
$X_{\beta_S(j)} \cup M^+_r(\beta_S(j)) \subseteq S$ follows.
Again, $X_{\beta_S(j)} \cup M^+_r(\beta_S(j)) \neq \emptyset$ by Prop.~\ref{notempty}, 
so in this case we obtain $\emptyset \neq X_{\beta_S(j)} \cup M^+_r(\beta_S(j)) \subseteq S$.
This proves the lemma.
\qed
\end{proof}

Before describing the remaining steps of the algorithm, we need some additional notation.

Let $T^{\mathrm{rev}}$ and $T'^{\mathrm{rev}}$ denote the labeled PQ-trees obtained by 
reversing the children of $r$ and $r'$, respectively.
We write $j^{\mathrm{rev}}$ for the index $m'-j+1$ corresponding to $j$ in $T^{\mathrm{rev}}$, 
and we also let $X^{\mathrm{rev}} = \{ j^{\mathrm{rev}} \mid j \in X\}$ for any set $X \subseteq [m']$.
For a fragment $F=([a',b'],[a,b])$ we let 
$F^{\mathrm{rev}}=([b'^{\mathrm{rev}},a'^{\mathrm{rev}}],[b^{\mathrm{rev}},a^{\mathrm{rev}}])$, 
so a fragmentation $\mathcal{F}$ for $(T,T',S)$ clearly yields a fragmentation 
$\mathcal{F}^{\mathrm{rev}}=\{F^{\mathrm{rev}} | F \in \mathcal{F}\}$ for $(T^{\mathrm{rev}},T'^{\mathrm{rev}},S)$.
Note that if $j$ is left-aligned (right-aligned) in $\mathcal{F}$, 
then the index $j^{\mathrm{rev}}$ is right-aligned (left-aligned, resp.) 
in $\mathcal{F}^{\mathrm{rev}}$.

Given some $i \in [m]$, let us order the vertices $v$ in $M^+_r(i)$ increasingly according to $Q_r^{\mathrm{right}}(v)$.
Similarly, we order vertices $v$ in $M^-_r(i)$ increasingly according to $Q_r^{\mathrm{left}}(v)$. 
In both cases, we break ties arbitrarily.
Also, we order vertices of $M^+_{r'}(j)$ and $M^-_{r'}(j)$ in $T'$ the same way for some $j \in [m']$. 
Now, we construct the sets $\leftind{P}^+(j)$ and $\leftind{P}^-(j)$, 
both containing pairs of vertices, in the following way. 
We put a pair $(v,w)$ into $\leftind{P}^+(j)$, if $v \in M^+_{r'}(j)$, $w \in M^+_r(\leftind{j})$, and 
$v$ has the same rank (according to the above ordering) in $M^+_{r'}(j)$ as the rank of $w$ in $M^+_r(\leftind{j})$. 
Similarly, we put a pair $(v,w)$ into $\leftind{P}^-(j)$, if $v \in M^-_{r'}(j)$, $w \in M^-_r(\leftind{j})$, and 
$v$ has the same rank in $M^-_{r'}(j)$ as $w$ in $M^-_r(\leftind{j})$. 
In addition, we define the sets $\rightind{P}^+(j)$ and $\rightind{P}^-(j)$ analogously, by 
substituting $\rightind{j}$ for $\leftind{j}$ in the definitions.
The key properties of these sets are summarized in the following proposition. 

\begin{myprop}
\label{pairs} W.l.o.g. we can assume $\phi_S(v)=w$ in the following cases: \\
(1)
If $(v,w) \in \leftind{P}^+(j)$ 
and $|M^+_{r'}(j)|=|M^+_{r}(\leftind{j})|$ for some left-aligned $j$.\\
(2)
If $(v,w) \in \leftind{P}^-(j)$  
and $|M^-_{r'}(j)|=|M^-_{r}(\leftind{j})|$ for some left-aligned $j$.\\
(3)
If $(v,w) \in \rightind{P}^+(j)$
and $|M^+_{r'}(j)|=|M^+_{r}(\rightind{j})|$ for some right-aligned $j$.\\
(4)
If $(v,w) \in \rightind{P}^-(j)$   
and $|M^-_{r'}(j)|=|M^-_{r}(\rightind{j})|$ for some right-aligned $j$.
\end{myprop}

\begin{proof}
We only show (1), as all the other statements are analogous. 
To see (1), observe that as $j$ is  left-aligned, $|M^+_{r'}(j)|=|M^+_{r}(\leftind{j})|$ implies that 
$\phi_S$ must map $M^+_{r'}(j)$ to $M^+_{r}(\leftind{j})$ bijectively. 

By Lemma~\ref{I_S_prop}, if $Q^{\mathrm{right}}(v) < Q^{\mathrm{right}}(v')$, then 
$Q^{\mathrm{right}}(\phi_S(v)) < Q^{\mathrm{right}}(\phi_S(v'))$ holds as well. 
Note also that the vertices of $L_r(j,j')$ for some $j' \in [m']$ are equivalent 
in the sense that they have the same neighborhood. Thus, we can assume w.l.o.g. that 
if $v$ precedes $v'$ in the above defined ordering of $M^+_{r'}(j)$, then $\phi_S(v)$ 
precedes $\phi_S(v')$ in the similar ordering of $M^+_{r}(\leftind{j})$. 
Thus, $\phi_S$ must indeed map $v$ to $w$, by the definition of $\leftind{P}^+(j)$.
\qed
\end{proof}

Given two non-trivial fragments $F$ and $H$ of a fragmentation 
with $F$ preceding $H$, we define three disjoint subsets of vertices in $R^{-1}(r')$ 
starting in $F$ and ending in $H$. These sets will be denoted by $\mathcal{L}(F,H)$, 
$\mathcal{R}(F,H)$, and $\mathcal{X}(F,H)$, and we construct them as follows.
Suppose that $v \in L_{r'}(y,j)$ for some $y$ and $j$ contained in $F$ and $H$, respectively.
We put $v$ in exactly one of these three sets, if $(v,w) \in \leftind{P}^-(j)$ for some vertex $w$, and 
$\leftind{y} \leq  Q_r^{\mathrm{left}}(w) \leq \rightind{y}$. 
Now, if $Q_r^{\mathrm{left}}(w) =\leftind{y}$ then we put $v$ into $\mathcal{L}(F,H)$, 
if $Q_r^{\mathrm{left}}(w) =\rightind{y}$ then we put $v$ into $\mathcal{R}(F,H)$, 
and if $\leftind{y}< Q_r^{\mathrm{left}}(w) <\rightind{y}$ then we put $v$ into $\mathcal{X}(F,H)$. 
Loosely speaking, if each vertex in $H$ is left-aligned, 
and some vertex of $\mathcal{R}(F,H)$ starts at $y$,
then $y$ should be right-aligned. Similarly, if each vertex in $H$ is left-aligned, 
and some vertex of $\mathcal{X}(F,H)$ starts at $y$, then $y$ should be either wide or skew.
Since we would like to ensure each index to be left-aligned, we will try to get rid of vertices
of $\mathcal{R}(F,H)$ and $\mathcal{X}(F,H)$.

We say that two indices $y_1, y_2 \in [m']$ are \emph{conflicting} for $(F,H)$, if $y_1 \leq y_2$, 
$M^+_{r'}(y_1) \cap \mathcal{R}(F,H) \neq \emptyset$ and $M^+_{r'}(y_2) \cap \mathcal{L}(F,H) \neq \emptyset$.
In such a case, we say that any $j \geq \max\{j_1,j_2\}$ contained in $H$ is \emph{conflict-inducing} 
for $(F,H)$ (and for the conflicting pair $(y_1,y_2)$), where $j_1$ denotes the minimal index for which 
$L_{r'}(y_1,j_1)  \cap \mathcal{R}(F,H) \neq \emptyset$, and $j_2$ denotes the minimal index for which 
$L_{r'}(y_2,j_2)  \cap \mathcal{L}(F,H) \neq \emptyset$.
Informally, if a conflict-inducing index in $H$ is left-aligned, then this shows that a right-aligned 
index should precede a left-aligned index in $F$, which cannot happen.
In addition, if $\mathcal{L}(F,H) \neq \emptyset$, 
then let $L^{\mathrm{max}}(F,H)$ denote the largest index $y$ in $F$ 
for which $M^+_{r'}(y)  \cap \mathcal{L}(F,H) \neq \emptyset$. 
Let the \emph{L-critical index for $(F,H)$} be the smallest index $j$ contained in $H$ for which 
$L_{r'}(L^{\mathrm{max}}(F,H),j) \cap \mathcal{L}(F,H) \neq \emptyset$.
Similarly, if $\mathcal{R}(F,H) \neq \emptyset$, then
let  $R^{\mathrm{min}}(F,H)$ denote the smallest index $y$ in $F$ 
for which $M^+_{r'}(y) \cap \mathcal{R}(F,H) \neq \emptyset$.
Also, let the \emph{R-critical index for $(F,H)$} be the smallest index $j$ in $H$ 
for which $L_{r'}(R^{\mathrm{min}}(F,H),j) \cap \mathcal{R}(F,H) \neq \emptyset$.

Now, an index $j$ in $H$ is \emph{LR-critical} for $(F,H)$, 
if either $j$ is the R-critical index for $(F,H)$ and $\mathcal{L}(F,H) = \emptyset$, 
or $j=\max\{j_L, j_R\}$ where $j_L$ is the L-critical and $j_R$ is the R-critical index for $(F,H)$.
Note that both cases require $\mathcal{R}(F,H) \neq \emptyset$.
Moreover, $H$ contains an LR-critical index for $(F,H)$, if and only if $\mathcal{R}(F,H) \neq \emptyset$.
Intuitively, if an LR-critical index in $H$ is left-aligned, then this implies that
some index $y$ in $F$ is right-aligned. 

Note that the definitions of the sets $\mathcal{L}(F,H),\mathcal{R}(F,H)$, and $\mathcal{X}(F,H)$ 
together with the definitions connected to them as described above depend on the given fragmentation, 
so whenever the fragmentation changes, these must be adjusted appropriately  as well. 
(In particular, vertices in $\mathcal{L}(F,H),\mathcal{R}(F,H)$, and $\mathcal{X}(F,H)$ 
must start and end in two different non-trivial fragments.)

Let $(\mathcal{F},U)$ be an annotated fragmentation for $(T,T',S)$.
Our aim is to ensure that the properties given below hold for each index $j \in [m']$.
Intuitively, these properties mirror the expectation that every index should be left-aligned.
Note that although we cannot decide whether  $(\mathcal{F},U)$ is a correct fragmentation 
without knowing the solution $S$, we are able to check whether these properties 
hold for some index $j$ in  $(\mathcal{F},U)$.

\begin{description}
\item[Property 1:] 
$G'[X'_j]$ is isomorphic to $G[X_{\leftind{j}}]$.\item[Property 2:] 
$|M_{r'}^+(j)| \leq |M_{r}^+(\leftind{j})| \leq |M_{r'}^+(j)|+k$ and 
$|M_{r'}^-(j)| \leq |M_{r}^-(\leftind{j})| \leq |M_{r'}^-(j)|+k$.
\item[Property 3:] 
If $j$ is non-trivial, then 
$|M_{r'}^+(j)| = |M_{r}^+(\leftind{j})|$ and 
$|M_{r'}^-(j)| = |M_{r}^-(\leftind{j})|$.
\item[Property 4:] 
If $j$ is non-trivial, then 
$|L_{r'}(y,j)|=|L_{r}(\leftind{y},\leftind{j})|$ for any $y < j$ 
contained in the same fragment as $j$.
\item[Property 5:] 
If $j$ is non-trivial, then for every $(v,w) \in \leftind{P}^+(j)$ such that
$Q_{r'}^{\mathrm{right}}(v)$ $= y$ is non-trivial, $\leftind{y} \leq Q_r^{\mathrm{right}}(w) \leq \rightind{y}$ holds.
Also, for every $(v,w) \in \leftind{P}^-(j)$ such that  
$Q_{r'}^{\mathrm{left}}(v)=y$ is non-trivial, $\leftind{y} \leq Q_r^{\mathrm{left}}(w) \leq \rightind{y}$ holds.
\item[Property 6:]
If $j$ is non-trivial, then no vertex in $\mathcal{X}(F,H)$ (for some $F$ and $H$) ends in $j$.
\item[Property 7:]
$j$ is not conflict-inducing for any $(F,H)$.
\item[Property 8:]
$j$ is not LR-critical for any $(F,H)$.
\item[Property 9:] 
If $j$ is non-trivial, then for every $(v,w) \in \leftind{P}^+(j)$ such that
$Q_{r'}^{\mathrm{right}}(v)$ $= y$ is non-trivial, $Q_r^{\mathrm{right}}(w) = \leftind{y}$ holds.
Also, for every $(v,w) \in \leftind{P}^-(j)$ such that  
$Q_{r'}^{\mathrm{left}}(v)=y$ is non-trivial, $Q_r^{\mathrm{left}}(w) = \leftind{y}$ holds.
\item[Property 10:] 
If $j$ is non-trivial, then for each important trivial index $u \in U$, 
$|L_{r'}(j,u)|=|L_{r}(\leftind{j},\leftind{u})|$ holds if $u>j$, and 
$|L_{r'}(u,j)|=|L_{r}(\leftind{u},\leftind{j})|$ holds if $u<j$.
\end{description}
Observe that each of these properties depend on the fragmentation $\mathcal{F}$, 
and Property 10 depends on the set $U$ as well. 
Also, if some property holds for an index $j$ in $(\mathcal{F},U)$, then this does not imply that 
the property holds for $j^{\mathrm{rev}}$ in $(\mathcal{F}^{\mathrm{rev}},U^{\mathrm{rev}})$, 
as most of these properties are not symmetric. For example, $\leftind{j}$ and $\rightind{j}$ both 
have a different meaning in the fragmentation $\mathcal{F}$ and in $\mathcal{F}^{\mathrm{rev}}$.
We say that an index $j \in [m']$ \emph{violates} Property $\ell$ ($1 \leq \ell \leq 10$)
in an annotated fragmentation $(\mathcal{F},U)$, if Property $\ell$ does not hold for $j$ in $(\mathcal{F},U)$.
If the first nine properties hold for each index of $[m']$ both in $(\mathcal{F},U)$ and its reversed version 
$(\mathcal{F}^{\mathrm{rev}},U^{\mathrm{rev}})$, then we say that $(\mathcal{F},U)$ is \emph{9-proper}.
We say that $(\mathcal{F},U)$ is \emph{proper}, if it is 9-proper, and Property 10 holds hold for each index of 
$[m']$ in  $(\mathcal{F},U)$. Observe that $(\mathcal{F},U)$ can be proper even if 
Property 10 does not hold in $(\mathcal{F}^{\mathrm{rev}},U^{\mathrm{rev}})$.

Let us describe our strategy.
We start with an annotated fragmentation where $U= \emptyset$ and the fragmentation contains 
only the unique fragment $([1,m'],[1,m])$, implying that there are no trivial indices.
Given an annotated fragmentation $(\mathcal{F},U)$, we do the following:
if one of Properties $1,2, \dots, 10$ does not hold for some index $j$ in $(\mathcal{F},U)$ for $(T,T',S)$, 
or one of the first nine properties does not hold for some $j$ in the reversed annotated fragmentation 
$(\mathcal{F}^{\mathrm{rev}},U^{\mathrm{rev}})$ for $(T^{\mathrm{rev}},T'^{\mathrm{rev}},S)$,
then we either output a necessary set or an independent subproblem, 
or we modify the given annotated fragmentation.
To do this, we will branch on $I_S(j)$, and handle each possible case according to the type of $j$. 
Note that we do not care whether Property 10 holds for the indices in the reversed instance.
If the given annotated fragmentation is proper, algorithm $\mathcal{A}$ will find a solution 
using Lemmas \ref{right_aligned} and \ref{isomorphism}. 

Observe that we can assume w.l.o.g. that there exists an $\ell$ ($1 \leq \ell \leq 10$) such that 
Properties $1, \dots, \ell-1$ hold for each index 
both in the annotated fragmentation $(\mathcal{F},U)$ and in its reversed version 
$(\mathcal{F}^{\mathrm{rev}},U^{\mathrm{rev}})$, but Property $\ell$ is violated by an index in $[m']$ in 
$(\mathcal{F},U)$. Otherwise, we simply reverse the instance.
Let us now remark that we only reverse the instance if this condition is not true.

To begin, the algorithm takes the first index $j$ that violates Property $\ell$, 
and branches into $(k+1)^2$ directions to choose $I_S(j)$, using Prop.~\ref{I_S_bounds}.
Then, $\mathcal{A}$ handles each of the cases in a different manner, 
according to whether $j$ turns out to be left-aligned, right-aligned, skew, or wide.
We consider these cases in a general way that is essentially independent from $\ell$, 
and mainly relies on the type of $j$.
We suppose that $j$ is contained in a fragment $F=([a',b'],[a,b])$, and
we say that $j$ is \emph{extremal}, if $j=a'$.

{\bf Left-aligned index.}
We deal with the case when $j$ is left-aligned in Section~\ref{sect_left_aligned},
whose results are summarized by the following lemma.

\begin{mylemma}
\label{left_aligned}
Suppose that Property $\ell$ ($ 1 \leq \ell \leq 10$) does not hold for some $j \in [m']$
in the annotated fragmentation $(\mathcal{F},U)$, but
all the previous properties hold for each index 
both in $(\mathcal{F},U)$ and in $(\mathcal{F}^{\mathrm{rev}},U^{\mathrm{rev}})$.
If $j$ is left-aligned, then algorithm $\mathcal{A}$ 
can do one of the followings in linear time:
\begin{itemize}
\item produce a necessary set of size at most $2k+1$,
\item produce an independent subproblem,
\item produce an index that is either wide or skew,
\item reject correctly.
\end{itemize}
\end{mylemma}

By Lemma~\ref{left_aligned}, the only case when algorithm $\mathcal{A}$ does not reject or produce an output is 
the case when it produces an index $j'$ that is wide or skew. If this happens, then $\mathcal{A}$ 
branches on those choices of $I_S(j')$ where $j'$ is indeed wide or skew, 
and handles them according to the cases described below.
Note that as a consequence, the maximum number of branches in a step may increase from  $(k+1)^2$ to $2(k+1)^2-1$. 
(This means that we do not treat the branchings on $I_S(j)$ and $I_S(j')$ separately, 
and rather consider it as a single branching with at most $2(k+1)^2-1$ directions.)

Observe that if $j$ is trivial, then it is both left- and right-aligned. 
We treat trivial indices as left-aligned.

{\bf Wide index.}
Suppose that $j$ is wide, i.e. $|I_S(j)|>1$. 
In this case, we can construct a \outputstyle{necessary set} of size $2$.
Recall that, using the arguments of the proof of Lemma~\ref{max_fragments}, 
we can either find a vertex $z \in R^{-1}(r)$ such that if $z \notin S$ then
$\emptyset \neq X_{\beta_S(j)} \cup M^+_r(\beta_S(j)) \subseteq S$, 
or we can find a vertex $w \in R^{-1}(r)$ such that if $w \notin S$ then 
$\emptyset \neq X_{\alpha_S(j)} \cup M^-_r(\alpha_S(j)) \subseteq S$. 
Clearly, $\mathcal{A}$ can output a necessary set of size 2 in both cases.

{\bf Extremal right-aligned or skew cases.}
Assume that $j=a'$ and $j$ is skew or right-aligned.
In these cases, $X_i \cup M^+_r(i) \cup M^-_r(i)$ must be contained in 
$S$ for each $i$ in $[a,\alpha_S(a')-1]$, so in particular, $X_a \cup M^+_r(a) \cup M^-_r(a) \subseteq S$. 
As $X_a \cup M^+_r(a) \cup M^-_r(a) \neq \emptyset$ by Prop.~\ref{notempty}, 
we can construct a \outputstyle{necessary set} 
of size $1$ by taking an arbitrary vertex from this set,  and $\mathcal{A}$ can stop by outputting it.

{\bf Non-extremal skew case.}
Suppose that $j>a'$ and $j$ is skew, meaning that $I_S(j)=[i,i]$ for some $i$ with 
$\leftind{j} < i < \rightind{j}$. 
In this case, we can divide the fragment $F$, or more precisely, we can delete $F$
from the fragmentation $\mathcal{F}$ and add the new fragments 
$([a',j-1],[a,i-1])$ and $([j,b'],[i,b])$.
Note that the newly introduced fragments are non-trivial by the bounds on $i$. 
We also modify $U$ by declaring every trivial index of the fragmentation to be important
(no matter whether it was important or not before).

{\bf Non-extremal right-aligned case.}
Suppose that $j>a'$ and $j$ is right-aligned.  
In this case, we replace $F$ by new fragments $F_1=([a',j-1],[a,\rightind{j}-1])$ and 
$F_2=([j,b'],[\rightind{j},b])$. This yields a fragmentation where $F_1$ is non-trivial and $F_2$ is trivial. 
We refer to this operation as performing a \emph{right split} at $j$. 
If this happens because $j$ violated Property $\ell$ for some $\ell \leq 9$, then 
we set every trivial index (including those contained in $F_2$) to be important, by putting them into $U$. 
In the case when $\ell = 10$, we do not modify $U$, so the trivial indices of $F_2$ will not be important.

The above process either stops by producing an appropriate output, or it ends by providing
an annotated fragmentation that is proper. 
Thanks to the observations of Lemma~\ref{measure}, stating that the properties ensured 
during some step in this process will not be violated later on (except for a few cases), 
we will be able to bound the running time of this process in Section~\ref{sect_running_time},
by proving that the height of the explored search tree is bounded by a function of $k$.
In the remaining steps of the algorithm, the set $U$ will never be modified, and 
the only possible modification of the actual fragmentation will be to perform a right split. 

The following two lemmas capture some useful properties of an arbitrary
annotated fragmentation $(\mathcal{F},U)$ obtained by the algorithm after this point. 
Lemma~\ref{invariant} states facts about an annotated fragmentation obtained from 
a 9-proper annotated fragmentation by applying right splits to it.
Lemma~\ref{split} gives sufficient conditions for the properties of an annotated fragmentation 
to remain true after applying a right split to it.

\begin{mylemma}
\label{invariant}
Let $(\mathcal{F},U)$ be a 9-proper annotated fragmentation whose trivial indices are all important.
Suppose that $\mathcal{F}'$ is obtained by applying an arbitrary number of right splits 
to the fragmentation $\mathcal{F}$. 
Then the followings hold for each $j \in [m']$ that is 
either non-trivial or not important in $(\mathcal{F}',U)$: \\
(1) 
$|M_{r'}^+(j)| = |M_{r}^+(\rightind{j})|$ and 
$|M_{r'}^-(j)| = |M_{r}^-(\rightind{j})|$. \\
(2)
The following holds for every non-trivial or not important $y \neq j$ and $v \in L_{r'}(j,y)$. 
If $(v,w) \in \rightind{P}^+(j)$ for some  $w \in M^+_{r}(\rightind{j})$, 
then $Q_{r}^{\mathrm{right}}(w)=\rightind{y}$.
Similarly, if $(v,w) \in \rightind{P}^-(j)$ for some $w \in M^-_{r}(\rightind{j})$, 
then $Q_{r}^{\mathrm{left}}(w)=\rightind{y}$.
\end{mylemma} 

\begin{proof}
First, we show that the statements of the lemma hold for $(\mathcal{F},U)$.
To see this, recall that each trivial index in $(\mathcal{F},U)$ is important, therefore 
statements (1) and (2) for $(\mathcal{F},U)$ are equivalent to 
Properties 3 and 9 for $(\mathcal{F}^{\mathrm{rev}},U^{\mathrm{rev}})$, respectively.
Since $(\mathcal{F},U)$ is 9-proper, these properties indeed hold for each index in 
$(\mathcal{F}^{\mathrm{rev}},U^{\mathrm{rev}})$.

To see that these statements remain true after applying a sequence of right splits to $(\mathcal{F},U)$, 
we need two simple observations.
First, notice that the value of $\rightind{j}$ for an index $j \in [m']$ does not change in a right split.
Second, the set of non-trivial or not important trivial indices does not change either, since the performed right splits 
do not modify the set $U$ of important trivial indices.
Thus, statements (1) and (2) for some index $j$ have exactly the same meaning 
in $(\mathcal{F}',U)$ as in $(\mathcal{F},U)$. 
This proves the lemma.
\qed
\end{proof}

Given a fragmentation $\mathcal{F}$ for $(T,T',S)$, a fragment $F \in \mathcal{F}$, and  
some $\ell$ ($1 \leq \ell \leq 9$), let $\pi(\mathcal{F},F,\ell)$ be $1$ if 
Property $\ell$ holds for each index  in $F \in \mathcal{F}$, and $0$ otherwise.

\begin{mylemma}
\label{split}
Let $\mathcal{F}'$ be a fragmentation obtained from $\mathcal{F}$ by dividing a fragment $F \in \mathcal{F}$
into fragments $F_1$ and $F_2$ with a right split (with $F_1$ preceding $F_2$). Let $1 \leq \ell \leq 9$. \\
(1)
Suppose $j$ is not contained in $F_2$ and $\ell \neq 8$.
If Property $\ell$ holds for $j$ in $\mathcal{F}$ (or in $\mathcal{F}^{\mathrm{rev}}$),
then Property $\ell$ holds for $j$ in $\mathcal{F}'$ (or in $\mathcal{F}'^{\mathrm{rev}}$) as well. \\
(2)
Suppose $\pi(\mathcal{F},H,\ell)=1$ for a fragment $H$.
If $H \neq F$ then $\pi(\mathcal{F}',H,\ell)=1$, and if $H=F$ then 
$\pi(\mathcal{F}',F_1,\ell)=1$. \\
(3)
Suppose $\pi(\mathcal{F}^{\mathrm{rev}},H^{\mathrm{rev}},\ell)=1$ for a fragment $H \in \mathcal{F}$.
If $H \neq F$ then $\pi(\mathcal{F}'^{\mathrm{rev}},H^{\mathrm{rev}},\ell)=1$, and if $H=F$ then 
$\pi(\mathcal{F}'^{\mathrm{rev}},F_1^{\mathrm{rev}},\ell)=1$.\\
(4)
If  $\pi(\mathcal{F}^{\mathrm{rev}},F^{\mathrm{rev}},\ell)=1$, then
$\pi(\mathcal{F}'^{\mathrm{rev}},F_2^{\mathrm{rev}},\ell)=\pi(\mathcal{F}',F_2,\ell)=1$. \\
(5)
If $(\mathcal{F},U)$ is a proper annotated fragmentation, then so is $(\mathcal{F}',U)$. 
\end{mylemma}

\begin{proof}
To see (1), we need some basic observations. 
First, if $j$ is not contained in $F_2$, then $\leftind{j}$ is the same according to 
$\mathcal{F}'$ as it is in $\mathcal{F}$, and this is also true for $\rightind{j}$. 
Second, the set of non-trivial indices in $\mathcal{F}'$ is a subset of the non-trivial indices in $\mathcal{F}$.
These conditions directly imply (1) for each case where $\ell \notin \{6,7,8\}$, 
using only the definitions of these properties.

Now, observe that if a vertex in $R^{-1}(r')$ is contained in $\mathcal{L}(H'_1,H'_2)$, 
for some $H'_1$ and $H'_2$ in the fragmentation $\mathcal{F}'$, then
it is contained in $\mathcal{L}(H_1,H_2)$ for some $H_1$ and $H_2$ in $\mathcal{F}$ as well.
Clearly, the analogous fact holds also for the sets 
$\mathcal{R}(H'_1,H'_2)$ and $\mathcal{X}(H'_1,H'_2)$ for some $H'_1$ and $H'_2$.
Thus, if $j$ violates Property 6 or 7 in $\mathcal{F}'$, then it also violates it in $\mathcal{F}$, 
proving (1). 

Clearly, (2) and (3) follow directly from (1) in the cases where $\ell \neq 8$. 
For the case $\ell=8$, observe that $\pi(\mathcal{F},H,8)=1$ implies 
$\mathcal{R}(H_0,H) = \emptyset$ for every $H_0$ preceding $H$.
Hence, the requirements of statement (2) follow immediately.
The analogous claim in the reversed instance shows that (3) also holds for $\ell=8$.

To prove (4), let $j$ be contained in $F_2$.
Note that Properties $3, 4, \dots, 9$ vacuously hold for $j$ in $\mathcal{F}'$, because $F_2$ is trivial.
Using that $\leftind{j} = \rightind{j}$ and the definitions of Properties 1 and 2, 
we get that if one of these two properties holds for $j^{\mathrm{rev}}$ in $\mathcal{F}^{\mathrm{rev}}$, 
then it holds for $j$ in $\mathcal{F}'$ as well. 
Finally, observe that if Property $\ell$ holds for some trivial index $j$ in $\mathcal{F}'$, 
then it trivially holds for $j^{\mathrm{rev}}$ in $\mathcal{F}'^{\mathrm{rev}}$, proving (4).

To prove (5), assume that $(\mathcal{F},U)$ is proper. By (2), (3), and (4), we immediately obtain that
$(\mathcal{F}',U)$ is 9-proper, so we only have to verify that Property 10 holds. But since 
the set $U$ of important trivial indices is the same in both fragmentations, and $\leftind{j}$ is the same in 
$(\mathcal{F}',U)$ as in $(\mathcal{F},U)$ for each non-trivial or important trivial index $j$ of $\mathcal{F}'$, 
Property 10 also remains true for each index.
\qed
\end{proof}
Given a proper annotated fragmentation $(\mathcal{F},U)$, 
algorithm $\mathcal{A}$ makes use of Lemma~\ref{right_aligned} below.

To state Lemma~\ref{right_aligned}, we need one more definition: 
we call an index $j$ \emph{right-constrained}, if $j$ is contained in a non-trivial
fragment $F$, and there exists a vertex $v \in M^+_{r'}(j)$ such that $\phi_S(v) \in M^+_r(\rightind{j})$.
Note that this definition depends on the solution $S$.
Algorithm $\mathcal{A}$ maintains a set $W$ to store indices which turn out to be right-constrained.
We will show that if $j$ is right-constrained, then $j+1$ must be right-aligned and thus a right split can be performed, 
except for the case when $j$ is the last index of the fragment.
We will denote by $Z_{\mathcal{F}}$ the set of indices $j$ for which $j$ is the last index of some 
non-trivial fragment in $\mathcal{F}$.
If no confusion arises, we will drop the subscript $\mathcal{F}$.

Lemma~\ref{right_aligned} gives sufficient conditions for $\mathcal{A}$ to 
do some of the followings.
\begin{itemize} 
\item Find out that some non-trivial index $j$ is right-aligned. 
In this case, $\mathcal{A}$ performs a right split at $j$ in the actual fragmentation.
\item Find out that some index $j$ is right-constrained, and put it into $W$.
\item Reject, or stop by outputting a necessary set of size 1.
\end{itemize} 
The algorithm applies Lemma~\ref{right_aligned} repeatedly, until it either stops or 
it finds that none of the conditions given in the lemma apply.

\begin{table}[t]
\begin{center}
    \psfrag{U}[][]{$\in U$}
    \psfrag{nU}[][]{$\notin U$}
    \psfrag{nZ}[][]{$\notin Z$}
    \psfrag{(i)}[][]{(i)}
    \psfrag{(ii)}[][]{(ii)}
    \psfrag{(iii)}[][]{(iii)}
    \psfrag{(iv)}[][]{(iv)}	
    \psfrag{(v)}[][]{(v)}
    \psfrag{triv}[][]{trivial}
	\psfrag{nontriv}[][]{non-trivial}
%    \psfrag{rightc}[][]{right-constrained}
    \psfrag{W}[][]{$\in W$}
    \psfrag{ax}[l][l]{$a$}
    \psfrag{bx}[l][l]{$b$}
%    \psfrag{N es L}[][]{$\overbrace{N \! \setminus \! Z \quad Z}^{\displaystyle N}$}
\includegraphics[scale=0.5]{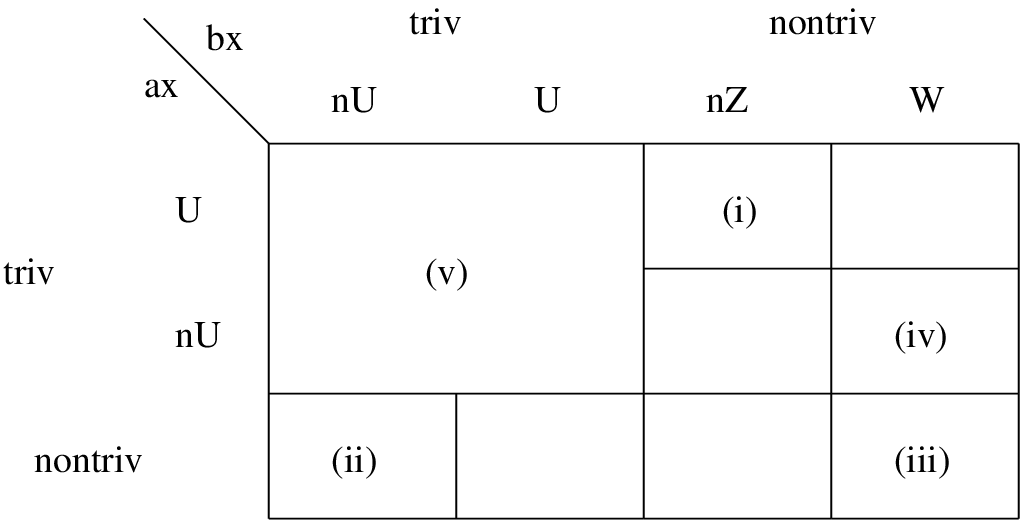}
\end{center}
\caption{
The cases of Lemma~\ref{right_aligned}, where $U$ denotes important indices, $Z$ denotes 
the last indices of non-trivial fragments, and $W$ denotes right-constrained indices.
}
\label{fig_table_cases2}
\end{table}

\begin{mylemma}
\label{right_aligned}
Let $(\mathcal{F},U)$ be a proper annotated fragmentation for $(T,T',S)$ obtained by algorithm $\mathcal{A}$, 
and let $a,b \in [m']$ with $a<b$. \\
(i) If $a$ is trivial but not important, $b$ is non-trivial, $b \notin Z$ and $L_{r'}(a,b) \neq \emptyset$, 
then $b+1$ is right-aligned. \\
(ii) If $a$ is non-trivial, $b$ is trivial but not important, and $L_{r'}(a,b) \neq \emptyset$, 
then $a$ is right-constrained. Also, if $a \notin Z$ then $a+1$ is right-aligned.  \\
(iii) If $a$ is non-trivial, $b$ is right-constrained, and $L_{r'}(a,b) \neq \emptyset$, 
then $a$ is right-constrained. Also, if $a \notin Z$ then $a+1$ is right-aligned. \\
(iv) If $a \in U$, $b$ is right-constrained, and 
$|L_{r'}(a,b)| \neq |L_r(\leftind{a},\rightind{b})|$, 
then algorithm $\mathcal{A}$ can either \outputstyle{reject} or output a \outputstyle{necessary set} of size 1. \\
(v) If $a$ and $b$ are trivial and 
$|L_{r'}(a,b)| \neq |L_r(\leftind{a},\leftind{b})|$, 
then algorithm $\mathcal{A}$ can either \outputstyle{reject} or output a \outputstyle{necessary set} of size 1.
\end{mylemma}

\begin{proof}
Let $A$ and $B$ be the fragments in $\mathcal{F}$ containing $a$ and $b$, respectively.
Recall that the conditions of Lemma~\ref{invariant} are true for every 
proper annotated fragmentation for $(T,T',S)$ obtained by algorithm $\mathcal{A}$, in particular for $\mathcal{F}$.

First, suppose that the conditions of (i) hold. 
As $a$ is a not important trivial index in $(\mathcal{F},U)$, 
claim (1) of Lemma~\ref{invariant} implies $|M^+_{r'}(a)| = |M^+_{r}(\rightind{a})|$.
Let $v \in L_{r'}(a,b)$ and $(v,w) \in \rightind{P}^+(a)$. 
As $a$ is trivial, it is right-aligned as well, so we obtain $\phi_S(v)=w$ by Prop.~\ref{pairs}.
Using claim (2) of Lemma~\ref{invariant} for $a$, 
we obtain $Q_{r}^{\mathrm{right}}(w)=\rightind{b}$.
By $\phi_S(v)=w$, this implies $\beta_S(b) \geq \rightind{b}$. Thus, 
$\alpha_S(b+1) \geq \rightind{b}+1 = \rightind{(b+1)}$, showing that $b+1$ is indeed right-aligned.

The proof of (ii) is analogous with the proof of (i). By exchanging the roles of $a$ and $b$, 
we obtain $Q_{r}(\phi_S(v))=[\rightind{a},\rightind{b}]$ for some $v \in L_{r'}(a,b)$ 
in a straightforward way. Observe that this proves $a$ to be right-constrained.
If $a \notin Z$, then $A$ contains $a+1$ as well. 
Hence, from $\beta_S(a) \geq \rightind{a}$ we get
$\alpha_S(a+1) \geq \rightind{a}+1=\rightind{(a+1)}$. Thus, $a+1$ is right-aligned.

To see (iii) and (iv), suppose that $b$ is right-constrained and $u^+$ is a vertex
in $M^+_{r'}(b)$ with $\phi_S(u^+) \in M^+_r(\rightind{b})$.
Suppose $u^- \in M^-_{r'}(b)$ for some $u^-$. 
Clearly, $u^- u^+$ is an edge in $G'$, so $\phi_S(u^-)$ and $\phi_S(u^+)$ must be adjacent in $G$ as well.
By $\phi_S(u^+) \in M^+_r(\rightind{b})$ we get 
$Q_r^{\mathrm{right}}(\phi_S(u^-)) \geq \rightind{b}$. By Prop.~\ref{limit_indices}, 
this implies $Q_r^{\mathrm{right}}(\phi_S(u^-)) = \rightind{b}$.
Using claim (1) of Lemma~\ref{invariant} for $b$, we get $M^-_r(\rightind{b}) = \phi_S(M^-_{r'}(b))$. 
Letting $v \in L_{r'}(a,b)$ and $(v,w) \in \rightind{P}^-(b)$ 
we obtain $\phi_S(v)=w$ as in Prop.~\ref{pairs}.

To prove (iii), assume also that $a$ is non-trivial.
By claim (2) of Lemma~\ref{invariant} for $b$, this implies  
$Q_{r}(w)=[\rightind{a},\rightind{b}]$.
This means that $a$ is right-constrained. 
From $a \notin Z$ we again obtain that $a+1$ is right-aligned, 
using the arguments of the proof of (ii).

To see (iv), assume $a \in U$.
Using Prop.~\ref{limit_indices}, $I_S(a)=[\leftind{a},\leftind{a}]$, and the above mentioned arguments, we get that 
$\phi_S(L_{r'}(a,b))= L_r(\leftind{a},\rightind{b}) \setminus S$.
Therefore, if $|L_{r'}(a,b)| > |L_r(\leftind{a},\rightind{b})|$ then $\mathcal{A}$ can reject,
and if $|L_{r'}(a,b)| < |L_r(\leftind{a},\rightind{b})|$ then it can output a necessary set of size 1 
by outputting $\{s\}$ for an arbitrary $s \in L_r(\leftind{a},\rightind{b})$.

Finally, assume that the conditions of (v) hold for $a$ and $b$.
As both of them are left-aligned, Prop.~\ref{limit_indices} implies 
$\phi_S(L_{r'}(a,b))= L_r(\leftind{a},\leftind{b}) \setminus S$.
Hence, $\mathcal{A}$ can proceed essentially the same way as in the previous case.
\qed
\end{proof}

After applying Lemma~\ref{right_aligned} repeatedly, algorithm $\mathcal{A}$ either stops 
by outputting 'No' or a necessary set of size 1, or it finds that none of the conditions (i)-(v) 
of Lemma~\ref{right_aligned} holds. 
Observe that each $w \in W$ must be the last index of the fragment containing $w$, 
since whenever $\mathcal{A}$ puts some index $j \notin Z$ into $W$, then it also sets $j+1$ right-aligned, 
resulting in a right split.

Let $(\mathcal{F},U)$ be the final annotated fragmentation obtained.
Note that the algorithm does not modify the set $U$ of important trivial indices when applying Lemma~\ref{right_aligned}, 
and it can only modify the actual fragmentation by performing a right split.
Thus, statements (1) and (2) of Lemma~\ref{invariant} remain true for $(\mathcal{F},U)$.
By claim (5) of Lemma~\ref{split}, we obtain that $(\mathcal{F},U)$ remains proper as well.
Making use of these lemmas, Lemma~\ref{isomorphism} yields that $\mathcal{A}$ can find a solution in linear time. 
This finishes the description of algorithm $\mathcal{A}$.

\begin{mylemma} 
\label{isomorphism}
Let $(\mathcal{F},U)$ be a proper annotated fragmentation for $(T,T',S)$ obtained by algorithm $\mathcal{A}$. 
If none of the conditions (i)-(v) of Lemma~\ref{right_aligned} holds, 
then $\mathcal{A}$ can produce a solution in linear time.
\end{mylemma}

\begin{proof}
We construct an isomorphism $\phi$ from $G'$ to an induced subgraph of $G$. 
Our basic approach is to treat almost all indices as if they were left-aligned, except for the vertices of $W$. 
Recall that $Z$ denotes the set of indices that are the last index of some non-trivial fragment, and 
$W \subseteq Z$ is the set of right-constrained vertices that $\mathcal{A}$ has found using Lemma~\ref{right_aligned}.
Let $N$ contain those non-trivial indices in $[m']$ that are not in $W$.
Also, let $Y$ denote the set of trivial indices in $[m']$ that are not important.
Clearly, $[m']=N \cup W \cup U \cup Y$.

As Property 1 holds for each index both in $\mathcal{F}$ and in $\mathcal{F}^{\mathrm{rev}}$, 
we know that there is an isomorphism 
$\phi_j^{\mathrm{left}}$ from $G'[X'_j]$ to $G[X_{\leftind{j}}]$ and an isomorphism 
$\phi_j^{\mathrm{right}}$ from $G'[X'_j]$ to $G[X_{\rightind{j}}]$ 
for each $j \in [m']$.
By \cite{booth-lueker-interval-isomorphism}, $\phi_j^{\mathrm{left}}$ and $\phi_j^{\mathrm{right}}$ 
can be found in time linear in $|X'_j|$.  We set $\phi(x)= \phi_j^{\mathrm{left}}(x)$ for each 
$x \in X'_j$ where $j \in N \cup U \cup Y$, and we set 
$\phi(x)= \phi_j^{\mathrm{right}}(x)$ for each $x \in X'_j$ where $j \in W$.
Our aim is to extend $\phi$ on vertices of $R^{-1}(r')$ such that it remains an isomorphism.
To this end, we set a variable $\Delta(j)$ for each $j \in [m']$, by letting $\Delta(j)=\leftind{j}$ if 
$j \in N \cup U \cup Y$, and $\Delta(j)=\rightind{j}$ if $j \in W$.
Clearly, $\Delta(j)=\leftind{j}=\rightind{j}$ if $j \in U \cup Y$.

The purpose of the notation $\Delta$ is the following.
Given some $a<b$, in almost every case we will let $\phi$ map vertices of $L_{r'}(a,b)$ 
bijectively to vertices of $L_r(\Delta(a),\Delta(b))$.
This can be done if $a$ and $b$ \emph{match}, meaning that $|L_{r'}(a,b)|=|L_r(\Delta(a),\Delta(b))|$.
However, there remain cases where $a$ and $b$ do not match. Each 
such case will fulfill one of the following conditions:
\begin{itemize}
\item[(A)] 
$a \in W$ and  $|L_{r'}(a,b)|=|L_r(\leftind{a},\Delta(b))|$.
\\
In this case, we let $\phi$ map $L_{r'}(a,b)$ bijectively to $L_r(\leftind{a},\Delta(b))$.
Clearly, the block $[\leftind{a},\Delta(b)]$ contains $\Delta(a)=\rightind{a}$. 
Thus, vertices of $\phi(L_{r'}(a,b))$ will be adjacent to vertices of 
$M_r(\Delta(a)) \cup X_{\Delta(a)}$.
Since either $a-1 \in N$ or $a-1$ is not in the same fragment as $a$, 
we obtain $\Delta(a-1) < \leftind{a}$.
Hence, vertices of $\phi(L_{r'}(a,b))$ will not be adjacent 
to vertices of $M^-_r(\Delta(a-1)) \cup X_{\Delta(a-1)}$. 
\item[(B)] 
$b \in Z$ and $|L_{r'}(a,b)|=|L_r(\Delta(a),\rightind{b})|$.  \\
In this case, we will let $\phi$ map $L_{r'}(a,b)$ bijectively to $L_r(\Delta(a),\rightind{b})$.
Again, $[\Delta(a),\rightind{b}]$ contains $\Delta(b)=\leftind{b}$, so
vertices of $\phi(L_{r'}(a,b))$ will be adjacent to vertices of 
$M_r(\Delta(b)) \cup X_{\Delta(b)}$.
Also, by $b \in Z$ we obtain $\Delta(b+1) \geq \leftind{(b+1)} > \rightind{b}$,
so the vertices of $\phi(L_{r'}(a,b))$ 
will not be adjacent to vertices of $M^+_r(\Delta(b+1)) \cup X_{\Delta(b+1)}$.
\end{itemize}

It is easy to see that the above construction ensures that
vertices of $\phi(L_{r'}(a_1,b_1))$ and $\phi(L_{r'}(a_2,b_2))$ 
are neighboring if and only if 	$L_{r'}(a_1,b_1)$ and $L_{r'}(a_2,b_2)$ are neighboring. 
(In particular, it is not possible that 
some vertex of $\phi(M^-_{r'}(j))$ ends in $\leftind{j}$ but some 
vertex of $\phi(M^+_{r'}(j))$ starts in $\rightind{j}$. )

It remains to show that if $a,b \in [m']$ and $a<b$, then they either match, or $L_{r'}(a,b)=\emptyset$, 
or one of the conditions (A) or (B) hold. First, let us show those cases where $a$ and $b$ match.

\begin{itemize}
\item[(a)]
If $a,b \in N$, then $|L_{r'}(a,b)|=|L_r(\leftind{a},\leftind{b})|=|L_r(\Delta(a),\Delta(b))|$ 
because Properties 3 and 9 hold for $b$ in $(\mathcal{F},U)$.
\item[(b)]
If either $a \in N$ and $b \in U$ or vice versa, 
then $|L_{r'}(a,b)|=|L_r(\leftind{a},\leftind{b})|=|L_r(\Delta(a),\Delta(b))|$, 
since Property 10 holds for $a$ and $b$ in $(\mathcal{F},U)$.
\item[(c)]
If $a, b \in W \cup Y$ then Lemma~\ref{invariant} for $b$
guarantees $|L_{r'}(a,b)|=|L_r(\rightind{a},\rightind{b})|$.
Using that $\rightind{a} = \Delta(a)$ and $\rightind{b} = \Delta(b)$ hold if $a,b \in W \cup Y$, 
this shows $|L_{r'}(a,b)|=|L_r(\Delta(a),\Delta(b))|$.
\item[(d)] 
If $a \in U$ and $b \in W$ then $|L_{r'}(a,b)|=|L_r(\leftind{a},\rightind{b})|$, 
since the conditions of (iv) in Lemma~\ref{right_aligned} do not apply.
By $\leftind{a} = \Delta(a)$ and $\rightind{b} = \Delta(b)$, this means 
that $|L_{r'}(a,b)|=|L_r(\Delta(a),\Delta(b))|$.
\item[(e)] 
If $a,b \in U \cup Y$ then $|L_{r'}(a,b)|=|L_r(\leftind{a},\leftind{b})|=|L_r(\Delta(a),\Delta(b))|$, as
the conditions of (v) in Lemma~\ref{right_aligned} do not apply.
\end{itemize}

Next, we show $L_{r'}(a,b)= \emptyset$ for some $a$ and $b$ with $a<b$.
First, if $a \in Y$ and $b \in N \setminus Z$, then this holds because 
(i) of Lemma~\ref{right_aligned} is not applicable. 
Also, $L_{r'}(a,b)= \emptyset$ must be true if $a \in N$ and $b \in Y$, as otherwise
(ii) of  Lemma~\ref{right_aligned} would apply.
Third, $L_{r'}(a,b)= \emptyset$ if $a \in N$ and $b \in W$, since 
(iii) of Lemma~\ref{right_aligned} does not apply.

\begin{table}[t]
\begin{center}
    \psfrag{a}[][]{a}
    \psfrag{b}[][]{b}
    \psfrag{c}[][]{c}
    \psfrag{d}[][]{d}
    \psfrag{e}[][]{e}
    \psfrag{f}[][]{f}
    \psfrag{g}[][]{g}
    \psfrag{h}[][]{h}
    \psfrag{ax}[l][l]{$a$}
    \psfrag{bx}[l][l]{$b$}
    \psfrag{N}[][]{$N$}
    \psfrag{Y}[][]{$Y$}
    \psfrag{W}[][]{$W$}
    \psfrag{U}[][]{$U$}
    \psfrag{--}[][]{$-$}
    \psfrag{N es L}[][]{$\overbrace{\, N \! \setminus \! Z \; \; \quad Z\,}^{\displaystyle N}$}
%    \psfrag{N es L}[][]{$\overbrace{N \setminus Z \quad Z \cap N \;}^{\displaystyle N}$}
\includegraphics[scale=0.5]{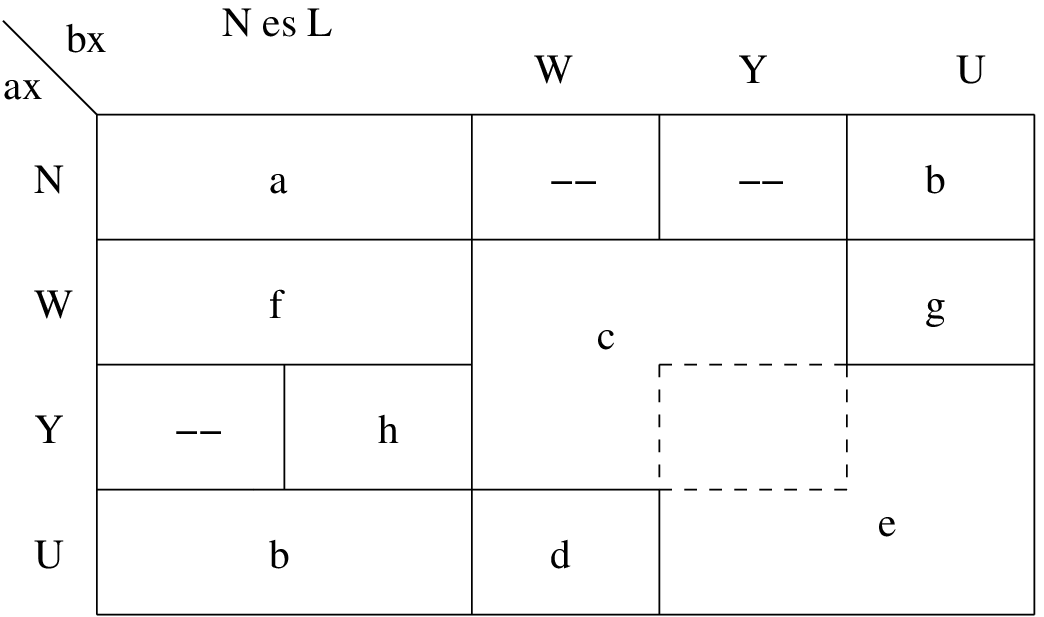}
\end{center}
\caption{
The cases of the proof for Lemma~\ref{isomorphism}.
}
\label{fig_table_cases}
\end{table}

We complete the proof by showing (A) or (B) for all remaining cases.
\begin{itemize}
\item[(f)]
If $a \in W$ and $b \in N$, then (A) holds, because
by Properties 3 and 9 for $b$ in $(\mathcal{F},U)$, we obtain $|L_{r'}(a,b)|=|L_r(\leftind{a},\leftind{b})|$.
\item[(g)]
If $a \in W$ and $b \in U$, then we have $|L_{r'}(a,b)|=|L_r(\leftind{a},\leftind{b})|$, 
since Property 10 holds for $a$ in $(\mathcal{F},U)$. 
Hence, this case also fulfills (A).
\item[(h)]
If $a \in Y$ and $b \in Z$, then $|L_{r'}(a,b)|=|L_r(\rightind{a},\rightind{b})|$
by (1) and (2) of Lemma~\ref{invariant} for $b$. By $\leftind{a}=\rightind{a}$, (B) holds.
\end{itemize}

Table~\ref{fig_table_cases} shows that we considered every case. 
Thus, 
$\phi$ is an isomorphism from $G'$ to an induced subgraph of $G$, so $\mathcal{A}$ 
can output $V(G) \setminus \phi(V(G'))$ as a solution. 
It is also clear that this takes linear time. 
\qed
\end{proof}

\subsection{Running time analysis for algorithm $\mathcal{A}$}
\label{sect_running_time}

Let $N(\mathcal{F})$ denote the set of non-trivial fragments in $\mathcal{F}$.
We define the \emph{measure} $\mu(\mathcal{F})$ of a given fragmentation 
$\mathcal{F}$ for $(T,T',S)$ as follows:
$$ \mu(\mathcal{F}) =  
\sum_{
\substack{
F \in N(\mathcal{F}) \\
1 \leq \ell \leq 9}} 
\pi(\mathcal{F},F,\ell) 
+ 
\sum_{
\substack{
F \in N(\mathcal{F}^{\mathrm{rev}}) \\
1 \leq \ell \leq 9}} 
\pi(\mathcal{F}^{\mathrm{rev}},F,\ell).$$
Note that $\mu(\mathcal{F}) = \mu(\mathcal{F}^{\mathrm{rev}})$ is trivial, 
so reversing a fragmentation does not change its measure. Recall 
that $\mathcal{F}^{\mathrm{rev}}$ is a fragmentation for $(T^{\mathrm{rev}},T'^{\mathrm{rev}},S)$.

\begin{mylemma}
\label{measure}
Let $\mathcal{F}_1, \dots, \mathcal{F}_t, \mathcal{F}_{t+1}$ be a series a fragmentations such that for each $i \in [t]$
algorithm $\mathcal{A}$  obtains $\mathcal{F}_{i+1}$ from $\mathcal{F}_i$ by applying a right split at an 
index $j_i$ violating Property $\ell_i$ in $\mathcal{F}_i$. 
Let $H_i$ denote the fragment of $\mathcal{F}_i$ containing $j_i$.
\\
(1) $\mu(\mathcal{F}_{i+1}) \geq \mu(\mathcal{F}_i)$ for each $i \in [t]$.  
If $\ell_i \neq 8$, then $\mu(\mathcal{F}_{i+1}) > \mu(\mathcal{F}_i)$ also holds. 
\\
(2)
If $\mu(\mathcal{F}_1)=\mu(\mathcal{F}_{t+1})$, then $\ell_i = 8$ for every $i \in [t]$,
and $H_i$ contains every index in $H_{i+1}$ for each $i \in [t]$. \\
(3)
If $\mu(\mathcal{F}_1)=\mu(\mathcal{F}_{t+1})$, then $t \leq k$.
\end{mylemma}

\begin{proof}
To prove (1), observe that claims (2) and (3) of Lemma~\ref{split} imply directly 
that $\mu(\mathcal{F}_{i+1}) \geq \mu(\mathcal{F}_i)$.
Let $H'_i$ be the non-trivial fragment obtained from $H_i$ after the right split at $j_i$. 
Now, by the choice of $j_i$, Property $\ell_i$ is violated by $j_i$ in $\mathcal{F}_i$, 
but is not violated by any index $j'$ preceding $j_i$ in $\mathcal{F}_i$. 
In all cases where $\ell_i \neq 8$, claim (1) of Lemma~\ref{split} implies that the 
indices preceding $j_i$ cannot violate Property $\ell_i$ in $\mathcal{F}_{i+1}$, 
yielding $\pi(\mathcal{F}_i,H_i,\ell_i)=0$ but $\pi(\mathcal{F}_i,H'_i,\ell_i)=1$.
Considering claims (2) and (3) of Lemma~\ref{split} again, (1) follows.

Observe that if $\mu(\mathcal{F}_1)=\mu(\mathcal{F}_t)$, then $\ell_i =8$ for every $i \in [t]$
follows directly from the above discussion.
Suppose that $H_i$ is a counterexample for (2), meaning that $H_i$ does not contain 
the indices of $H_{i+1}$. Since $\mathcal{F}_{i+1}$ is obtained from $\mathcal{F}_i$ by a right split, 
this can only happen if $H_{i+1}$ is a non-trivial fragment of $\mathcal{F}_i$ different from $H_i$. 
Recall that a fragment $B$ contains some index violating Property 7, 
if and only if $\mathcal{R}(A,B) \neq \emptyset$ holds for some fragment $A$ in the fragmentation. 
Hence, $\pi(\mathcal{F}_{i+1},H_{i+1},8)=0$ implies $\pi(\mathcal{F}_{i},H_{i+1},8)=0$.

Since the algorithm always chooses the first index violating some property to branch on, 
$j_i$ must be the smallest index that is LR-critical for some pair of fragments in $\mathcal{F}_i$. 
Therefore, $H_i$ must precede $H_{i+1}$. 
But now, the choice of $j_{i+1}$ indicates $\pi(\mathcal{F}_{i+1},H'_i,8)=1$, 
where $H'_i$ is the non-trivial fragment of $\mathcal{F}_{i+1}$
obtained by splitting $H_i$ at $j_i$ in $\mathcal{F}_i$.
Together with $\pi(\mathcal{F}_i,H_i,8)=0$ and statements (2) and (3) of Lemma~\ref{split}, 
this shows $\mu(\mathcal{F}_{i+1}) > \mu(\mathcal{F}_i)$, a contradiction. 

To prove (3), note that by the claim proven above, $H_1 \in \mathcal{F}_1$ contains every $j_i$.
Let $\mathcal{P}$ denote the set of non-trivial fragments in $\mathcal{F}_1$ preceding $H$.
By the construction of the fragmentations $\mathcal{F}_i$, each fragment in $\mathcal{P}$
is a non-trivial fragment of $\mathcal{F}_i$ as well, preceding $H_i$.
We denote by $\mathcal{P}_{\mathcal{R},i}$ those fragments $F$ in $\mathcal{P}$ for which 
$\mathcal{R}(F,H_i) \neq \emptyset$ holds in $\mathcal{F}_i$.
Since $j_i$ is LR-critical for some pair of fragments in $\mathcal{F}_i$, 
we get $\mathcal{P}_{\mathcal{R},i} \neq \emptyset$ for any $i \in [t]$. 
Note also $\mathcal{P}_{\mathcal{R},i+1} \subseteq \mathcal{P}_{\mathcal{R},i}$. 

For some $F \in \mathcal{P}_{\mathcal{R},i}$, then we define $d^{i}(F)$ as follows. 
If $\mathcal{L}(F,H_i) = \emptyset$ in $\mathcal{F}_i$, then 
let $y^L_i$ be the first index contained in $F$ minus one, 
otherwise let $y^L_i$ have the value of $L^{\mathrm{max}}(F,H_i)$ in $\mathcal{F}_i$. 
Also, let $y^R_i$ be the value of $R^{\mathrm{min}}(F,H_i)$ in  $\mathcal{F}_i$. 
We set $d^{i}(F) = y^R_i - y^L_i$.
Let $A \in \mathcal{P}_{\mathcal{R},i+1} \cap \mathcal{P}_{\mathcal{R},i}$ 
be a non-trivial fragment such that $j_i$ is LR-critical for $(A,H_i)$. 
We show  $d^{i+1}(A) > d^{i}(A)$.
First note that neither $L^{\mathrm{max}}(A,H_{i+1}) > L^{\mathrm{max}}(A,H_i)$ 
nor $R^{\mathrm{min}}(A,H_{i+1}) < R^{\mathrm{min}}(A,H_i)$ is possible, since 
$\mathcal{L}(A,H_{i+1}) \subseteq \mathcal{L}(A,H_i)$ and 
$\mathcal{R}(A,H_{i+1}) \subseteq \mathcal{R}(A,H_i)$ always hold. 
This implies $y^L_{i+1} \leq y^L_{i}$ and $y^R_{i+1} \geq y^R_i$.

Clearly, $j_i$ is either L-critical or R-critical for $(A,H_i)$.
First, let us assume that $j_i$ is L-critical for $(A,H_i)$. 
Observe that the definition of L-criticality
implies that for any vertex $v$ starting at $L^{\mathrm{max}}(A,H_i)$ and contained in  
$\mathcal{L}(A,H_i)$ in $\mathcal{F}_i$,  we know $Q^{\mathrm{right}}_{r'}(v) \geq j_i$. 
Since $\mathcal{F}_{i+1}$ is obtained by performing the right split at $j_i$, 
every index of $H_{i+1}$ precedes $j_i$, implying that such a $v$ cannot be contained in 
$\mathcal{L}(A,H_{i+1})$ in $\mathcal{F}_{i+1}$. 
Thus, $L^{\mathrm{max}}(A,H_{i+1}) \neq L^{\mathrm{max}}(A,H_i)$, from which $y^L_{i+1} < y^L_{i}$ follows. 
Therefore, we have $d^{i+1}(A) > d^{i}(A)$.

Second, let us assume that $j_i$ is R-critical for $(A,H_i)$. By the definition of R-criticality, 
for any  vertex $v$ starting at $R^{\mathrm{min}}(A,H_i)$ 
and contained in $\mathcal{R}(A,H_i)$ in $\mathcal{F}_i$, 
we know  $Q^{\mathrm{right}}_{r'}(v) \geq j_i$. Again, we know that every index of $H_{i+1}$ precedes $j_i$. 
From this, we have that $v$ cannot be contained in $\mathcal{R}(A,H_{i+1})$ in $\mathcal{F}_{i+1}$, 
implying $y^R_{i+1} > y^R_i$. Therefore, we have $d^{i+1}(A) > d^{i}(A)$ in this case as well.

Now, we claim that $1 \leq d^{i}(A) \leq \sigma(A)$ for any $A \in \mathcal{P}_{\mathcal{R},i}$.
First, it is clear that for any $\ell<8$, Property $\ell$ holds for each index both in $\mathcal{F}_i$ 
and in the reversed fragmentation $\mathcal{F}_i^{\mathrm{rev}}$,
as otherwise the algorithm would branch on an index violating Property $\ell$.  
Thus, $L^{\mathrm{max}}(A,H_i) \geq R^{\mathrm{min}}(A,H_i)$ cannot happen, 
as this would mean that there is a conflict-inducing index in $H_i$ for $(A,H_i)$, violating Property 7.
This directly implies $1 \leq d^{i}(A)$.

On the other hand, assume $d^{i}(A)= y^R_i - y^L_i > \sigma(A)$. 
This implies that $h=y^R_i - \sigma(A)$ is contained in $A$, but no vertx of 
$\mathcal{L}(A,H_i) \cup \mathcal{R}(A,H_i)$ starts in $h$.
However, by Properties 3 and 5  for $y^R_i$, we know that some vertex in 
$M^+_r(\leftind{(y^R_i)})=M^+_r(\rightind{h})$ ends in $H_i$. 
Using these properties for $h^{\mathrm{rev}}$ in the reversed instance, 
we obtain that some arc $v$ in $M^+_{r'}(h)$ must also end in $H_i$.
By Property 5 for $h$, $v$ must be contained in one of the sets 
$\mathcal{L}(A,H_i)$, $\mathcal{R}(A,H_i)$, $\mathcal{X}(A,H_i)$. 
But $y^L_i<h<y^R_i$, so we obtain $v \in \mathcal{X}(A,H_i)$. 
Therefore, some position in $H_i$ violates Property 6, a contradiction. 
This proves $1 \leq d^{i}(A) \leq \sigma(A)$.

Now, observe that for any $i \in [t]$, $j_i$ is LR-critical for some $(A,H_i)$ with $A \in \mathcal{P}_{\mathcal{R},i}$. If 
$A \in  \mathcal{P}_{\mathcal{R},i+1}$ as well, then $d^{i+1}(A) > d^i(A)$. By our bounds on $d^i(A)$, this yields that 
there can be at most $\sigma(A)$ indices $i$ where $j_i$ is LR-critical for $(A,H_i)$. 
(Here we also used that $d^{i}(A)$ cannot decrease.)
This clearly implies $t \leq \sum_{F \in \mathcal{P}} \sigma(F) = \delta(H_1)$. 

To finish the proof, we show $\delta(H_1) \leq k$. 
Let $b'$ be the last index preceding the indices in $H_1$, and let $b=b'+\delta(H_1)$. 
Recall that $\phi_S (B^+_{r'}(1,b')) = B^+_r(1,b) \setminus S$ by Prop.~\ref{limit_indices}.
Using that Properties 1 and 3 hold for every index in $[m']$ and that 
$B^+_r(i) \neq \emptyset$ by Prop.~\ref{notempty} for any $i \in [m]$, we obtain
$$
|B^+_{r'}(1,b')| + k \geq
|B^+_r(1,b)| =
\bigg| \bigcup_{1 \leq j \leq b'} B^+_r(\leftind{j}) \bigg| + 
\bigg| \bigcup_{
\substack{
1 \leq i \leq b, \\
i \notin \{\leftind{j} : 1 \leq j \leq b'\}
}} B^+_r(i) \bigg|
$$
$$
\geq
|B^+_{r'}(1,b')| + \sum_{F \in \mathcal{P}} \sigma(F)
=
|B^+_{r'}(1,b')| +\delta(H_1).
$$
This shows $k \geq \delta(H_1)$, proving the lemma. 
\qed
\end{proof}

Now, we can state the key properties of algorithm $\mathcal{A}$, which prove Theorem~\ref{FPT}.

\begin{mylemma} 
\label{algorithm}
Given an input $(G',G)$ where $|V(G')|=n$ and $|V(G)|=n+k$, 
algorithm $\mathcal{A}$ either produces a reduced input in $O(n)$ time, 
or branches into at most $f(k)$ directions for some function $f$
such that in each branch it either correctly refuses the instance, 
or outputs an independent subproblem or a necessary set of size at most $2k+1$.
Moreover, each branch takes $O(n)$ time.
\end{mylemma}

\begin{proof}
Let us overview the steps of algorithm $\mathcal{A}$.
First, it  tries to apply the reduction rules described in Sect.~\ref{reductions}.
In this phase, it either outputs a reduced input in linear time, 
or it may branch into at most $(4k+1) 2^{4k} (k(7k/2+8)+1)= 2^{O(k)}$ branches.
In each branch it either correctly outputs 'No', outputs a necessary set of size at most 2, 
or outputs an independent subproblem having parameter at most $k-1$ but at least $1$.
These steps can be done in linear time, as argued in Sect.~\ref{reductions}.

If none of the reductions in Sect.~\ref{reductions} can be applied, then $\mathcal{A}$ 
first checks whether a reduced input can be output by using Lemma~\ref{qq_subtree}.
If not, then it branches into 3 directions, according to whether $S$ is local, and if not, 
whether the children of $r'$ should be reversed to achieve the properties of Lemma~\ref{I_S_prop}.
In the first branch, it outputs a necessary set of size at most 2.
In the other two branches, it checks whether the annotated fragmentation $AF_0$
produced in the beginning is proper. 
While the annotated fragmentation is not proper, $\mathcal{A}$ chooses the smallest $\ell$ 
and the smallest index $j$ violating Property $\ell$ (maybe in the reversed instance), and 
branches into at most $2(k+1)^2-1$ directions.
In these branches, $\mathcal{A}$ either modifies the actual annotated fragmentation or stops by outputting 
an independent subproblem, a necessary set of size at most $2k+1$, or rejecting. 

Let us consider a sequence of $t$ such branchings performed by $\mathcal{A}$, and let $AF_0, AF_1, \dots, AF_t$ be 
the sequence of annotated fragmentations produced in this process. 
(We interpret these as annotated fragmentations for $(T,T',S)$ and not for $(T^{\mathrm{rev}},T'^{\mathrm{rev}},S)$.)
Let us call a continuous subsequence $\mathcal{S}$ of $AF_0, AF_1, \dots, AF_t$ a \emph{segment}, 
if each annotated fragmentation in $\mathcal{S}$ has the same number of non-trivial fragments, and $\mathcal{S}$ is maximal
with respect to this property. By Lemma~\ref{max_fragments}, the algorithm can reject if there are more than $2k$ 
non-trivial fragments in a fragmentation, so $AF_0, AF_1, \dots, AF_{t}$ can contain at most $2k$ segments.
Let $\mathcal{S} = AF_{t_1},AF_{t_1+1}, \dots, AF_{t_2}$ be such a segment. 
Clearly, each $AF_{h}$ ($t_1<h \leq t_2$) is obtained from $AF_{h-1}$ by performing a right split 
either in the original or in the reversed instance 
(the latter meaning that $AF_{h}^{\mathrm{rev}}$ is obtained from $AF_{h-1}^{\mathrm{rev}}$ by a right split).

Let $AF_p$ be the first 9-proper annotated fragmentation in the segment.
Using Lemma~\ref{measure}, each subsequence of $AF_{t_1}, \dots, AF_p$ 
where the measure does not increase can have length at most $k$. 
(The measure of an annotated fragmentation is the measure of its fragmentation.)
By (2) of Lemma~\ref{measure}, $AF_{t_1}$ has a non-trivial fragment containing each of those indices for which the algorithm 
performed a branching (because of Property 8) in some $AF_h$, $t_1 \leq h \leq p$. Taking into account 
that the number of non-trivial fragments can not exceed $2k$, but branchings can also happen in the reversed instance, 
we obtain that there can be at most $4k$ maximal subsequences in $AF_{t_1}, \dots, AF_p$ of length at least 2 
where the measure is constant. 
Using Lemma~\ref{max_fragments}, we get that $\mu(AF_p) \leq 36k$, implying $p \leq t_1+4k^2+36k$.

Clearly, $\mathcal{A}$ obtains $AF_{p+1},AF_{p+2}, \dots, AF_{t_2}$ 
while trying to ensure Property 10, by performing right splits in the original instance. 
Observe that if $\mathcal{A}$ obtains $AF_{h}$ ($p+1 \leq h \leq t_2$)
by applying a right split at $j$, then by the choice of $j$, Property 10 holds for each index $j' \leq j$ in
any $AF_{h'}$ where $h' \geq h$. This, together with Lemma~\ref{max_fragments} implies that 
$\mathcal{A}$ can perform at most $2k$ such branchings, implying that $t_2 \leq p+2k \leq t_1+4k^2 + 38k$. 
Altogether, this implies $t \leq 2k (4k^2+38k)$, proving that the maximum length of
a sequence of branchings performed by $\mathcal{A}$ in order to obtain a proper
annotated fragmentation can be at most $8k^3+76k^2=O(k^3)$.

Essentially, this means that the search tree that $\mathcal{A}$ investigates has height at most $O(k^3)$. 
Since one branching results in at most $2(k+1)^2-1$ directions, 
we obtain that the total number of resulting branches in a run of algorithm $\mathcal{A}$ 
can be bounded by a function $f$ of $k$. In each of these branches, if $\mathcal{A}$ does not stop, 
then it has a proper annotated fragmentation $(\mathcal{F},U)$. 
After this, algorithm $\mathcal{A}$ does not perform any more branchings. 
Instead, it applies Lemma~\ref{right_aligned} repeatedly. If the algorithm reaches a state where
Lemma~\ref{right_aligned} does not apply, then it outputs a solution in linear time
using Lemma~\ref{isomorphism}. 

It is easy to verify that each branch can be performed in linear time. 
The only non-trivial task is to show that the repeated application of Lemma~\ref{right_aligned} 
can be implemented in linear time, but this easily follows from the fact that 
none of the conditions of Lemma~\ref{right_aligned} can be applied twice for a block $[a,b]$.
\qed
\end{proof}

\subsection{The proof of Lemma~\ref{left_aligned}}
\label{sect_left_aligned}

In this section we prove Lemma~\ref{left_aligned}.
Suppose that Property $\ell$ ($1 \leq \ell \leq  10$) does not hold for some $j \in [m']$
in the annotated fragmentation $(\mathcal{F},U)$, but
all the previous properties hold for each index 
both in $(\mathcal{F},U)$ and in $(\mathcal{F}^{\mathrm{rev}},U^{\mathrm{rev}})$.
Suppose also that $j$ is left-aligned, i.e. $I_S(j)=[\leftind{j},\leftind{j}]$.
Below we describe the detailed steps of algorithm $\mathcal{A}$ depending on the property
that is violated by $j$. 

\begin{quote} 
{\bf Property 1:}
$G'[X'_j]$ is isomorphic to $G[X_{\leftind{j}}]$.
\end{quote}
If $j$ violates Property 1, then 
$G'[X'_j]$ is not isomorphic to $G[X_{\leftind{j}}]$, which implies $S \cap X_{\leftind{j}} \neq \emptyset$.
From $I_S(j)=[\leftind{j},\leftind{j}]$  we obtain that $S \cap X_{\leftind{j}}$ must be a solution
for $(G'[X'_j],G[X_{\leftind{j}}])$. Conversely, if $(G',G)$ is solvable, then
any solution for $(G'[X'_j],G[X_{\leftind{j}}])$ can be extended to a solution for $(G',G)$.
By $m>m'$, $G-X_{\leftind{j}}$ can not be isomorphic to $G'-X'_j$, so 
$S \subseteq X_{\leftind{j}}$  is not possible. 
Therefore, if the parameter of $(G'[X'_j],G[X_{\leftind{j}}])$ is more than $k-1$ (or less than $1$), 
then the algorithm can refuse the instance.
Thus, $\mathcal{A}$ can either \outputstyle{reject}, or it can output the 
\outputstyle{independent subproblem} $(G'[X'_j],G[X_{\leftind{j}}])$.

\begin{quote}
{\bf Property 2:}
$|M_{r'}^+(j)| \leq |M_{r}^+(\leftind{j})| \leq |M_{r'}^+(j)|+k$ and 
$|M_{r'}^-(j)| \leq |M_{r}^-(\leftind{j})| \leq |M_{r'}^-(j)|+k$.
\end{quote}
By $I_S(j)=[\leftind{j},\leftind{j}]$ and Prop.~\ref{limit_indices}, we can observe that 
$M_{r}^+(\leftind{j}) \setminus S = \phi_S(M_{r'}^+(j))$ and 
$M_{r}^-(\leftind{j}) \setminus S = \phi_S(M_{r'}^-(j))$. 
If $j$ violates Property 2, then this contradicts to $|S|\leq k$, 
and thus algorithm $\mathcal{A}$ can \outputstyle{reject}. 

\begin{mylemma}
\label{step2}
If Properties 1 and 2 hold for each index both in $(\mathcal{F},U)$ and in $(\mathcal{F}^{\mathrm{rev}},U^{\mathrm{rev}})$, 
and there is an index $h \in [m']$ contained in a non-trivial fragment $F$ such that 
$|M_{r'}^+(h)|>k$ or $|M_{r'}^-(h)|>k$, then there is no solution for $(G',G)$.
\end{mylemma}

\begin{proof}
As Property 2 holds for each index in $F=([a',b'],[a,b])$, 
$|M^+_{r'}(j)| \leq |M^+_r(\leftind{j})|$ holds for each $j \in [m']$.
Similarly, as Property 2 holds for each index in the reversed instance, we obtain that 
$|M^+_{r'}(j)| \leq |M^+_r(\rightind{j})|$ 
must hold for each $j \in [m']$. Supposing $|M_{r'}^+(h)|>k$, we get 
$$
\sum_{a \leq i \leq b} |M^+_r(i)|  = 
\sum_{a' \leq j < h} |M^+_r(\leftind{j})| + 
\sum_{0 \leq d < \sigma(F)} |M^+_r(\leftind{h}+d)| 
$$ 
$$
+
\sum_{h \leq j \leq b'} |M^+_r(\rightind{j})| \geq 
|M^+_{r'}(h)| + \sum_{a' \leq j \leq b'} |M^+_{r'}(j)| 
 >   k + \sum_{a' \leq j \leq b'} |M^+_{r'}(j)|.
$$
Observe that we used  $\sigma(F)>0$ in the first inequality.

Proposition~\ref{limit_indices} yields $\phi_S (B^+_{r'}(a',b')) = B^+_r(a,b) \setminus S$, from which 
$|B^+_r(a,b)| \leq |B^+_{r'}(a',b'))|+k$ follows. Using that Property 1 holds for each index, 
we also have $|X'_j|=|X_{\leftind{j}}|$ for each $j\in [m']$, 
implying $\sum_{a \leq i \leq b} |X_i| \geq \sum_{a' \leq j \leq b'} |X'_j|$. 
Hence, we obtain
$$\sum_{a \leq i \leq b} |M^+_r(i)| \leq \sum_{a' \leq j \leq b'} |M^+_{r'}(j)| + k,
$$ 
contradicting the above inequality. The case $|M_{r'}^-(h)|>k$ can be handled in the same way.
\qed
\end{proof}

\begin{quote}
{\bf Property 3:}
If $j$ is non-trivial, then 
$|M_{r'}^+(j)| = |M_{r}^+(\leftind{j})|$ and 
$|M_{r'}^-(j)| = |M_{r}^-(\leftind{j})|$.
\end{quote}
By $I_S(j)=[\leftind{j},\leftind{j}]$ and Prop.~\ref{limit_indices}, we get 
$M_{r}^+(\leftind{j}) \setminus S = \phi_S(M_{r'}^+(j))$ and 
$M_{r}^-(\leftind{j}) \setminus S = \phi_S(M_{r'}^-(j))$.
Clearly, if $|M_{r}^+(\leftind{j})|<|M_{r'}^+(j)|$ or $|M_{r}^-(\leftind{j})|<|M_{r'}^-(j)|$, 
then algorithm $\mathcal{A}$ can output 'No'.
If this is not the case, then $S$ must contain at least one vertex from 
$M_{r}^+(\leftind{j})$ or $M_{r}^-(\leftind{j})$, because $j$ violates Property 3.
If $|M_{r'}^+(j)|> k$ or $|M_{r'}^-(j)| > k$,  
then $\mathcal{A}$ can output 'No' as well, by Lemma~\ref{step2}.
Thus, if $\mathcal{A}$ does not \outputstyle{reject}, 
then it can output a \outputstyle{necessary set} of size at most $k+1$ in both cases, 
by taking $|M_{r'}^+(j)|+1$ or $|M_{r'}^-(j)|+1$ arbitrary vertices from 
$M_{r}^+(\leftind{j})$ or $M_{r}^-(\leftind{j})$, respectively.

\begin{quote}
{\bf Property 4:}
If $j$ is non-trivial, then 
$|L_{r'}(y,j)|=|L_{r}(\leftind{y},\leftind{j})|$ for any $y < j$ 
contained in the same fragment as $j$.
\end{quote}
Suppose that $|L_{r'}(y,j)| \neq |L_{r}(\leftind{y},\leftind{j})|$ for some $y < j$ 
contained in the same fragment that contains $j$.
Since $j$ is left-aligned, we get that $y$ must also be left-aligned as well by $y < j$,
i.e. $I_S(y)=[\leftind{y},\leftind{y}]$.
By Prop.~\ref{limit_indices}, this implies $L_r(\leftind{y},\leftind{j}) \setminus S = \phi_S(L_{r'}(y,j))$.
Thus, if $|L_{r'}(y,j)| > |L_{r}(\leftind{y},\leftind{j})|$ then
$\mathcal{A}$ can \outputstyle{reject}. Otherwise,
$L_{r}(\leftind{y},\leftind{j})$ contains at least one vertex from $S$.
Since each vertex in $L_{r}(\leftind{y},\leftind{j})$ has the same neighborhood, 
$\mathcal{A}$ can output $\{s\}$ as a \outputstyle{necessary set} for some arbitrarily chosen $s$ 
in $L_{r}(\leftind{y},\leftind{j})$.

\begin{quote}
{\bf Property 5:}
If $j$ is non-trivial, then for every $(v,w) \in \leftind{P}^+(j)$ such that
$Q_{r'}^{\mathrm{right}}(v)$ $= y$ is non-trivial, $\leftind{y} \leq Q_r^{\mathrm{right}}(w) \leq \rightind{y}$ holds.
Also, for every $(v,w) \in \leftind{P}^-(j)$ such that  
$Q_{r'}^{\mathrm{left}}(v)=y$ is non-trivial, $\leftind{y} \leq Q_r^{\mathrm{left}}(w) \leq \rightind{y}$ holds.
\end{quote}
Suppose that $j$ violates Property $5$, because 
$(v,w) \in \leftind{P}^+(j)$  such that $Q_{r'}^{\mathrm{right}}(v)= y$ is non-trivial, 
but $\leftind{y} \leq Q_r^{\mathrm{right}}(w) \leq \rightind{y}$ does not hold.
We show that $\mathcal{A}$ can output 'No' in this case.
As Property~3 holds for $j$, $|M_{r'}^+(j)| = |M_{r}^+(\leftind{j})|$.
As $j$ is left-aligned, $\phi_S(v)=w$ by Prop.~\ref{pairs}.
But from this, Prop.~\ref{limit_indices} implies $\alpha_S(y) \leq  Q_r^{\mathrm{right}}(w) \leq \beta_S(y)$. 
By Prop.~\ref{limit_indices} we know
$\leftind{y} \leq \alpha_S(y) \leq \beta_S(y) \leq \rightind{y}$ as well.
Therefore, $\mathcal{A}$ can indeed refuse the instance.
Supposing that Property 5 does not hold because of the case 
where some $(v,w) \in \leftind{P}^-(j)$ is considered leads to the same result,
so it is straightforward to verify that $\mathcal{A}$ can \outputstyle{reject} in both cases.

The observation below, used in the forthcoming three cases, is easy to see:
\begin{myprop}
\label{LR_labels} Suppose that the first five properties  hold for a given (annotated) fragmentation.
Let $y$ and $j$ be indices of $[m']$ contained in non-trivial fragments $F$ and $H$, respectively, 
and suppose that $j$ is left-aligned. Then $v \in L_{r'}(y,j)$ implies the followings. \\
(1) $v \in \mathcal{L}(F,H) \cup \mathcal{R}(F,H) \cup \mathcal{X}(F,H)$. \\
(2) If $v \in \mathcal{L}(F,H)$, then $\alpha_S(y)=\leftind{y}$. \\
(3) If $v \in \mathcal{R}(F,H)$, then $\beta_S(y)=\rightind{y}$. \\
(4) If $v \in \mathcal{X}(F,H)$, then $y$ is either wide or skew. 
\end{myprop}

\begin{quote}
{\bf Property 6:}
If $j$ is non-trivial, then no vertex in $\mathcal{X}(F,H)$ (for some $F$ and $H$) ends in $j$.
\end{quote}
Suppose that Property 6 does not hold for $j$, so there is a vertex in $L_{r'}(y,j) \cap \mathcal{X}(F,H)$ for some $y<j$.
As $j$ is left-aligned, Prop.~\ref{LR_labels} implies that $y$ is either \outputstyle{wide or skew}.

\begin{quote}
{\bf Property 7:}
$j$ is not conflict-inducing for any $(F,H)$.
\end{quote}
Suppose that $j$ violates Property 7 because it is conflict-inducing for some $(F,H)$ 
and for some conflicting pair of indices $(y_1,y_2)$. 
Let $j_1$ be the minimal index for which $L_{r'}(y_1,j_1) \cap \mathcal{R}(F,H) \neq \emptyset$, 
and let $j_2$ be the minimal index for which  $L_{r'}(y_2,j_2) \cap \mathcal{L}(F,H) \neq \emptyset$. 
Since $j \geq \max\{j_1,j_2\}$, and $j$ is left-aligned, we know that both $j_1$ and $j_2$ are left-aligned as well. 
By Prop.~\ref{LR_labels}, this implies 
$\beta_S(y_1) = \rightind{y_1}$ and $\alpha_S(y_2) = \leftind{y_2}$. 
If $y_1 < y_2$, then this yields a contradiction by Prop. \ref{limit_indices}, 
so $\mathcal{A}$ can \outputstyle{reject}.
In the case where $y_1=y_2=y$, we get $I_S(y)=[\leftind{y},\rightind{y}]$, and since $y$ is non-trivial, 
algorithm $\mathcal{A}$ can output $y$ as a \outputstyle{wide index}. 

For the case of Property 8, we need the following simple lemma:

\begin{mylemma}
\label{short_gap}
Suppose that a fragmentation for $(T,T',S)$ contains a fragment $F=([a',b'],[a,b])$ with $0 < b'-a' \leq \sigma(F)$, 
and the first four properties hold for each index contained in $F$ both in the given fragmentation and its reversed version.
Then $\mathcal{A}$ can produce a necessary set of size at most $2k+1$.
\end{mylemma}

\begin{proof}
Since Properties 1 and 3 hold for each index contained in $F$, we obtain  
$|B^+_{r'}(a',b')| = |B^+_{r}(\leftind{a'},\leftind{b'})|$. 
Using Prop.~\ref{limit_indices} we have $B^+_r(a,b) \setminus S = \phi_S(B^+_{r'}(a',b'))$.
Proposition~\ref{notempty} yields $B^+_{r'}(j) \neq \emptyset$ for any $j$, so 
we get $|B^+_{r}(a,b)| > |B^+_{r'}(a',b')|$. 
Hence, fixing an arbitrary set $N \subseteq B^+_r(a,b)$ of size $|B^+_{r'}(a',b')|+1$, 
we get that $N$ is a nonempty necessary set. We claim $|B^+_{r'}(a',b')| \leq 2k$, which implies $|N|\leq 2k+1$. 
Thus, $\mathcal{A}$ can indeed output $N$, proving the lemma.

It remains to show $|B^+_{r'}(a',b')| \leq 2k$.
Recall $|B^+_{r'}(a',b')| = |B^+_{r}(\leftind{a'},\leftind{b'})|$.
As Properties 1 and 3 hold for each index contained in $F^{\mathrm{rev}}$ in the reversed fragmentation,
we get $|B^+_{r'}(a',b')| = |B^+_{r}(\rightind{a'},\rightind{b'})|$ as well.
Using $\rightind{a'} - \leftind{b'} = a'-b' +\sigma(F) \geq 0$, we obtain that 
$B^+_r(\leftind{a'},\leftind{b'}) \cap B^+_{r}(\rightind{a'},\rightind{b'}) \subseteq B^+_r(\leftind{b'})$.
Moreover, if $b'-a' <\sigma(F)$ also holds, then actually $B^+_r(\leftind{a'},\leftind{b'}) 
\cap B^+_{r}(\rightind{a'},\rightind{b'}) = \emptyset$.

By the above paragraph,  $b'-a' < \sigma(F)$ implies $|B^+_{r}(a,b)| \geq 2|B^+_{r'}(a',b')|$, 
so we get $|B^+_{r'}(a',b')| \leq k$.
On the other hand, $b'-a'=\sigma(F)$ yields
$|B^+_{r'}(a',b')|+k \geq |B^+_{r}(a,b)| \geq 2|B^+_{r'}(a',b')|-|B^+_r(\leftind{b'})|$, implying 
$|B^+_{r'}(a',b')| \leq k + |B^+_r(\leftind{b'})|$. 
Taking into account that $|B^+_{r'}(a')|=|B^+_r(\leftind{b'})|=|B^+_{r'}(b')|$
holds by Properties 1 and 3 for $b'$ and for $a'^{\mathrm{rev}}$, we have 
$|B^+_{r'}(a',b')| \geq 2|B^+_r(\leftind{b'})|$. 
Summarizing all these, $|B^+_{r'}(a',b')| \leq 2k$ follows. 
\qed
\end{proof}

\begin{quote}
{\bf Property 8:}
$j$ is not LR-critical for any $(F,H)$.
\end{quote}
Suppose that $j$ violates Property 8, so $j$ is LR-critical for some $(F,H)$. 
In this case, $R^{\mathrm{min}}(F,H)=y^R$ is an index contained in $F$. Since $j$ is left-aligned, 
the R-critical index for $(F,H)$ is also left-aligned, hence Prop.~\ref{LR_labels} yields $\beta_S(y^R)=\rightind{y^R}$.
Let $a'$ be the first index of $[m']$ contained in $F$.

First, if $y^R < a'+\sigma(F)$, then we apply Lemma~\ref{short_gap} as follows. 
Clearly, by  $\beta_S(y^R)=\rightind{y^R}$ we can perform a right split at $y^R$. 
The obtained fragmentation will contain the fragment $F'=([a',y^R],[\leftind{a'},\rightind{y^R}])$, 
so $y^R-a' < \sigma(F) = \sigma(F')$ shows that $\mathcal{A}$ can produce a \outputstyle{necessary set} 
of size at most $2k+1$ by using Lemma~\ref{short_gap}.

Now, suppose  $y^R \geq a'+\sigma(F)$. In this case, there is an index $t$ in $F$ for which $\rightind{t}=\leftind{y^R}$. 
By Properties 3 and 5 for $y^R$, we know that there is a vertex in 
$M^+_r(\leftind{y^R})$ that ends in the fragment $H$. 
Using  Properties 3 and 5 again for $t^{\mathrm{rev}}$ in the reversed instance, 
we know that there must be a vertex $v$ in $M^+_{r'}(t)$ that ends in the fragment $H$. 
By Prop.~\ref{LR_labels}, $v \in \mathcal{L}(F,H) \cup \mathcal{R}(F,H) \cup \mathcal{X}(F,H)$. 
Observe that $v \notin \mathcal{X}(F,H)$, as Property 6 holds for every index in $[m']$. 
Also, $v \notin \mathcal{R}(F,H)$ by the definition of $y^R=R^{\mathrm{min}}(F,H)$. 
Thus, we know that $v \in \mathcal{L}(F,H)$, implying $y^L=L^{\mathrm{max}}(F,H) \geq t$ as well.
As Property 7 holds for each index, we also have $y^L<y^R$.

To finish the case, observe that since $j$ is left-aligned and LR-critical for $(F,H)$,
Prop.~\ref{LR_labels} yields $\alpha_S(y^L)=\leftind{y^L}$. 
Using $\beta_S(y^R)=\rightind{y^R}$ again, we can produce a fragmentation for $(T,T',S)$ that contains 
the fragment $F'=([y^L,y^R],[\leftind{y^L},\rightind{y^R}])$. 
(This can be thought of as performing a right split at $y^R$, 
and a right split at $(y^L)^{\mathrm{rev}}$ in the reversed instance.)
Hence, $y^R-y^L \leq y^R-t = \sigma(F) = \sigma(F')$ shows that 
$\mathcal{A}$ can produce a \outputstyle{necessary set} of size at most $2k+1$ by using Lemma~\ref{short_gap}.

\begin{quote}
{\bf Property 9:} 
If $j$ is non-trivial, then for every $(v,w) \in \leftind{P}^+(j)$ such that
$Q_{r'}^{\mathrm{right}}(v)$ $= y$ is non-trivial, $Q_r^{\mathrm{right}}(w)= \leftind{y}$ holds.
Also, for every $(v,w) \in \leftind{P}^-(j)$ such that  
$Q_{r'}^{\mathrm{left}}(v)=y$ is non-trivial, $Q_r^{\mathrm{left}}(w)= \leftind{y}$ holds.
\end{quote}
Observe that if Property 9 does not hold for an index $j$, then by Prop.~\ref{LR_labels}, 
either $M_{r'}^+(j)$ or $M_{r'}^-(j)$ contains a vertex in $\mathcal{R}(F,H) \cup \mathcal{X}(F,H)$ for some $(F,H)$. 
But this means that one of Properties 6 and 8 must be violated, 
which is a contradiction. Thus, $\mathcal{A}$ can correctly \outputstyle{reject}.

\begin{quote}
{\bf Property 10:}
If $j$ is non-trivial, then for each important trivial index $u \in U$, 
$|L_{r'}(j,u)|=|L_{r}(\leftind{j},\leftind{u})|$ holds if $u>j$, and 
$|L_{r'}(u,j)|=|L_{r}(\leftind{u},\leftind{j})|$ holds if $u<j$.
\end{quote}
Suppose that $j$ violates Property 10, because 
$|L_{r'}(j,u)| \neq |L_{r}(\leftind{j},\leftind{u})|$ for some $u>j$. 
(The case when $u<j$ can be handled in the same way.)
Since $u$ is contained in a trivial fragment, 
$I_S(u)=[\leftind{u},\leftind{u}]$.
Thus, by $I_S(j)=[\leftind{j},\leftind{j}]$ and Prop.~\ref{limit_indices}, we get
$L_{r}(\leftind{j},\leftind{u}) \setminus S =  \phi_S(L_{r'}(j,u))$.
If $|L_{r'}(j,u)| > |L_{r}(\leftind{j},\leftind{u})|$, 
then $\mathcal{A}$ can \outputstyle{reject} the instance.
Otherwise, we can argue as before that $\{s\}$ is a \outputstyle{necessary set} for any 
$s \in L_{r}(\leftind{j},\leftind{u})$.

\bibliography{ref_interval}

\begin{thebibliography}{10}

\bibitem{booth-lueker-consecutive}
K.~S. Booth and G.~S. Lueker.
\newblock Testing for the consecutive ones property, interval graphs, and
  planarity using pq-tree algorithms.
\newblock {\em J. Comput. Syst. Sci.}, 13:335--379, 1976.

\bibitem{cai-random-separation}
L.~Cai, S.~M. Chan, and S.~O. Chan.
\newblock Random separation: A new method for solving fixed-cardinality
  optimization problems.
\newblock In {\em IWPEC 2006: Proceedings of the 2nd International Workshop on
  Parameterized and Exact Computation}, volume 4169 of {\em Lecture Notes in
  Computer Science}, pages 239--250. Springer, 2006.

\bibitem{booth-colbourn}
C.~J. Colbourn and K.~S. Booth.
\newblock Linear time automorphism algorithms for trees, interval graphs, and
  planar graphs.
\newblock {\em SIAM J. Comput.}, 10(1):203--225, 1981.

\bibitem{diaz-thilikos-grid}
J.~D\'{\i}az and D.~M. Thilikos.
\newblock Fast fpt-algorithms for cleaning grids.
\newblock In {\em STACS 2006: Proceedings of the 23rd Annual Symposium on
  Theoretical Aspects of Computer Science}, volume 3884 of {\em Lecture Notes
  in Computer Science}, pages 361--371. Springer, 2006.

\bibitem{dinitz-itai-rodeh-subtree-iso}
Y.~Dinitz, A.~Itai, and M.~Rodeh.
\newblock On an algorithm of {Z}emlyachenko for subtree isomorphism.
\newblock {\em Inf. Process. Lett.}, 70(3):141--146, 1999.

\bibitem{downey-fellows-book}
R.~G. Downey and M.~R. Fellows.
\newblock {\em Parameterized complexity}.
\newblock Monographs in Computer Science. Springer-Verlag, New York, 1999.

\bibitem{eppstein-subgraphs}
D.~Eppstein.
\newblock Subgraph isomorphism in planar graphs and related problems.
\newblock {\em J. Graph Algorithms Appl.}, 3(3):1--27, 1999.

\bibitem{grohe-flum-book}
J.~Flum and M.~Grohe.
\newblock {\em Parameterized complexity theory}.
\newblock Texts in Theoretical Computer Science. An EATCS Series.
  Springer-Verlag, New York, 2006.

\bibitem{garey-johnson-book}
M.~R. Garey and D.~S. Johnson.
\newblock {\em Computers and Intractability: A Guide to the Theory of
  {NP}-Completeness}.
\newblock W. H. Freeman \& Co., New York, 1979.
\newblock A Series of Books in the Mathematical Sciences.

\bibitem{gilmore-hoffman}
P.~C. Gilmore and A.~J. Hoffman.
\newblock A characterization of comparability graphs and of interval graphs.
\newblock {\em Canad. J. Math.}, 16:539--548, 1964.

\bibitem{hajiaghayi-nishimura-logbounded-fragmentation}
M.~Hajiaghayi and N.~Nishimura.
\newblock Subgraph isomorphism, log-bounded fragmentation, and graphs of
  (locally) bounded treewidth.
\newblock {\em J. Comput. Syst. Sci.}, 73(5):755--768, 2007.

\bibitem{lingas-outerplanar}
A.~Lingas.
\newblock Subgraph isomorphism for biconnected outerplanar graphs in cubic
  time.
\newblock {\em Theor. Comput. Sci.}, 63(3):295--302, 1989.

\bibitem{booth-lueker-interval-isomorphism}
G.~S. Lueker and K.~S. Booth.
\newblock A linear time algorithm for deciding interval graph isomorphism.
\newblock {\em J. ACM}, 26(2):183--195, 1979.

\bibitem{marx-schlotter-dam-cleaning}
D.~Marx and I.~Schlotter.
\newblock Parameterized graph cleaning problems.
\newblock {\em Discrete Applied Mathematics}, 157(15):3258--3267, 2009.

\bibitem{matula-subtree}
D.~W. Matula.
\newblock Subtree isomorphism in $o(n^{5/2})$.
\newblock {\em Ann. Discrete Math.}, 2:91--106, 1978.

\bibitem{niedermeier-book}
R.~Niedermeier.
\newblock {\em Invitation to fixed-parameter algorithms}, volume~31 of {\em
  Oxford Lecture Series in Mathematics and its Applications}.
\newblock Oxford University Press, Oxford, 2006.

\bibitem{zemlyachenko-canonical}
V.~N. Zemlyachenko.
\newblock Canonical numbering of trees, 1970.
\newblock (In Russian).

\bibitem{zemlyachenko-tree-iso}
V.~N. Zemlyachenko.
\newblock Determining tree isomorphism.
\newblock In {\em Voprosy Kibernetiki, Proc. of the Seminar on Combinatorial
  Mathematics, Moscow, 1971}, pages 54--60. Akad. Nauk SSSR, Scientific Council
  on the Complex Problem "Cybernetics", 1973.
\newblock (In Russian).

\end{thebibliography}
\bibliographystyle{abbrv}

\end{document}